\documentclass[nonacm,acmsmall]{acmart}

\acmJournal{PACMPL}
\acmVolume{1}
\acmNumber{POPL} 
\acmArticle{1}
\acmYear{2022}
\acmMonth{1}
\acmDOI{} 
\startPage{1}

\setcopyright{none}

\usepackage{bm}
\usepackage[caption=false,font=footnotesize]{subfig}
\usepackage{url}
\usepackage{mathtools}
\usepackage[utf8]{inputenc}
\usepackage[T1]{fontenc}
\usepackage{hyperref}
\usepackage{bussproofs}
\usepackage{tikz}
\usepackage{tikz-cd}
\usepackage{stmaryrd}
\usepackage{esvect}

\DeclarePairedDelimiter\bra{\langle}{\rvert}
\DeclarePairedDelimiter\ket{\lvert}{\rangle}
\DeclarePairedDelimiterX\braket[2]{\langle}{\rangle}{#1 \delimsize\vert #2}

\newenvironment{bprooftree}
  {\leavevmode\hbox\bgroup}
  {\DisplayProof\egroup}

\newcommand{\secref}[1]{\S\ref{#1}}

\newcommand{\Ob}{\text{Ob}}%

\newcommand{\id}{\mathrm{id}}
\newcommand{\Id}{\mathrm{Id}}

\newcommand{\op}{\ensuremath{\mathrm{op}}}

\newcommand{\up}{\ensuremath{{\uparrow\,}}}
\newcommand{\down}{\ensuremath{\mathop{\downarrow}}}

\newcommand{\Na}{\ensuremath{\mathbb N}}

\newcommand{\VQPL}{%
  \ensuremath{\mathtt{VQPL}}
}
\newcommand{\QPL}{%
  \ensuremath{\mathtt{QPL}}
}
\newcommand{\PFPC}{\ensuremath{\mathtt{PFPC}}}

\newcommand{\Paths}{\mathrm{Paths}}
\newcommand{\TPaths}{\mathrm{TPaths}}

\newcommand{\AAA}{\ensuremath{\mathbf{A}}}

\newcommand{\BBB}{\ensuremath{\mathbf{B}}}
\newcommand{\CC}{\ensuremath{\mathbf{C}}}

\newcommand{\DCPO}{\ensuremath{\mathbf{DCPO}}}
\newcommand{\dcpo}{\DCPO}
\newcommand{\dcpobs}{\ensuremath{\dcpo_{\perp!}}}

\newcommand{\RR}{\ensuremath{\mathbf{R}}}
\newcommand{\Set}{\ensuremath{\mathbf{Set}}}

\newcommand{\KL}{\ensuremath{\dcpo}_{\MM}}
\newcommand{\QQ}{\ensuremath{\mathbf{Q}}}
\newcommand{\QQs}{\ensuremath{\QQ_*}}
\newcommand{\DOM}{\ensuremath{\mathbf{DOM}}}

\newcommand{\vN}{\ensuremath{\mathbf{vN}}}
\newcommand{\vNs}{\ensuremath{\vN_*}}
\newcommand{\HA}{\ensuremath{\mathbf{HA}}}
\newcommand{\HAs}{\ensuremath{\HA_*}}
\newcommand{\TD}{\ensuremath{\mathbf{TD}}}
\newcommand{\PD}{\ensuremath{\mathbf{PD}}}
\newcommand{\PDe}{\ensuremath{\PD_e}}

\newcommand{\JJ}{\ensuremath{\mathcal{J}}}
\newcommand{\II}{\ensuremath{\mathcal{I}}}
\newcommand{\MM}{\ensuremath{\mathcal{M}}}
\newcommand{\VV}{\ensuremath{\mathcal{V}}}
\newcommand{\UU}{\ensuremath{\mathcal{U}}}

\newcommand{\LL}{\ensuremath{\mathcal{L}}}
\newcommand{\TTT}{\ensuremath{\mathcal{T}}}
\newcommand{\SSS}{\ensuremath{\mathcal{S}}}

\newcommand{\OO}{\ensuremath{\mathbf{O}}}

\newcommand{\sfold}{\ensuremath{\mathrm{fold}}}
\newcommand{\qfold}{\mathrm{fold}}
\newcommand{\qunfold}{\mathrm{unfold}}
\newcommand{\sunfold}{\ensuremath{\mathrm{unfold}}}
\newcommand{\CN}{\ensuremath{\mathbb C}}
\newcommand{\Alpha}{\ensuremath{\mathrm{A}}}

\newcommand{\ff}{\mathbf{f}}
\newcommand{\vq}{\mathbf{v}}
\newcommand{\wq}{\mathbf{w}}
\newcommand{\qq}{\mathbf{q}}
\newcommand{\rr}{\mathbf{r}}
\newcommand{\xx}{\mathbf{x}}
\newcommand{\yy}{\mathbf{y}}

\newcommand{\BGamma}{\mathbf{\Gamma}}

\newcommand{\lambdaq}{{\lambda}}
\newcommand{\run}{\text{run}}

\newcommand{\qsubst}{\ensuremath{\twoheadleftarrow}}

\newcommand{\mat}[1]{\ensuremath{\mathrm{M}_{#1}}}
\newcommand{\trace}[1]{\ensuremath{\mathrm{tr}(#1)}}

\newcommand{\matrixAlg}[1]{\ensuremath{\mathrm{M}_{#1}(\mathbb C)}}

\newcommand{\probto}[1]{\ensuremath{\xrightarrow{#1}}}

\newcommand{\Halt}{\mathrm{Halt}}

\newcommand{\Val}{\mathrm{Val}}
\newcommand{\ValC}{\mathrm{ValC}}
\newcommand{\Prog}{\mathrm{Prog}}
\newcommand{\ValRel}{\mathrm{ValRel}}
\newcommand{\ValReld}{\ValRel(X,P,e)}

\newcommand{\qbit}{\textnormal{\textbf{qbit}}}
\newcommand{\bit}{\textnormal{\textbf{bit}}}

\newcommand{\fold}{\textnormal{\text{fold}}}
\newcommand{\bfold}{\textbf{fold}}

\newcommand{\unfold}{\textnormal{\text{unfold}}}
\newcommand{\bunfold}{\textbf{unfold}}

\newcommand{\FOLD}[1]{\mathbb{I}^{#1}}
\newcommand{\UNFOLD}[1]{\mathbb{E}^{#1}}

\newcommand{\lrb}[1]{{\llbracket #1 \rrbracket}}

\newcommand{\qlrb}[1]{\mathbf{\lrb{#1}}}

\newcommand{\flrb}[1]{{\llparenthesis #1 \rrparenthesis}}
\newcommand{\elrb}[2]{{\lVert #1 \rVert^{#2}}}
\newcommand{\elrbc}[1]{{\elrb{#1}{\vec C}}}
\newcommand{\elrbs}[1]{{\lVert #1 \rVert}}

\newcommand{\tleq}{\ensuremath{\vartriangleleft}}
\newcommand{\tleqd}{\ensuremath{\tleq_{X,P}^e}}
\newcommand{\qtleqd}{\ensuremath{\ \overline{\tleq_{X,P}^e}\ } }
\newcommand{\tleqone}{\ensuremath{\tleq_{X_1,P_1}^{e_1}}}

\newcommand{\tleqtwo}{\ensuremath{\tleq_{X_2,P_2}^{e_2}}}
\newcommand{\qtleqtwo}{\ensuremath{\ \overline{\tleq_{X_2,P_2}^{e_2}}\ }}

\newcommand{\btleq}{\ensuremath{\blacktriangleleft}}

\newcommand{\ol}[1]{\ensuremath{\ \overline{#1}\ }}

\newcommand{\defeq}{\stackrel{\textrm{{\scriptsize def}}}{=}}
\newcommand{\eqdef}{\defeq}

\newcommand{\naturalto}{\ensuremath{\Rightarrow}}

\newcommand{\ktimes}{%
  \mathbin{\vbox{\offinterlineskip
    \mathsurround=0pt
    \ialign{\hfil##\hfil\cr
      \normalfont\scalebox{1}{.}\cr
      \noalign{\kern-.05ex}
      $\times$\cr}
}}%
}

\newcommand{\kplus}{%
  \mathbin{\vbox{\offinterlineskip
    \mathsurround=0pt
    \ialign{\hfil##\hfil\cr
      \normalfont\scalebox{1}{.}\cr
      \noalign{\kern+.2ex}
      $+$\cr}
}}%
}

\newcommand{\kstar}{%
  \mathbin{\vbox{\offinterlineskip
    \mathsurround=0pt
    \ialign{\hfil##\hfil\cr
      \normalfont\scalebox{1}{.}\cr
      \noalign{\kern.1ex}
      $\star$\cr}
}}%
}

\newcommand{\kto}{%
  \mathbin{\vbox{\offinterlineskip
    \mathsurround=0pt
    \ialign{\hfil##\hfil\cr
      \normalfont\scalebox{2}{.}\cr
      \noalign{\kern-.7ex}
      $\to$\cr}
}}%
}

\newcommand{\ktwo}{%
  \mathbin{\vbox{\offinterlineskip
    \mathsurround=0pt
    \ialign{\hfil##\hfil\cr
      \normalfont\scalebox{1}{.}\cr
      \noalign{\kern.1ex}
      $\to$\cr}
}}%
}

\newcommand{\kcirc}{%
  \ensuremath{
    \mathbin{\raisebox{0.9pt}{\scalebox{0.7}{$\varodot$}}}
  }%
}%

\newcommand{\kid}{%
\ensuremath{\textbf{id}}
}

\makeatletter
\newsavebox{\@brx}
\newcommand{\llangle}[1][]{\savebox{\@brx}{\(\m@th{#1\langle}\)}%
\mathopen{\copy\@brx\kern-0.5\wd\@brx\usebox{\@brx}}}
\newcommand{\rrangle}[1][]{\savebox{\@brx}{\(\m@th{#1\rangle}\)}%
\mathclose{\copy\@brx\kern-0.5\wd\@brx\usebox{\@brx}}}
\makeatother

\makeatletter
\def\moverlay{\mathpalette\mov@rlay}
\def\mov@rlay#1#2{\leavevmode\vtop{%
   \baselineskip\z@skip \lineskiplimit-\maxdimen
   \ialign{\hfil$\m@th#1##$\hfil\cr#2\crcr}}}
\newcommand{\charfusion}[3][\mathord]{
    #1{\ifx#1\mathop\vphantom{#2}\fi
        \mathpalette\mov@rlay{#2\cr#3}
      }
    \ifx#1\mathop\expandafter\displaylimits\fi}
\makeatother
\newcommand{\cupdot}{\charfusion[\mathbin]{\cup}{\cdot}}

\numberwithin{equation}{section}

\theoremstyle{plain}
\newtheorem{theorem}{Theorem}[subsection]
\newtheorem{lemma}[theorem]{Lemma}
\newtheorem{proposition}[theorem]{Proposition}

\newtheorem{corollary}[theorem]{Corollary}

\theoremstyle{definition}
\newtheorem{definition}[theorem]{Definition}
\newtheorem{example}[theorem]{Example}
\newtheorem{notation}[theorem]{Notation}

\theoremstyle{remark}
\newtheorem{remark}[theorem]{Remark}
\newtheorem{assumption}[theorem]{Assumption}

\usetikzlibrary{decorations.pathmorphing}
\usetikzlibrary{decorations.markings}
\usetikzlibrary{decorations.pathreplacing}
\usetikzlibrary{arrows}
\usetikzlibrary{shapes}

\pgfdeclarelayer{edgelayer}
\pgfdeclarelayer{nodelayer}
\pgfsetlayers{edgelayer,nodelayer,main}

\tikzstyle{braceedge}=[decorate,decoration={brace,amplitude=10pt}]
\tikzstyle{square box}=[rectangle,fill=white,draw=black,minimum height=6mm,minimum width=6mm,yshift=0.7mm]
\tikzstyle{wire label}=[font=\footnotesize, auto,swap]

\tikzstyle{none}=[inner sep=0pt]
\tikzstyle{empty}=[rectangle,fill=none,draw=none]
\tikzstyle{scaled}=[rectangle,fill=none,draw=none, font=\small]

\tikzstyle{to}=[->,draw=black]
\tikzstyle{naturalto}=[-{Implies},double distance=1.5pt]
\tikzstyle{hook}=[right hook->, draw=black]
\tikzstyle{blueArr}=[->, draw=blue]

\tikzstyle{equal-arrow}=[double equal sign distance]

\tikzstyle{every picture}=[baseline=-0.25em]

\newcommand{\stikz}[2][1]{\scalebox{#1}{
\InputIfFileExists{#2}{}{\input{./tikz/#2}}
}}
\newcommand{\cstikz}[2][1]{\begin{center}\stikz[#1]{#2}\end{center}}

\begin{document}

\title{Semantics for Variational Quantum Programming}         


\author{Xiaodong Jia}
\affiliation{
  \department{School of Mathematics}
  \institution{Hunan University}
  \city{Changsha}
  \postcode{410082}
  \country{China}          
}

\author{Andre Kornell}
\affiliation{
  \department{Department of Computer Science}
  \institution{Tulane University}
  \city{New Orleans}
  \state{Louisiana}
  \country{USA}          
}

\author{Bert Lindenhovius}
\affiliation{
  \department{Department of Knowledge-Based Mathematical Systems}
  \institution{Johannes Kepler Universität}
  \city{Linz}
  \country{Austria}
}

\author{Michael Mislove}
\affiliation{
  \department{Department of Computer Science}
  \institution{Tulane University}
  \city{New Orleans}
  \state{Louisiana}
  \country{USA}          
}

\author{Vladimir Zamdzhiev}
\affiliation{
  \department{Université de Lorraine, LORIA}
  \institution{Inria}
  \city{Nancy}
  \postcode{F 54000}
  \country{France}          
}

\begin{abstract}
We consider a programming language that can manipulate both classical and
quantum information. Our language is type-safe and designed for variational
quantum programming, which is a hybrid classical-quantum computational
paradigm.  The classical subsystem of the language is the Probabilistic
FixPoint Calculus (PFPC), which is a lambda calculus with mixed-variance
recursive types, term recursion and probabilistic choice. The quantum
subsystem is a first-order linear type system that can manipulate quantum
information. The two subsystems are related by mixed classical/quantum terms
that specify how classical probabilistic effects are induced by quantum
measurements, and conversely, how classical (probabilistic) programs can
influence the quantum dynamics. We also describe a sound and computationally
adequate denotational semantics for the language.  Classical probabilistic
effects are interpreted using a recently-described commutative probabilistic
monad on $\DCPO$. Quantum effects and resources are interpreted in a category
of von Neumann algebras that we show is enriched over (continuous) domains.
This strong sense of enrichment allows us to develop novel semantic methods
that we use to interpret the relationship between the quantum and classical
probabilistic effects. By doing so we provide the first denotational analysis
that relates models of classical probabilistic programming to models of
quantum programming.
\end{abstract}

\ccsdesc[300]{Semantics of programming languages}

\keywords{Quantum Programming, Probabilistic Programming, Semantics}  

\maketitle

\section{Introduction}
\label{sec:intro}

Variational quantum algorithms \cite{vqa1,vqa2} are increasingly
important in quantum computation. The main idea is to use hybrid
classical-quantum algorithms that work in tandem to solve computational
problems. The classical part of the computation is executed on a classical
processor and the quantum part on a quantum device. During the overall
computation, intermediary results produced by the quantum device occur with
certain probabilities, and then are passed to the classical processor, which
performs computations that are used to tune the parameters of
the quantum component of the algorithm, thereby influencing the quantum
dynamics.

These kinds of hybrid classical-quantum algorithms pose interesting challenges
for the design of suitable programming languages.  Clearly, if we wish to
understand how to program in such scenarios, we need to devise a type system
equipped with an operational semantics that correctly models the manipulation
of \emph{quantum resources}. This includes accounting for the fact that quantum
measurements induce \emph{probabilistic computational effects} that are
inherited by the classical side of the system. Moreover, quantum information
behaves very differently from classical information. For instance, quantum
information cannot be copied \cite{no-cloning}.  In order to avoid potential
runtime errors, a \emph{substructural typing discipline}
\cite{girard,benton-small,benton-wadler} where contraction is restricted is
appropriate for the quantum subsystem. But, when manipulating classical
information, such restrictions are unnecessary and often inconvenient.
Therefore we wish to have a classical (non-linear) subsystem together with a
quantum (linear) one that interact nicely with each other.  Furthermore,
separating the quantum and classical modes of operations has the added benefit
that it makes it easier to \emph{extend} existing classical programming
languages with the necessary features for type-safe variational quantum
programming.

The purpose of the present paper is to address this challenge by describing a
type-safe programming language that combines classical (probabilistic)
computation with quantum computation. Another one of our goals is to provide a
denotational interpretation so that we may establish useful reasoning
principles and therefore cement the design of our language.

\subsection{Our Contributions}

We describe a programming language that is suitable for hybrid
classical-quantum computation that we call $\mathtt{VQPL}$, the
\emph{Variational Quantum Programming Language} (\secref{sec:syntax}).  The
language has two kinds of judgements: a classical (non-linear) judgement that
represents classical programs, and a quantum (linear) judgement that represents
quantum programs. Our type system also contains hybrid classical-quantum
formation rules that explain how classical probabilistic and quantum
computation interact with each other (see Figure \ref{fig:syntax-mixed}).

From an operational perspective, \VQPL supports both classical probabilistic
and quantum effects. The quantum dynamics are modelled via a probabilistic
reduction relation on quantum configurations (terms with quantum data embedded
within them), where the probabilities of reduction are determined in accordance
with the laws of quantum mechanics. The classical dynamics are modelled via a
probabilistic reduction relation on terms, where the probabilities of reduction
are induced by the quantum dynamics. We show that our system $\VQPL$ is type-safe
(\secref{sub:type-safety}).

We also provide a denotational interpretation of our system. We use a
recently-described commutative probabilistic monad on the category $\DCPO$
\cite{m-monad} in order to interpret the classical probabilistic effects (we
recall this construction in \secref{sec:probabilistic}).  We interpret quantum
effects and resources in the category of hereditarily atomic von Neumann
algebras (\secref{sec:operator-algebras}), which are mathematical structures
used by physicists to study quantum foundations \cite{takesaki:oa1}.  We prove that this category
is enriched over \emph{continuous domains} (\secref{subsection:enrichment vN}).
This is a very strong sense of enrichment that allows us to develop novel
semantic methods that we use to interpret the relationship between the quantum
and classical probabilistic effects (\secref{sec:relationship}). In particular,
we show that the theory of Kegelspitzen \cite{keimelplotkin17} provides a
crucial link between the two different ways that probability arises on the
classical and quantum sides, respectively. This allows us to systematically
present all the relevant mathematical structure within a categorical model
(\secref{sec:categorical-model}) and to use our model to provide a sound and strongly
adequate interpretation of $\VQPL$ (\secref{sec:semantics}). Our paper is the first to present a mathematical and denotational analysis on the
link between models of classical probabilistic programming and quantum
programming. We discuss related work and provide concluding remarks in
\secref{sec:related}.

\section{VQPL - The Variational Quantum Programming Language}
\label{sec:syntax}

In this section we describe the syntax and operational semantics for \VQPL. The
classical subsystem  is the \emph{Probabilistic FixPoint Calculus}
($\mathtt{PFPC}$), the same language as in \cite{m-monad}.  $\PFPC$ is a
call-by-value simply-typed lambda calculus with mixed-variance recursive types,
(induced) term recursion and discrete probabilistic choice.  The quantum
fragment of the language is a first-order linear type system with inductive
types and equipped with the usual primitives for manipulating quantum
information.  This fragment is most similar to \cite{qpl-fossacs}, however in
the present paper we choose a Church-style syntax in order to more easily
relate it to the classical subsystem. The distinguishing feature of our system
is the mixed linear/non-linear and quantum/classical rules that allow the
programmer to switch between the classical and quantum modes of operation.
These features make our language suitable for programming variational
quantum algorithms, where both classical and quantum computation work in
synchrony in order to solve computational problems. Our mixed quantum/classical
rules have some similarities with the QWIRE/EWIRE languages \cite{qwire,ewire},
but both of these languages have some severe limitations that make them unsuitable for
describing variational quantum algorithms, whereas our language does not. This
is discussed in more detail in \secref{sec:related}.

In order to make the paper easier to read, we use \textbf{bold} notation for the
quantum types, contexts and terms, so that we can easily distinguish them from
the classical primitives.

\subsection{The Type Structure}

We use $X,Y$ to range over \emph{classical type variables} and we use $\mathbf
X,\mathbf Y$ to range over \emph{quantum type variables}.  We use $\Theta$ and
$\mathbf \Theta$ to range over classical and quantum \emph{type contexts},
respectively. Type variables and type contexts are used for the formation of
recursive types, just like in FPC \cite{fpc-syntax,fiore-plotkin}.  We say that
a classical type context $\Theta = X_1, \ldots, X_n$ is \emph{well-formed},
written $\Theta \vdash$, whenever all type variables within it are distinct, and
likewise for a quantum type context.  The \emph{classical types} of our
language are ranged over by $P,R$, and the \emph{quantum types} are ranged over
by $\AAA, \BBB$. The grammars and formation rules for our types are specified
in Figure \ref{fig:syntax-grammars-types}. The notation $\Theta \vdash P$
indicates that type $P$ is \emph{well-formed} in type context $\Theta$, and
likewise for the quantum types. Of course, we are only interested in
well-formed types and from now on we only deal with such types.  The
\emph{closed classical types} are those where $\cdot \vdash P$ and the
\emph{closed quantum types} are those where $\cdot \vdash \AAA$.  Notice that
recursive types may be formed with no restrictions on the admissible logical
polarities, just like in FPC.

\begin{figure}
\small{
\centering
\begin{tabular}{l l l l l l l l}
  Quantum Type Variables   & $\mathbf X, \mathbf Y$ &     & \multicolumn{4}{l}{Classical Type Variables \quad $X,Y$      }    \\ 
  Quantum Type Contexts    & $\mathbf \Theta$       & ::= & $\mathbf X_1, \mathbf X_2, \ldots, \mathbf X_n$ \\
  Classical Type Contexts  & $\Theta$               & ::= & $X_1, X_2, \ldots, X_n$  \\[1ex]
  Quantum Types            & $\mathbf A, \mathbf B$ & ::= & \multicolumn{5}{l}{ $\mathbf X$ | $\mathbf I$ | \qbit\ | $\mathbf{A \oplus B}$ | $\mathbf{A \otimes B}$ | $\mathbf{\mu \mathbf X.A}$} \\
	Classical Types          & $P, R$                 & ::= & \multicolumn{5}{l}{$X$  | $1$ | $P+R$ | $P \times R$ | $P \to R$ | $Q(\mathbf A, \mathbf B)$ | $\mu X.P$} \\
  Observable Quantum Types & $\mathbf O$            & ::= & \multicolumn{5}{l}{ $\mathbf X$ | $\mathbf I$ | $\mathbf{O_1 \oplus O_2}$ | $\mathbf{O_1 \otimes O_2}$ | $\mathbf{\mu \mathbf X.O}$} \\
  Observable Classical Types & $O$                  & ::= & \multicolumn{5}{l}{ $X$ | $1$ | $O_1 + O_2$ | $O_1 \times O_2$ | $\mu X.O$} \\[1ex]
\end{tabular}
\[
    \begin{bprooftree}
    \AxiomC{$\mathbf \Theta \vdash $}
    \UnaryInfC{$\mathbf \Theta \vdash \mathbf \Theta_i$}
    \end{bprooftree}
    \quad
    \begin{bprooftree}
    \AxiomC{$\mathbf \Theta \vdash $}
    \UnaryInfC{$\mathbf \Theta \vdash \mathbf I$}
    \end{bprooftree}
    \quad
    \begin{bprooftree}
    \AxiomC{$\mathbf \Theta \vdash $}
    \UnaryInfC{$\mathbf \Theta \vdash \qbit$}
    \end{bprooftree}
    \quad
    \begin{bprooftree}
    \AxiomC{$\mathbf \Theta \vdash\mathbf  A$}
    \AxiomC{$\mathbf \Theta \vdash\mathbf  B$}
    \RightLabel{$\star \in \{\oplus, \otimes \}$}
    \BinaryInfC{$\mathbf \Theta \vdash\mathbf  A \star\mathbf  B$}
    \end{bprooftree}
    \quad
    \begin{bprooftree}
    \AxiomC{$\mathbf \Theta,\mathbf  X \vdash\mathbf  A$}
    \UnaryInfC{$\mathbf \Theta \vdash \mu\mathbf  X. \mathbf A$}
    \end{bprooftree}
\]
\[
    \begin{bprooftree}
    \AxiomC{$\Theta \vdash $}
    \UnaryInfC{$\Theta \vdash \Theta_i$}
    \end{bprooftree}
    \quad
    \begin{bprooftree}
    \AxiomC{$\Theta \vdash $}
    \UnaryInfC{$\Theta \vdash 1$}
    \end{bprooftree}
    \quad
    \begin{bprooftree}
    \AxiomC{$\Theta \vdash P$}
    \AxiomC{$\Theta \vdash R$}
    \RightLabel{$\star \in \{+, \times, \to \}$}
    \BinaryInfC{$\Theta \vdash P \star R$}
    \end{bprooftree}
    \quad
    \begin{bprooftree}
    \AxiomC{$\cdot \vdash \mathbf A$}
    \AxiomC{$\cdot \vdash \mathbf B$}
    \AxiomC{$\Theta \vdash$}
    \TrinaryInfC{$\Theta \vdash Q(\mathbf A, \mathbf B)$}
    \end{bprooftree}
    \quad
    \begin{bprooftree}
    \AxiomC{$\Theta, X \vdash P$}
    \UnaryInfC{$\Theta \vdash \mu X. P$}
    \end{bprooftree}
\]
\[
|\mathbf X| \eqdef X \quad |\mathbf I| \eqdef 1 \quad |\mathbf{O_1 \otimes O_2} | \eqdef |\mathbf{O_1}| \times |\mathbf{O_2}| \quad |\mathbf{O_1 \oplus O_2} | \eqdef |\mathbf{O_1}| + |\mathbf{O_2}|  \quad |\mathbf{\mu X. O} | \eqdef \mu X. |\mathbf{O}|
\]
}
\vspace{-3ex}\caption{Grammars and formation rules for types and translation between observable types.}
\label{fig:syntax-grammars-types}
\end{figure}

We now explain how our types should be understood. On the quantum side:
$\mathbf I$ is the quantum unit type; $\qbit$ is the type of qubits (quantum
bits); $\AAA \oplus \BBB$ represents quantum sum types; $\AAA \otimes \BBB$
represents quantum pair types; $\mathbf{\mu \mathbf X.A}$ is used to form
quantum inductive types. All terms of quantum type obey a linear typing
discipline and so these types should be viewed as being linear. On the
classical side: $1$ is the classical unit type; $P+R$ is for classical sum
types; $P \times R$ is for classical pair types; $P \to R$ is for classical
(higher-order) function types; $\mu X. P$ is used to form classical recursive types; $Q(\AAA,
\BBB)$ is the type of \emph{first-order} quantum lambda abstractions between
quantum types $\AAA$ and $\BBB$. All terms of classical type follow a
non-linear typing discipline (no restrictions on weakening and contraction), so
they should be understood as being non-linear.  Notice that the type
$Q(\AAA,\BBB)$ is \emph{classical} (non-linear). This is because our quantum
lambda abstractions are first-order and therefore they may be used any number
of times (including zero). This type would correspond to $!(\AAA
\multimap \BBB)$ in a call-by-value linear lambda calculus and may be informally thought of in this way.

\begin{example}
\label{ex:basic-types}
Some important (closed) types are defined as follows:
  \emph{Booleans} as $\mathtt{Bool} \eqdef 1 + 1$;
  \emph{Bits} as $\mathbf{Bit} \eqdef \mathbf I \oplus \mathbf I$;
  \emph{Natural numbers} as $\mathtt{Nat} \eqdef \mu X. 1 + X$;
  \emph{Linear/Quantum natural numbers} as $\mathbf{QNat} \eqdef \mathbf \mu \mathbf X. \mathbf I \oplus \mathbf X$;
  \emph{Lists of type} $A$ as $\mathtt{List}(A) \eqdef \mu X. 1 + (A \times X)$;
  \emph{Linear/Quantum lists of type} $\AAA$ as $\mathbf{List}(\AAA) \eqdef \mu \mathbf{X. I \oplus (A \otimes X)}$;
  \emph{Classical Streams of type} $A$ as $\mathtt{Stream}(A) \eqdef \mu X. 1 \to A \times X$.
\end{example}

A \emph{subset} of our classical/quantum types are the \emph{observable}
classical/quantum types, which are defined in Figure
\ref{fig:syntax-grammars-types}. We use $O$ and $\mathbf O$ to range over the
observable classical/quantum types, respectively.  These types play an
important role for some of the mixed quantum-classical rules that we explain
later. The observable quantum types may also be understood from a physical
perspective because values of these types correspond to physically observable
information. An example of a non-observable quantum type is $\qbit$. Indeed,
observing a qubit in the physical sense is done via a quantum measurement, which
destroys the qubit and produces a bit as output (note that $\mathbf{Bit}$ is
observable in our system). The observable classical types are exactly the
\emph{ground types}, i.e., types formed without any use of classical/quantum
function space. The observable quantum types are in a 1-1 correspondence with
the observable classical types. For each observable quantum type $\OO$, we
write $|\OO|$ to indicate its observable classical counterpart. See
Figure \ref{fig:syntax-grammars-types} for a precise definition of $|-|$.

\subsection{The Term Language}

\begin{figure}
\small{
\centering
\begin{tabular}{l l l l l l l l}
  Quantum Variables   & $\mathbf x, \mathbf y$ \qquad     & \multicolumn{4}{l}{Classical Variables \quad $x, y$ } \qquad Quantum Configurations\quad $\mathcal C$ ::= $[\ket \psi, \ell, \mathbf q]$ \\   
  Classical Terms          & $m,n$                  & ::= & \multicolumn{5}{l}{ $x\ |\ ()\ |\ (m,n)\ |\ \pi_1 m\ |\ \pi_2 m\ |\  \fold\ m\ |\ \unfold\ m\ |\ 
  \lambda x. m\ |\ mn\ |\ $}\\
  && & $\emph{in}_1 m\ |\ \emph{in}_2 m\ |\  (\text{case}\ m\ \text{of}\ \emph{in}_1 x \Rightarrow n_1\ |\ \emph{in}_2 y \Rightarrow n_2)\ | $  \\
                           &                        &     & \multicolumn{5}{l}{$ \lambdaq \mathbf{(x_1, \ldots, x_n) . q } \ |\ \text{new}\ |\ \text{meas}\ |\ U\ |\ \text{run } \mathcal C$ } \\
  Quantum Terms            & $\mathbf q, \mathbf r$ & ::= & \multicolumn{5}{l}{ $\mathbf{ x\ |\ *\ |\ q;r \ |\ q \otimes r\ |\ let\ x \otimes y = q\ in\ r\ |\ fold\ q\ |\ unfold\ q\ |\ } m\mathbf q\ |\ \mathbf{init}\ m\ |$ } \\
                           &                        &     & \multicolumn{5}{l}{ $\mathbf{ in_1\ q\ | in_2\ q\ |\ (case\ q\ of\ in_1 x \Rightarrow r_1\ |\ in_2 y \Rightarrow r_2 }) \ |\ \mathbf{let}\  x = \text{lift}\ \mathbf q\ \mathbf{ in\ r}  $ } \\
  Classical Values         & $v,w$                  & ::= & \multicolumn{5}{l}{ $x\ |\ ()\ |\ (v,w)\ |\ \emph{in}_1 v\ |\ \emph{in}_2 v\ |\ \fold\ v\ |\ \lambda x. m\ |\ \lambdaq \mathbf{(x_1, \ldots, x_n) . q } \ |\ $ }\\
  &&& $\text{new}\ |\ \text{meas}\ |\ U $  \\
  Quantum   Values         & $\mathbf v, \mathbf w$ & ::= & \multicolumn{5}{l}{ $\mathbf{ x\ |\ *\ |\ v \otimes w\ |\ in_1\ v\ |\ in_2\ v\ |\ fold\ v  } $ }
\end{tabular}
\vspace{-2ex}\caption{Grammars for terms, values and quantum configurations. }
\label{fig:syntax-grammars-terms}
\[
  \begin{bprooftree}
  \AxiomC{\phantom{$\Phi \vdash $}}
  \UnaryInfC{$\Phi, x:P \vdash x: P$}
  \end{bprooftree}
  \begin{bprooftree}
  \AxiomC{\phantom{$\Phi \vdash $}}
  \UnaryInfC{$\Phi \vdash (): 1$}
  \end{bprooftree}
  \begin{bprooftree}
  \def\ScoreOverhang{0.5pt}
  \AxiomC{$ \Phi \vdash m : P$}
  \AxiomC{$ \Phi \vdash n : R$}
  \BinaryInfC{$ \Phi \vdash (m, n) : P \times R$}
  \end{bprooftree}
  \ 
  \begin{bprooftree}
  \def\ScoreOverhang{0.5pt}
  \AxiomC{$ \Phi \vdash m : P_1 \times P_2$}
  \RightLabel{$i \in \{1,2\}$}
  \UnaryInfC{$ \Phi \vdash \pi_i m : P_i$}
  \end{bprooftree}
\]

\[
  \begin{bprooftree}
  \def\ScoreOverhang{0.5pt}
  \AxiomC{$ \Phi \vdash m : P$}
  \UnaryInfC{$ \Phi \vdash \emph{in}_{1} m : P+R$}
  \end{bprooftree}
  \begin{bprooftree}
  \def\ScoreOverhang{0.5pt}
  \AxiomC{$ \Phi \vdash m : R$}
  \UnaryInfC{$ \Phi \vdash \emph{in}_{2} m : P+R$}
  \end{bprooftree}
  \begin{bprooftree}
  \def\ScoreOverhang{0.5pt}
  \AxiomC{$ \Phi \vdash m : P_1+P_2$}
  \AxiomC{$ \Phi, x : P_1 \vdash n_1 : R$}
  \AxiomC{$ \Phi, y : P_2 \vdash n_2 : R$}
  \TrinaryInfC{$ \Phi \vdash ( \text{case}\ m\ \text{of}\ \emph{in}_1 x \Rightarrow n_1\ |\ \emph{in}_2 y \Rightarrow n_2 ) : R$}
  \end{bprooftree}
\]

\[
  \begin{bprooftree}
  \def\ScoreOverhang{0.5pt}
  \AxiomC{$ \Phi, x: P \vdash m : R$}
  \UnaryInfC{$ \Phi \vdash \lambda x^P . m : P \to R$}
  \end{bprooftree}
  \begin{bprooftree}
  \def\ScoreOverhang{0.5pt}
  \AxiomC{$ \Phi \vdash m : P \to R$}
  \AxiomC{$ \Phi \vdash n : P$}
  \BinaryInfC{$ \Phi \vdash mn : R$}
  \end{bprooftree}
  \begin{bprooftree}
  \def\ScoreOverhang{0.5pt}
  \AxiomC{$ \Phi \vdash m : P[\mu X. P / X]$}
  \UnaryInfC{$ \Phi \vdash \fold{}\ m: \mu X. P$}
  \end{bprooftree}
  \ 
  \begin{bprooftree}
  \def\ScoreOverhang{0.5pt}
  \AxiomC{$ \Phi \vdash m : \mu X. P$}
  \UnaryInfC{$ \Phi \vdash \unfold\ m : P[\mu X. P / X]$}
  \end{bprooftree}
\]
\vspace{-2ex}\caption{Formation rules for terms in the (classical) FPC subsystem.}
\label{fig:classical-term-syntax}


\[
  \begin{bprooftree}
  \AxiomC{\phantom{$\vdash \Phi, x: \mathbf A$}}
  \UnaryInfC{$ \Phi; \mathbf{ x: A \vdash x: A}$}
  \end{bprooftree}
  \quad
  \begin{bprooftree}
  \AxiomC{\phantom{$\vdash \Phi, x: \mathbf A$}}
  \UnaryInfC{$ \Phi; \cdot \vdash *: \mathbf I  $}
  \end{bprooftree}
  \begin{bprooftree}
  \AxiomC{$ \Phi; \mathbf{ \Gamma_1 \vdash q : I} $}
  \AxiomC{$ \Phi; \mathbf{ \Gamma_2 \vdash r : A} $}
  \BinaryInfC{$ \Phi; \mathbf{ \Gamma_1, \Gamma_2 \vdash q ; r : \mathbf A} $}
  \end{bprooftree}
  \begin{bprooftree}
  \def\ScoreOverhang{0.5pt}
  \AxiomC{$ \Phi; \mathbf{ \Gamma_1 \vdash q : A} $}
  \AxiomC{$ \Phi; \mathbf{ \Gamma_2 \vdash r : B} $}
  \BinaryInfC{$ \Phi; \mathbf{ \Gamma_1, \Gamma_2 \vdash q \otimes r : A \otimes B} $}
  \end{bprooftree}
\]

\[
  \begin{bprooftree}
  \def\ScoreOverhang{0.5pt}
  \AxiomC{$ \Phi; \mathbf{ \Gamma_1 \vdash q : A_1 \otimes A_2} $}
  \AxiomC{$ \Phi; \mathbf{ \Gamma_2, x : A_1, y : A_2 \vdash r : B} $}
  \BinaryInfC{$ \Phi; \mathbf{ \Gamma_1, \Gamma_2 \vdash let\ x \otimes y = q\ in\ r : B } $}
  \end{bprooftree}
  \begin{bprooftree}
  \def\ScoreOverhang{0.5pt}
  \AxiomC{$ \Phi; \mathbf{ \Gamma \vdash q : A[\mu X. A / X] } $}
  \UnaryInfC{$ \Phi; \mathbf{ \Gamma \vdash \bfold\ q : \mu X. A } $}
  \end{bprooftree}
  \begin{bprooftree}
  \def\ScoreOverhang{0.5pt}
  \AxiomC{$ \Phi; \mathbf{ \Gamma \vdash q : \mu X. A} $}
  \UnaryInfC{$ \Phi; \mathbf{ \Gamma \vdash \bunfold\ q : A[\mu X. A / X]} $}
  \end{bprooftree}
\]

\[
  \quad
  \begin{bprooftree}
  \def\ScoreOverhang{0.5pt}
  \AxiomC{$ \Phi; \mathbf{ \Gamma \vdash q : A}$ }
  \UnaryInfC{$ \Phi; \mathbf{ \Gamma \vdash in_1\ q : A \oplus B} $}
  \end{bprooftree}
  \begin{bprooftree}
  \def\ScoreOverhang{0.5pt}
  \AxiomC{$ \Phi; \mathbf{ \Gamma \vdash q : B} $}
  \UnaryInfC{$ \Phi; \mathbf{ \Gamma \vdash in_2\ q : A \oplus B} $}
  \end{bprooftree}
  \]
  
  \[
  \begin{bprooftree}
  \def\ScoreOverhang{0.5pt}
  \AxiomC{$ \Phi; \mathbf{ \Gamma_1 \vdash q : A_1 \oplus A_2 } $}
  \AxiomC{$ \Phi; \mathbf{ \Gamma_2, x : A_1 \vdash r_1 : B} $}
  \AxiomC{$ \Phi; \mathbf{ \Gamma_2, y : A_2 \vdash r_2 : B} $}
  \TrinaryInfC{$ \Phi; \mathbf{ \Gamma_1, \Gamma_2 \vdash (case\ q\ of\ in_1 x \Rightarrow r_1\ |\ in_2 y \Rightarrow r_2) : B} $}
  \end{bprooftree}
\]
\vspace{-2ex}\caption{Formation rules for terms in the purely linear first-order subsystem.}
\label{fig:first-order-term-syntax}
\[
  \begin{bprooftree}
  \def\ScoreOverhang{0.5pt}
  \AxiomC{ \phantom{$\Phi; \vdash  $}}
  \UnaryInfC{$ \Phi \vdash \text{new} : Q(\bit, \qbit) $}
  \end{bprooftree}
  \quad
  \begin{bprooftree}
  \def\ScoreOverhang{0.5pt}
  \AxiomC{ \phantom{$\Phi; \vdash  $}}
  \UnaryInfC{$ \Phi \vdash \text{meas} : Q(\qbit, \bit) $}
  \end{bprooftree}
  \quad
  \begin{bprooftree}
  \def\ScoreOverhang{0.5pt}
  \AxiomC{$U$ is a unitary of arity of $n$}
  \UnaryInfC{$ \Phi \vdash U : Q(\qbit^{\otimes n}, \qbit^{\otimes n}) $}
  \end{bprooftree}
  \qquad
  \bit \equiv \mathbf I \oplus \mathbf I
\]
\vspace{-2ex}\caption{Formation rules for term constants that manipulate quantum information.}
\label{fig:quantum-constants-syntax}
\[
|\mathbf *| \eqdef () \quad |\mathbf{v \otimes w}| \eqdef (|\mathbf v|, |\mathbf w|)  \quad |\mathbf{in_1\ v} | \eqdef \emph{in}_1 |\mathbf v| \quad |\mathbf{in_2\ v} | \eqdef \emph{in}_2|\mathbf{v}| \quad |\mathbf{fold\ v}| \eqdef \text{fold}\ |\mathbf v|
\]
\vspace{-2ex}\caption{Translation between \emph{closed} and observable quantum/classical values. }
\label{fig:syntax-translate}

\[
  \begin{bprooftree}
  \def\ScoreOverhang{0.5pt}
  \AxiomC{$ \Phi; \mathbf{ x_1 : A_1, \ldots, x_n : A_n \vdash q : B} $ }
  \UnaryInfC{$ \Phi \vdash \lambdaq \mathbf{(x_1, \ldots, x_n) . q } : Q(\mathbf{A_1 \otimes \cdots \otimes A_n, B}) $ }
  \end{bprooftree}
  \quad
  \begin{bprooftree}
  \def\ScoreOverhang{0.5pt}
  \AxiomC{$ \Phi \vdash m : Q(\mathbf A, \mathbf B)$}
  \AxiomC{$ \Phi; \mathbf{ \Gamma \vdash q : A} $}
  \BinaryInfC{$ \Phi; \mathbf \Gamma \vdash m \mathbf q : \mathbf B $ }
  \end{bprooftree}
  \quad
  \begin{bprooftree}
  \def\ScoreOverhang{0.5pt}
  \AxiomC{$ \Phi \vdash \mathcal C : \mathbf O; \qbit^k$}
  \UnaryInfC{$ \Phi \vdash \text{run } \mathcal C : |\mathbf O|$}
  \end{bprooftree}
\]

\[
  \begin{bprooftree}
  \def\ScoreOverhang{0.5pt}
  \AxiomC{$ \Phi \vdash m : |\mathbf O| $}
  \UnaryInfC{$ \Phi; \cdot \vdash \mathbf{init}\ m : \mathbf O$}
  \end{bprooftree}
  \quad
  \begin{bprooftree}
  \def\ScoreOverhang{0.5pt}
  \AxiomC{$ \Phi; \mathbf{ \Gamma_1 \vdash q : \mathbf O } $}
  \AxiomC{$ \Phi, x : |\mathbf O| ; \mathbf{ \Gamma_2 \vdash r : \mathbf A} $}
  \BinaryInfC{$ \Phi; \mathbf{ \Gamma_1, \Gamma_2} \vdash \mathbf{let}\  x = \text{lift}\ \mathbf q\ \mathbf{ in\ r} : \mathbf A$}
  \end{bprooftree}
\]
\vspace{-2ex}\caption{Formation rules for terms that mediate between the quantum and classical modes of operation.}
\label{fig:syntax-mixed}
}
\end{figure}

For the formation of terms and term contexts, we implicitly assume that all
types within are closed and well-formed.  We use $x,y$ to range over
\emph{classical term variables} and $\xx, \yy$ to range over \emph{quantum term
variables}.  Classical term contexts are ranged over by $\Phi$ and quantum term
contexts by $\BGamma$. The (well-formed) term contexts are simply lists of
(distinct) variables with their types.

The term grammars of \VQPL are specified in Figure \ref{fig:syntax-grammars-terms}.
We write $\Phi \vdash m : P$ to indicate that a classical term $m$ is well-formed and has type $P$ given classical context $\Phi$.
We write $\Phi ; \BGamma \vdash \qq : \AAA$ to indicate that a quantum term $\qq$ is well-formed and has type $\AAA$, given classical context $\Phi$ and quantum context $\BGamma$.
A classical term $m$ of type $A$ is \emph{closed} when $\cdot \vdash m : A$ and in this case we also simply write $m : A.$ Likewise we write $\qq : \AAA$ when $\cdot; \cdot \vdash \qq : \AAA$, and then we also say $\qq$ is closed.
We use $v,w$ to range over classical values and $\vq, \wq$ to range over quantum values (see Figure \ref{fig:syntax-grammars-terms}).

\begin{example}
  Important closed values include:
  the (classical) \emph{false} and \emph{true} values given by $\mathtt{ff} \eqdef \mathtt{in}_1 () : \mathtt{Bool}$ and $\mathtt{tt} \eqdef \mathtt{in}_2 () : \mathtt{Bool};$
  the false and true bits are defined by $\mathbf{ff} \eqdef \mathbf{in}_1 * : \mathbf{Bit}$ and $\mathbf{tt} \eqdef \mathbf{in}_2 * : \mathbf{Bit};$
  the \emph{zero natural number} $\mathtt{zero} \eqdef \mathtt{fold}\ \mathtt{in}_1 () : \mathtt{Nat}$ and the \emph{successor function}
  $\mathtt{succ}\ \eqdef \lambda n^{\mathtt{Nat}}. \mathtt{fold}\ \mathtt{in}_2 n : \mathtt{Nat} \to \mathtt{Nat};$ quantum versions of $\mathtt{zero}$ and $\mathtt{succ}$ may also be defined.
\end{example}

Execution of (quantum) programs is described by the small-step call-by-value operational semantics  in Figures \ref{fig:structural-reduction}--\ref{fig:switching-rules-operational}.
If $m$ and $n$ are classical terms, we write $m \probto{p} n$ to indicate that  $m$ reduces to  $n$ with probability $p \in [0,1]$ in exactly one step.
Note that the probabilistic behaviour of reduction in the classical subsystem is induced by quantum measurements from the quantum subsystem.

\subsubsection{Quantum Configurations}

Reduction for the quantum fragment is described, as usual, in terms of \emph{quantum configurations} $[\ket \psi, \ell, \mathbf q]$, which may be seen as terms with embedded quantum information.
We describe quantum configurations following \cite{qpl-fossacs,quant-semantics}.
Given $n \in \mathbb N,$ we write $\qbit^n \defeq \qbit \otimes \cdots \otimes \qbit$ for the $n$-fold tensor product of $\qbit$.

\begin{definition}[Quantum Configuration]
  A \emph{quantum configuration} is a triple $[\ket \psi, \ell, \qq]$, where:
    $\ket \psi$ is a \emph{pure quantum state}, i.e., a normalised vector in $\mathbb C^{2^n};$ 
    $\qq$ is a quantum term;
    $\ell : \mathrm{QFV}(\qq) \to \{1, \ldots, \dim(\ket \psi)\}$ is a function from the set of free \emph{quantum} variables of $\qq$ into the indicated set, where $\dim(\ket \psi) \defeq n$ is \emph{dimension} of $\ket \psi$.
    We refer to $\ell$ as the \emph{linking function}.
  A quantum configuration $[\ket \psi, \ell, \qq]$ is \emph{well-formed} in classical context $\Phi$ with type $\AAA$ and $k$ auxiliary qubits, which we write as $\Phi \vdash [\ket \psi, \ell, \qq] : \AAA ; \qbit^k,$ whenever the following conditions are satisfied:
  \begin{itemize}
    \item $\Phi; \xx_1 : \qbit, \ldots, \xx_m : \qbit \vdash \qq : \AAA$ is a well-formed quantum term.
    \item $\dim(\ket \psi)= m + k.$
    \item The linking function $\ell : \{ \xx_1, \ldots, \xx_m \} \to \{1, \ldots, m+k\}$ is injective.
  \end{itemize}
  A configuration $[\ket \psi, \ell, \qq]$ is \emph{total} if $\Phi \vdash [\ket \psi, \ell, \qq] : \AAA ; \qbit^0$, which we abbreviate by  $\Phi \vdash [\ket \psi, \ell, \qq] : \AAA$. Thus, in a total configuration $\ell$ defines a 1 - 1 correspondence  between the qubits of $\ket \psi$ and the free quantum variables of $\qq.$
  A configuration is  \emph{closed} if $\cdot \vdash [\ket \psi, \ell, \qq] : \AAA ; \qbit^k$, for some $k \in \mathbb N.$
\end{definition}

We are primarily interested in well-formed configurations that are both total
and closed. Nevertheless, the premises of the structural reduction rules in the operational semantics 
include non-total configurations that have some auxiliary qubits not used by the quantum term, so it is necessary also to consider non-total configurations
(see \cite{quant-semantics,qpl-fossacs} for more details).
Otherwise, the configurations in the premises would not be typable, so it is necessary to allow auxiliary qubits as part of the formation conditions.
Likewise, the denotational semantics includes configurations that are not closed, because the interpretation of
closed terms may be defined using non-closed terms (e.g. lambda abstractions).

The linking function $\ell$ in a configuration $[\ket \psi, \ell, \qq]$ associates the free variables of
$\qq$, each of type $\qbit$, to specific qubits of the quantum state
$\ket \psi$; $\ell$ is needed because some of the qubits in $\ket \psi$ may be entangled, 
in which case $\ket \psi$ cannot be broken down  into smaller quantum states. 

We  use calligraphic letters $\mathcal C, \mathcal D$ to range over
quantum configurations. Given quantum configurations $\mathcal C$ and $\mathcal
D$ we  write $\mathcal C \probto{p} \mathcal D$ to indicate that $\mathcal
C$ reduces to $\mathcal D$ with probability $p$ in exactly one step. This is how we model the execution of quantum programs.
A \emph{value configuration} is a configuration $\mathcal V = [\ket \psi, \ell, \vq]$, where $\vq$ is a quantum value.
Reduction in the quantum subsystem terminates at value configurations, just as reduction in the classical system terminates at values.


\subsubsection{The Subsystem FPC}

We have organised the term formation rules and the associated reduction rules
into several subsystems, which we now describe.
Figure \ref{fig:classical-term-syntax} specifies the formation rules for the
classical terms that make up the subsystem FPC (which is well-known
\cite{fiore-thesis,fpc-syntax,fiore-plotkin}). The notation $P[\mu X. P / X]$
represents type substitution, which is defined as usual. The reduction rules
for these terms are standard and are shown in Figure \ref{fig:structural-reduction}.

\subsubsection{The Subsystem QPL}
Figure \ref{fig:first-order-term-syntax} describes the formation rules for the
quantum terms that make up a first-order linear subsystem with inductive types.
These terms 
and their reduction rules are all standard, but
they are now described on quantum configurations in Figure
\ref{fig:structural-reduction}. For the structural reduction rules involving
quantum evaluation contexts, the notation $\ell \cap \ell_0 = \varnothing$
indicates that the linking functions $\ell$ and $\ell_0$ have disjoint domains
and the notation $\ell \cupdot \ell_0$ indicates the disjoint union of the two
linking functions. The terms from Figure \ref{fig:first-order-term-syntax}
do not directly modify the quantum state $\ket \psi.$

Figure \ref{fig:quantum-constants-syntax} lists the formation rules for the
term constants that we use to manipulate quantum information. Note that these
constants are values of type $Q(\AAA,\BBB)$ and therefore are classical/non-linear
(and may be used any number of times).
The associated reduction rules for quantum function application using these constants are presented in Figure \ref{fig:quantum-information-operational}.
The term $ U(\mathbf{x_1 \otimes \cdots \otimes x_k)} $ applies the unitary operation $U$ to the qubits identified by the variables $\mathbf{x_1, \ldots, x_k}$ and modifies the quantum state in the configuration accordingly.
The term $\mathbf{new\ ff}$ (resp. $\mathbf{new\ tt}$) creates a new qubit in state $\ket 0$ (resp. $\ket 1$), creates a fresh new qubit variable that points to it and modifies the quantum configuration accordingly.
The term $\mathbf{meas\ \mathbf y}$ measures the qubit identified by variable $\mathbf y$ and produces bit $\mathbf{tt}$ or $\mathbf{ff}$ with probability given by the Born rule of quantum mechanics. This operation \emph{irreversibly} modifies the quantum state and causes a \emph{probabilistic computational effect}.


The terms in Figures \ref{fig:first-order-term-syntax} and 
\ref{fig:quantum-constants-syntax} can be thought of as jointly making up a
subsystem that we call \QPL (it is roughly equivalent to \QPL in
\cite{qpl,qpl-fossacs}), which is a first-order language for quantum
programming.

\begin{figure*}
\small{
\begin{minipage}{0.4\textwidth}
\begin{align*}
\pi_1 (v,w) \probto 1 v \qquad  \pi_2 (v,w) \probto 1 w\qquad\qquad   \unfold\ \fold\ v \probto 1 v \qquad 
(\lambda x. m)v \probto 1 m[v/x] \\
(\text{case}\ \emph{in}_1 v\ \text{of}\ \emph{in}_1 x \Rightarrow n_1\ |\ \emph{in}_2 y \Rightarrow n_2) \probto 1 n_1[v/x] \qquad 
(\text{case}\ \emph{in}_2 v\ \text{of}\ \emph{in}_1 x \Rightarrow n_1\ |\ \emph{in}_2 y \Rightarrow n_2) &\probto 1 n_2[v/y] \\
\end{align*}
\begin{align*}
\mathbf{ [\ket{\psi}, \ell, let\ x \otimes y = v \otimes w\ in\ r ] } &\probto 1 \mathbf{ [\ket{\psi}, \ell, r[v/x, w/y] ] } \\ 
\mathbf{ [\ket{\psi}, \ell, case\ in_1 v\ of\ in_1 x \Rightarrow r_1\ |\ in_2 y \Rightarrow r_2] } &\probto 1 \mathbf{ [\ket{\psi}, \ell, r_1[v/x]] } \\
\mathbf{ [\ket{\psi}, \ell, case\ in_2 v\ of\ in_1 x \Rightarrow r_1\ |\ in_2 y \Rightarrow r_2] } &\probto 1 \mathbf{ [\ket{\psi}, \ell, r_2[v/y]] } \\
[\ket{\psi}, \ell, \mathbf{\bunfold\ \bfold\ v}] &\probto 1 [\ket{\psi}, \ell, \mathbf v] \\
\mathbf{ [\ket{\psi}, \ell, *; r ] } \probto 1 \mathbf{ [\ket{\psi}, \ell, r ] } \qquad  [\ket{\psi}, \ell, \mathbf{let}\  x = \text{lift}\ \mathbf v\ \mathbf{ in\ r} ] &\probto 1 \mathbf{ [\ket{\psi}, \ell, r}[|\vq|/x] ]
\end{align*}
\end{minipage}

\begin{align*}
\scalebox{0.98}{$E$}          & ::= \scalebox{0.98}{$ [\cdot]\ |\ (E, m)\ |\ (v, E)\ |\ \pi_i E\ |\ \emph{in}_i E\ |\ (\text{case}\ E\ \text{of}\ \emph{in}_1 x \Rightarrow n_1\ |\ \emph{in}_2 y \Rightarrow n_2)\ |\ Em\ |\ vE\ |\ \text{(un)fold } E\ $} \\ 
\scalebox{0.98}{$\mathbf E$}  & ::= \scalebox{0.98}{$ [\cdot]\ |\ \mathbf E \otimes \mathbf q \ |\ \mathbf v \otimes \mathbf E\ |\ \mathbf{let\ x \otimes y = E\ in\ r} \ |\ \mathbf{in}_i \mathbf E\ |\ \mathbf{ (case\ E\ of\ in_1 x \Rightarrow r_1\ |\ in_2 y \Rightarrow r_2) }\ |\ \mathbf E; \mathbf q
  \ |\ \textbf{(un)fold } \mathbf E \ |\ $}\\
  &\ \quad\  \scalebox{0.98}{$\mathbf{let}\  x = \text{lift}\ \mathbf E\ \mathbf{ in\ r} $}
\end{align*}
\[
  \begin{bprooftree}
  \def\ScoreOverhang{0.5pt}
  \AxiomC{$m \probto{p} m' $}
  \UnaryInfC{$E[m] \probto{p} E[m'] $}
  \end{bprooftree}
  \quad
  \begin{bprooftree}
  \def\ScoreOverhang{0.5pt}
  \AxiomC{$[\ket{\psi}, \ell, \mathbf q] \probto{p} [\ket{\psi'}, \ell', \mathbf q'] $}
  \AxiomC{$\ell \cap \ell_0 = \varnothing = \ell' \cap \ell_0$}
  \BinaryInfC{$\left[ \ket{\psi}, \ell \cupdot \ell_0, \mathbf {E[q]} \right] \probto{p} \left[ \ket{\psi'}, \ell' \cupdot \ell_0, \mathbf{E[q']} \right] $}
  \end{bprooftree}
\]
\vspace{-2ex}\caption{Classical evaluation contexts $(E)$, quantum evaluation contexts $(\mathbf E)$ and associated reduction rules.}
\label{fig:structural-reduction}
\begin{align*}
& [\ket{\psi}, \ell, U(\mathbf{x_1 \otimes \cdots \otimes x_k) } ] \probto 1 [(\sigma \circ (U \otimes \id) \circ \sigma^{-1})\ket \psi, \ell, \mathbf x_1 \otimes \cdots \otimes \mathbf x_k ], \\
& \qquad \text{ for any permutation $\sigma$, s.t. } \sigma(i) = \ell(\mathbf{x}_i) , 1 \leq i \leq k. \\
& [\ket{\psi}, \varnothing, \mathbf{new\ ff}] \probto 1 [\ket{\psi} \otimes \ket 0, \{\mathbf x \mapsto \text{dim}(\ket{\psi})+1\}, \mathbf x] \quad  [\ket{\psi}, \varnothing, \mathbf{new\ tt}] \probto 1 [\ket{\psi} \otimes \ket 1, \{\mathbf x \mapsto \text{dim}(\ket{\psi})+1\}, \mathbf x] , \\
& \qquad \text{where $\mathbf x$ is chosen fresh.} \\
& \left[\alpha \left(\sum_i \alpha_i \ket{b_i} \otimes \ket 0 \otimes \ket{b'_i} \right) + \beta \left( \sum_i \beta_i \ket{b_i} \otimes \ket 1 \otimes \ket{b'_i} \right), \{\mathbf y \mapsto j\} , \mathbf{meas\ \mathbf y} \right] \probto{|\alpha|^2} \left[ \sum_i \alpha_i \ket{b_i} \otimes \ket{b'_i} , \varnothing, \mathbf{ff} \right] \\
& \left[\alpha \left(\sum_i \alpha_i \ket{b_i} \otimes \ket 0 \otimes \ket{b'_i} \right) + \beta \left( \sum_i \beta_i \ket{b_i} \otimes \ket 1 \otimes \ket{b'_i} \right), \{\mathbf y \mapsto j\} , \mathbf{meas\ \mathbf y} \right] \probto{|\beta|^2}  \left[ \sum_i \beta_i \ket{b_i}  \otimes \ket{b'_i}  , \varnothing, \mathbf{tt} \right] , \\
& \qquad \text{ where dim$(\ket{b_i}) = j-1$, so that the $j$-th qubit is measured.}
\end{align*}
\vspace{-3ex}\caption{Rules for manipulating quantum information.
}
\label{fig:quantum-information-operational}
\vspace{-.1in}
\begin{align*}
  [\ket{\psi}, \ell, (\lambdaq \mathbf{(x_1, \ldots, x_n) . q}) (\mathbf{v_1 \otimes \cdots \otimes v_n}) ] &\probto 1 [\ket{\psi}, \ell, \mathbf q [\mathbf{v_1/x_1, \ldots, v_n / x_n}] ] \\
  \text{run } [\ket \psi, \varnothing, \mathbf v] \probto 1 | \mathbf v|  \qquad\qquad 
 & [\ket{\psi}, \varnothing, \mathbf{init }\ |\mathbf v|] \probto 1 [\ket{\psi}, \varnothing, \mathbf v] 
\end{align*}

\[
  \begin{bprooftree}
  \def\ScoreOverhang{0.5pt}
  \AxiomC{ $m \probto p m' \phantom{[} $}
  \UnaryInfC{$\left[ \ket{\psi}, \ell, m \mathbf q \right] \probto{p} \left[ \ket{\psi}, \ell, m' \mathbf q \right] $}
  \end{bprooftree}
  \begin{bprooftree}
  \def\ScoreOverhang{0.5pt}
  \AxiomC{$\left[ \ket{\psi}, \ell, \mathbf q \right] \probto{p} \left[ \ket{\psi'}, \ell', \mathbf q' \right] $}
  \UnaryInfC{$\left[ \ket{\psi}, \ell, v \mathbf q \right] \probto{p} \left[ \ket{\psi'}, \ell', v \mathbf q' \right] $}
  \end{bprooftree}
  \]
  
  \[
  \begin{bprooftree}
  \def\ScoreOverhang{0.5pt}
  \AxiomC{ $ [\ket{\psi}, \ell, \mathbf q] \probto p  [\ket{\psi'}, \ell', \mathbf q']$}
  \UnaryInfC{ $\text{run } [\ket{\psi}, \ell, \mathbf q] \probto p \text{run } [\ket{\psi'}, \ell', \mathbf q']$}
  \end{bprooftree}
  \begin{bprooftree}
  \def\ScoreOverhang{0.5pt}
  \AxiomC{$ m \probto p m'$}
  \UnaryInfC{ $[\ket \psi, \varnothing, \mathbf{init }\ m] \probto p [\ket \psi, \varnothing, \mathbf{init }\ m' ]$}
  \end{bprooftree}
\]
\vspace{-2ex}\caption{Rules for quantum function application and extracting observable (quantum) information.}
\label{fig:switching-rules-operational}
}
\end{figure*}

\subsubsection{Mixed Classical/Quantum Terms}
Both subsystems $\mathrm{(P)FPC}$ (Figure
\ref{fig:classical-term-syntax}) and \QPL have been studied previously (for
very different purposes). The main distinguishing feature of \VQPL is that it demonstrates
how these subsystems can be combined and used simultaneously for 
variational quantum programming. The terms and formation rules that allow us to
achieve this are presented in Figure \ref{fig:syntax-mixed}, and we now describe them in greater detail.

The term $\lambdaq \mathbf{(x_1, \ldots, x_n) . q }$ is a value which represents a \emph{first-order} quantum lambda abstraction. Note that this value is actually classical (non-linear).
The term $m \qq$ represents \emph{quantum function application}. In our view, this is the most interesting rule in $\mathtt{VQPL}$. This is because its subterm $m : Q(\AAA,\BBB)$ represents a \emph{probabilistic classical program} that eventually reduces to a quantum lambda abstraction. Because of this, the term $m \qq$ combines classical probabilistic computation
with quantum computation (represented by the subterm $\qq$), and this is reflected in the associated reduction rules in Figure \ref{fig:switching-rules-operational}.
Providing a semantic interpretation of this term requires considerable effort and the development of novel semantic methods, as we show later.

The \emph{observable} quantum/classical values are simply quantum/classical
values of observable types with observable context. The closed observable quantum values are in 1-1
correspondence with the closed observable classical values, which is made
precise by the assignment $|-|$ from Figure \ref{fig:syntax-translate}.
Therefore $|-|$ may be seen as a translation, not only between the observable types, but also between their closed values, and so we may think of them as carrying the same information.
In the sequel, we see that this view extends to our denotational interpretation as well.
Observable types and values are important, because they play a special role in the terms we introduce next.

Given any configuration $\mathcal C$ of \emph{observable} type, the term
$``\text{run } \mathcal C"$ reduces the configuration $\mathcal C$ to some value
configuration $[\ket \psi, \varnothing, \vq]$, then extracts the observable quantum
value $\vq$ from it and produces its \emph{classical counterpart} $|\vq|$ as
the result of the overall computation.
Note that in this situation, the observable value $\vq$ does not depend on the quantum data, because it necessarily has an empty quantum context and thus empty linking function, and therefore the remaining quantum state $\ket \psi$ may be safely
discarded (this is consistent with affine approaches to quantum programming, see \cite{qpl,qpl-fossacs,quantum-games1}). 
The term $``\text{run } \mathcal C"$ is
classical, and it allows us to execute quantum algorithms on a quantum
computer and then extract the resulting \emph{observable} information into our
classical subsystem for further manipulation. The entire process is
probabilistic. It is convenient to introduce some syntactic sugar. Given
a closed quantum term $\cdot; \cdot \vdash \qq : \OO$ of observable type, we
can define $\text{run } \qq \eqdef \text{run } [1, \varnothing, \qq]$. Users of
the programming language are not expected to write the more general terms
$``\text{run } \mathcal C"$ (which are useful for formalising the operational
and denotational semantics), but only the sugarised terms $``\text{run } \qq"$.
The term $``\mathbf{init}\ m "$ performs the reverse function of that of "run",
i.e., given a classical (probabilistic) process $m$ of observable type, the
term $``\mathbf{init}\ m "$ prepares observable quantum information as
indicated by the observable value that $m$ reduces to in the end.

Finally, the $``\mathbf{let}\  x = \text{lift}\ \mathbf q\ \mathbf{ in\ r}"$
term allows us to execute a quantum term $\qq$ of observable type and then
promote the observable quantum information it produces to the classical world,
so that we may use it any number of times within the continuation $\rr$ (see
Figure \ref{fig:structural-reduction}). This term therefore implements what is
often called "dynamic lifting" in the quantum programming literature. From
a structural perspective, it is the only term that allows us to modify the
non-linear context of quantum terms and as such may be compared to the
corresponding rules of the LNL calculus \cite{benton-small,benton-wadler}. In
practice, this term is useful for describing quantum processes where we measure
a \emph{part of} our quantum state and use the measurement outcome to
influence the subsequent quantum dynamics. It is necessary for protocols like quantum teleportation.

\subsection{Type Safety}
\label{sub:type-safety}

The next two propositions show $\VQPL$ is type-safe. The first proposition shows type assignment is preserved by reduction,
and as a consequence, totality of quantum configurations also is preserved.

\begin{proposition}[Type Preservation]
  If $\Phi \vdash m \colon P$ and $m \probto{p} n,$ then $\Phi \vdash n \colon P.$
  Likewise, if $\Phi \vdash \mathcal C : \AAA;\qbit^k$ and $\mathcal C \probto{p} \mathcal D$, then $\Phi \vdash \mathcal D : \AAA;\qbit^k.$
  In both of these situations, if $p < 1$, then there exists a term $n'$ (resp. configuration $\mathcal D'$), such that $m \probto{1-p} n'$
  (resp. $\mathcal C \probto{1-p} \mathcal D'$).
\end{proposition}

\begin{proposition}[Progress]
  If $\cdot \vdash m \colon P$, then either $m$ is a value or there exists a classical term $n,$ such that $m \probto{p} n$ for some $p \in (0,1]$.
  Likewise, if $\cdot \vdash \mathcal C \colon \AAA; \qbit^k$, then either $\mathcal C$ is a value configuration or there exists a quantum configuration
  $\mathcal D,$ such that $\mathcal C \probto{p} \mathcal D$ for some $p \in (0,1]$.
\end{proposition}

\begin{remark}
  \label{rem:translation}
  As usual, Progress holds for all \emph{closed} terms/configurations, whereas
  Type Preservation holds for all well-formed terms/configurations, including
  the \emph{open} ones. Here we note that the static semantics is independent
  of the translation $| - |$ on values. This translation only matters for three
  rules in our operational semantics and it is
  restricted to closed observable values (in Figure
  \ref{fig:syntax-translate}), because otherwise Type Preservation would fail
  for \emph{open} terms/configurations.
\end{remark}

\subsection{Recursion and Asymptotic Behaviour of Reduction}

It is well-known that type recursion induces term recursion in $\mathtt{FPC}$  
\cite{fiore-thesis,fpc-syntax,harper-book}, and this also is  true for $\VQPL.$
The call-by-value fixpoint operator
\[\cdot \vdash \mathtt{fix}_{P \to R} \colon \left( \left( P \to R \right) \to P \to R \right) \to P \to R \]
may be \emph{derived} at any function type $P \to R$ (see \cite{fpc-syntax} and \cite[\S 8]{fiore-thesis} for more details).
This fixpoint operator allows us to write recursive functions.

The probability that a term $m$ reduces to a value $v$ (in any number of steps) may be determined as in \cite{quant-semantics}.
The \emph{probability weight} of a reduction path $\pi = \left( m_1  \probto{p_1} \cdots \probto{p_n} m_n \right)$ is $P(\pi) \eqdef \prod_{i=1}^n p_i.$
The probability that $m$ reduces to the value $v$ in \emph{at most} $n$ steps is
\[ P(m \probto{}_{\leq n} v) \defeq \sum_{\pi \in \Paths_{\leq n}(m,v) } P(\pi) , \]
where $\Paths_{\leq n}(m,v)$ is the set of all reduction paths from $m$ to $v$ whose length is at most $n$.
The probability that $m$ reduces to the value $v$ in \emph{any number} of steps is $P(m \probto{}_{*} v) \defeq \sup_i P(m \probto{}_{\leq i} v) . $
Similarly, the probability that a quantum configuration $\mathcal C$ reduces to a value configuration $\mathcal V$ (in any number of steps) is denoted
$P(\mathcal C \probto{}_* \mathcal V)$; it is determined in exactly the same way as above by substituting the notion of term with that of configuration.
Finally, the overall probability that a term $m$ or configuration $\mathcal C$ \emph{terminates} is given by:
\begin{equation}
\label{eq:halt}
  \Halt(m) \eqdef \sum_{v \in \Val(m)} P(m \probto{}_* v) \qquad \qquad
  \Halt(\mathcal C) \eqdef \sum_{\mathcal V \in \ValC(\mathcal C)} P(\mathcal C \probto{}_* \mathcal V) ,
\end{equation}
where 
$\ValC(\mathcal C) \defeq \{ \mathcal V \ |\ \mathcal V \text{ is a value configuration and } P(\mathcal C \probto{}_* \mathcal V) > 0 \}$ and
$\Val(m) \defeq \{ v \ |\ v$ is a value and $P(m \probto{}_* v) > 0 \}$.
Note that both sums may be countably infinite.

\subsection{Examples}

We now  illustrate $\VQPL$ with some example programs.

\begin{example}
  A fair coin toss can be defined by using some simple quantum resources: $ \text{coin}_{0.5} \eqdef \text{run}\ \left( \text{meas}(H(\text{new}\ \mathbf{ff})) \right): \mathtt{Bool} , $
  where $H$ represents the Hadamard unitary operation. More generally, by replacing $H$ with a suitable unitary operation $U_p$, a biased coin toss $\text{coin}_p$ may be defined
  for $p \in [0,1]$ with reduction behaviour $P(\text{coin}_p \probto{}_* \mathtt{ff}) = p$ and $P(\text{coin}_p \probto{}_* \mathtt{tt}) = 1-p.$
  Notice that coin$_p : \texttt{Bool}$ is a \emph{classical term}. This shows that classical discrete probabilistic choice is \emph{derivable} and therefore the language $\mathtt{PFPC}$ \cite{m-monad}
  is a subsystem of $\VQPL.$
\end{example}

\begin{example}
  The fixpoint operator $\mathtt{fix}$ allows us to write classical recursive functions (as usual), but it also enables us to write quantum recursive functions. The function $\mathtt{Ts}$
  below applies the unitary $T$ to each qubit in a list. Defining, for brevity, $\mathtt{Fqs} \defeq Q(\mathbf{List}(\qbit), \mathbf{List}(\qbit)) $, let:
  \[
    \mathtt{Ts'} \eqdef \left( \mathtt{fix}_{1 \to \mathtt{Fqs}} \lambda f^{1 \to \mathtt{Fqs}}. \lambda x^1.\ \lambdaq (\mathbf{qs}). \mathbf{case} (\mathbf{qs})\ \mathbf{of}\
    \mathbf{nil} \Rightarrow \mathbf{nil}\ |\ \mathbf{q::qs'} \Rightarrow T\qq :: (fx) \mathbf{qs'} \right) : 1 \to \mathtt{Fqs},
  \]
  where we used some (hopefully obvious) syntactic sugar for pattern matching of (linear) lists. The recursive call is performed by the $fx$ expression, which is of type $\mathtt{Fqs}.$
  Setting $\mathtt{Ts} \eqdef \mathtt{Ts'}() : Q(\mathbf{List}(\qbit), \mathbf{List}(\qbit))$, we get the desired function.
  An example quantum execution is given by
  $
    P\left( [\ket{111}, \ell, \mathtt{Ts}(\xx_1 :: \xx_2 :: \xx_3 :: \mathbf{nil}) ] \probto{}_*
    [(T \otimes T \otimes T)\ket{111}, \ell, \xx_1 :: \xx_2 :: \xx_3 :: \mathbf{nil} ] \right) = 1,
  $
  where $\ell$ is the linking function defined by $\ell(\xx_i) = i$.
  This shows that classical recursion \emph{induces} recursion on the quantum subsystem and therefore a quantum fixpoint operator is obviously derivable.
\end{example}

\section{Probabilistic Effects and (Commutative) Valuations Monads}
\label{sec:probabilistic}

As we already explained, $\PFPC$ is a subsystem of our language. In this
section we recall the construction of the commutative monad $\MM : \dcpo \to
\dcpo$ of \cite{m-monad}, which we also use in our denotational semantics.
Our classical judgements are
interpreted in the Kleisli category $\KL$ of $\MM$, which provides a sound and (strongly) adequate
model of $\PFPC$ \cite{m-monad}. 

\subsection{Domain-theoretic and Topological Preliminaries}
\label{sub:domain-preliminaries}

If $D$ is a partially ordered set (\emph{poset}), a nonempty subset $A$ of
$D$ is called \emph{directed} if each pair of elements in $A$ has an upper bound in $A$. 
Then $D$ is a \emph{directed complete poset} (\emph{dcpo}, for short) if
each of its directed subsets has a supremum.  
For example, the unit interval $[0, 1]$ is a dcpo in the usual ordering. 
A \emph{Scott-continuous map} $f \colon D \to E$ between (posets) dcpo's
is a function that is monotone and preserves (existing) suprema
of directed subsets. \emph{Pointed posets} have least elements, usually 
denoted by $\bot_{D}$ if $D$ is the ambient poset. 
A Scott-continuous function $f\colon D\to E$ between pointed dcpo's
is \emph{strict} if $f$ preserves the least element, that is, $f(\bot_{D}) = \bot_{E}$. 

The category $\dcpo$  of dcpo's and Scott-continuous functions is complete, cocomplete and 
Cartesian closed~\cite{abramskyjung:domaintheory}.  The categorical (co)product of the dcpo's 
$A_1$ and $A_2$ is denoted $A_1 \times A_2$ ($A_1$ + $A_2$), with $\pi_1,
\pi_2$ ($\emph{in}_1, \emph{in}_2$) the associated (co)projections. 
Initial and terminal objects of $\dcpo$ are denoted by $\varnothing$ and $1$, 
which are the empty dcpo and the singleton dcpo, respectively. For dcpo's $A$ and $B$, 
the internal hom of $A$ and $B$ in $\dcpo$ is $[A \to B]$, 
the space of all Scott-continuous functions $f: A \to B$ ordered pointwise.

The category $\dcpobs$ 
 of \emph{pointed} dcpo's and \emph{strict} Scott-continuous functions is symmetric monoidal closed when equipped with the smash product and strict Scott-continuous function space, and it also is  complete and cocomplete \cite{abramskyjung:domaintheory}.

The \emph{Scott topology}  $\sigma D$ on a dcpo $D$ consists of 
the upper subsets $U = \up U \defeq \{x\in D\mid (\exists u\in U)\, u\leq x\}$ 
that are \emph{inaccessible by directed suprema:} i.e., if  $A\subseteq D$ is directed and 
$\sup A\in U$, then $A\cap U\not=\emptyset$. The topological space $(D, \sigma D)$ is also 
written as $\Sigma D$. Scott-continuous functions between dcpo's $D$ and $E$ are exactly the continuous 
functions between $\Sigma D$ and $\Sigma E$~\cite[Proposition II-2.1]{gierzetal:domains}. We always equip $[0,1]$ with the Scott topology unless stated otherwise.

If $B\subseteq D$ and $D$ is a dcpo, then $B$ is a \emph{sub-dcpo}  if
every directed subset $A\subseteq B$ satisfies $\sup_D A\in B$, where 
$\sup_{D} A$ denotes the supremum of $A$ in $D$. In this case, $B$ is 
a dcpo in the induced order from $D$ and $\sup_{D} A =\sup_{B} A$
for each directed subset $A$ of $B$. 

\subsection{The Monad $\MM$}
\label{sub:monad}

If $X$ is a topological space, then the open set lattice $\mathcal O X$ is a 
complete lattice in the inclusion order, hence a dcpo. A \emph{subprobability valuation} on 
$X$ is a Scott-continuous function $\nu \colon \mathcal O X \to [0,1]$ that is strict ($\nu(\emptyset) = 0$)  
and modular $(\nu(U) + \nu(V) = \nu(U\cup V) + \nu(U\cap V) )$.  The set  $\VV X$ of subprobability valuations 
on $X$  is a dcpo in the \emph{stochastic order} 
defined by: $\nu_{1}\leq \nu_{2}$ if and only if  $\nu_{1}(U)\leq \nu_{2}(U)$ for all $U\in \mathcal OX$, 
for $\nu_{1}, \nu_{2} \in \VV X$, and
the supremum of a directed family of valuations $\{\nu_{i}\}_{i\in I}$ is computed pointwise:  
 $(\sup_{i\in I} \nu_{i})(U) \defeq \sup_{i\in I} \nu_{i}(U)$, for all $U\in \mathcal OX$. The least element of $\VV X$  
is the constantly zero valuation~${\mathbf 0}_{X}$. The dcpo 
$\VV X$ also enjoys a convex structure: if $\nu_{i}\in \VV X$ and 
$r_{i}\geq 0$ for $i = 1, \ldots, n$ and $\sum_{i=1}^{n}r_{i}\leq 1$, then the convex sum 
$\sum_{i=1}^{n}r_{i}\nu_{i}$, defined by $(\sum_{i=1}^{n}r_{i}\nu_{i})(U)   
\defeq \sum_{i=1}^{n}r_{i}\nu_{i}(U)$ for $U\in \mathcal OX$, also is in $\VV X$.

The \emph{Dirac valuations} $\delta_{x}$, for $x\in X$, are defined by $\delta_{x}(U) = 1$ 
if $x\in U$ and $\delta_{x} (U) =0$ otherwise. These are canonical examples of subprobability 
valuations, as are their convex sums, 
which we call \emph{simple valuations}: 
they have form $\sum_{i=1}^{n}r_{i}\delta_{x_{i}}$, where $x_{i}\in X$, $r_{i}\geq 0$, and $\sum_{i=1}^{n}r_{i}\leq 1$.
The simple valuations are denoted $\SSS X$, and $\SSS X\subseteq \VV X$, but $\SSS X$ is not a dcpo in general.

If $f\colon X\to [0,1]$ is a continuous function and $\nu\in \VV X$, the \emph{integral of $f$ against $\nu$} is given by the Choquet formula
\[
\int_{x\in X} f(x) d \nu \defeq \int_{0}^{1} \nu(f^{-1}((t, 1]))dt,
\]
where the right side  is a Riemann integral of the bounded antitone function $t \mapsto \nu(f^{-1}((t, 1]))$. Since $f$ is continuous,  and $(t, 1]\subseteq [0, 1]$ is Scott open for each $t\in [0, 1]$,  $f^{-1}((t, 1])$ is open in $X$, so $\nu(f^{-1}((t, 1]))$ is well defined. Thus the integral makes sense. 
If no confusion can occur, we simply write $\int_{x\in X} f(x) d\nu$ as $\int fd\nu$. 
Note that if $f\colon X\to [0, 1]$ is fixed, then the map $\nu \mapsto \int f d\nu \colon \VV X \to [0, 1]$ is Scott-continuous. Other basic properties of this integral can be found in~\cite{jones90}.

If $D$ is a dcpo, then $\VV D \eqdef \VV\Sigma D = \VV(D,\sigma(D))$ is well-defined, and \cite{jones90} proved that $\VV$ extends to a monad on $\DCPO$.

\subsubsection{Monad Structure}

The \emph{unit} of $\VV$ at $D$ is $\eta^{\VV}_{D}\colon D\to \VV D:: x\mapsto \delta_{x}$.
The \emph{Kleisli extension} $f^{\dagger}$ of a Scott-continuous map $f\colon D\to \VV E$ maps $\nu\in \VV D$ to 
$f^{\dagger}(\nu)\in \VV E$, where for $U\in \sigma E$, $f^{\dagger}(\nu)(U)  \defeq  \int_{x\in D}f(x)(U)d\nu.$
The \emph{multiplication} $\mu^{\VV}_{D}\colon \VV\VV D\to \VV D$ is given by $\mu \eqdef \id_{\VV D}^{\dagger}$.
Thus, $\VV$ defines an endofunctor on $\DCPO$ that sends a dcpo $D$ to $\VV D$, and a Scott-continuous map $h\colon D\to E$ to $\VV(h) \defeq(\eta_{E}\circ h)^{\dagger}$, so $\VV(h)(\nu)(U) =   \nu(h^{-1}(U))$ for $\nu\in \VV D$ and  $U\in \sigma E$. The valuation $h_*(\nu) =  \VV(h)(\nu)$ is called the \emph{push forward of $\nu$ by $h$}.

In fact, $\VV$ defines a strong monad on $\dcpo$~\cite{jones90}: the strength at $(D, E)$ is  
\[\tau^{\VV}_{DE}\colon D \times \VV E \to \VV(D\times E) :: (x, \nu) \mapsto \lambda U. \int_{y\in E} \chi_{U}(x, y)d\nu,\]
where $\chi_{U}$ is the characteristic function of $U\in \sigma(D\times E)$. However, it is unknown whether $\VV$ is a commutative monad on $\DCPO$.
This is equivalent to showing the  
Fubini-style equation 
\begin{equation}\label{eqn:Fub}
\int_{x\in D}\int_{y\in E}\chi_{U}(x, y)d\xi d\nu = \int_{y\in E}\int_{x\in D}\chi_{U}(x, y)d\nu d\xi,
\end{equation}
holds for dcpo's $D$ and $E$, where $U\in \sigma(D\times E)$ and $\nu\in \VV D, \xi\in \VV E$~\cite{JonesP89}.
To address this problem, the authors of~\cite{m-monad} define a subclass of
valuations that simultaneously validates \eqref{eqn:Fub}, contains
all simple valuations, and forms a dcpo in the stochastic order. They prove
their construction defines a commutative valuations monad $\MM$ on $\dcpo$. We
devote the rest of this subsection to describing their construction; for more
details, see~\cite{m-monad}.

\begin{definition}
For a dcpo $D$, $\MM D$ is the intersection of all sub-dcpo's of $\VV D$ containing $\SSS D$. 
\end{definition}

We call the valuations in $\MM D$  the \emph{$\MM$-valuations} on $D$. In fact, $\MM D$ is the smallest sub-dcpo of $\VV D$ containing $\SSS D$.
It follows that the $\MM$-valuations on $D$ consist of the simple valuations on $D$, directed suprema of simple valuations on $D$, directed suprema of directed suprema of simple valuations on $D$ and so forth, transfinitely. 
It is straightforward to show that \eqref{eqn:Fub} holds when $\xi$ and $\nu$ are simple valuations, and because the nested integral operations are Scott-continuous in the valuations components, it follows \eqref{eqn:Fub} holds for $\MM$-valuations.
This is the idea behind the proof of the following theorem.

\begin{theorem}[{\cite[Theorem 8]{m-monad}}]
\label{theorem:M is commutative}
  $\MM$ has the structure of a \emph{commutative} monad on $\DCPO$ when equipped with the (co)restricted monad operations of $\VV$.
\end{theorem}

Since the inclusions $\MM D \subseteq \VV D$ form a strong map of monads, we are justified in viewing $\MM$ as a submonad of $\VV$
and we use the same notation for the monad operations of $\MM$ and $\VV$.



\section{Quantum Effects and Hereditarily Atomic von Neumann Algebras}
\label{sec:operator-algebras}

We now turn our attention to the model for the quantum subsystem of $\VQPL$.
We begin with a short review of operator algebras, which can be used to
study quantum foundations. The standard references are \cite{Blackadar},
\cite{takesaki:oa1} and \cite{kadisonringrose:oa1}. 
For quantum computing, it is sufficient to consider a special class of operator
algebras that is known as the \emph{hereditarily atomic von Neumann algebras}.
This class consists of (possibly infinite) products of finite-dimensional
matrix algebras.  These algebras were studied in \cite{Kornell18}, where it is
shown that the dual category has a concrete description as \emph{quantum sets}.
Our main result in this section is to prove that the category of hereditarily
atomic von Neumann algebras is enriched over continuous domains
(\secref{subsection:enrichment vN}), which is crucial for providing a
semantic interpretation of the "$m\qq$" term.



\subsection{Definition of von Neumann algebras}
\label{sub:w*-def}

If $H$ is a Hilbert space, any linear map $x:H\to H$ is called an \emph{operator}, and $x$ is \emph{bounded} if it is continuous with respect to the norm $\|\cdot\|$ induced by the inner product $\langle\cdot,\cdot\rangle$ on $H$. The space $B(H)$ of all bounded operators on $H$ forms an algebra over $\CN$ with composition as multiplication. Moreover, $B(H)$ has an \emph{involution} $x\mapsto x^*$, where $x^*$ is the unique bounded operator satisfying $\langle x^*k,h\rangle=\langle k,xh\rangle$ for $h,k\in H$. A subalgebra $A\subseteq B(H)$ that is closed under the involution is called a \emph{$*$-subalgebra}. If, in addition, $xy=yx$ for each $x,y\in A$, we call $A$ \emph{commutative}. The  \emph{commutant} of a subset $S\subseteq B(H)$ is $S' = \{ y\in B(H)\mid xy = yx\ (\forall x\in S)\}$. 
\begin{definition}
Let $H$ be a Hilbert space. A \emph{von Neumann algebra} on $H$  is a  $*$-subalgebra $M$ of $B(H)$ such that $M''=M$. If $K$ is another Hilbert space, and $N$ is a von Neumann algebra on $K$, a linear map $\varphi:M\to N$ that preserves the multiplication and the involution is a \emph{$*$-homomorphism}. If, in addition, $\varphi$ is bijective, it is a $*$-isomorphism. 
\end{definition}
The commutant of any non-empty set in $B(H)$ always contains $1_H$, so $1_H\in M$ for any von Neumann algebra $M\subseteq B(H)$. We sometimes write $1_H$ as $1_M$ to emphasize it is the unit of $M$.

\begin{example}
  $B(H)$ itself is a von Neumann algebra, and if $H$ is $n$-dimensional, then $B(H)$ is $*$-isomorphic to $\mathrm{M}_n(\mathbb C)$, the algebra of $n\times n$-complex valued matrices. This example plays an important role in the definition of hereditarily atomic von Neumann algebras below.
\end{example}

\begin{example}\label{ex:ellinfty}
If $X$ is a set, then $\ell^2(X)\defeq \{f:X\to\mathbb C\mid \sum_{x\in X}|f(x)|^2<\infty\}$ is a Hilbert space with inner product $\langle f,g\rangle\defeq\sum_{x\in X}\overline{f(x)}g(x)$. The space $\ell^\infty(X)\defeq\{f:X\to\mathbb C\mid \sup_{x\in X}|f(x)|<\infty\}$ equipped with the norm $\|f\|\defeq\sup_{x\in X}|f(x)|$ can be embedded isometrically into $B(\ell^2(X))$ via the maps $f\mapsto m_f$, where $m_f:\ell^2(X)\to\ell^2(X)$ is the left multiplication $g\mapsto fg$ \cite[Proposition B.73]{landsman}. Thus, identifying $\ell^\infty(X)$ with its image in $B(\ell^2(X))$ shows it is a commutative von Neumann algebra on $\ell^2(X)$ \cite[Proposition B.108]{landsman}.	
\end{example}

Given Hilbert spaces $(H_\alpha)_{\alpha\in\Alpha}$,  the sum $\bigoplus_{\alpha\in\Alpha}H_\alpha \defeq \{ (h_\alpha)_{\alpha\in\Alpha}\in \prod_{\alpha\in\Alpha}H_\alpha\mid \sum_{\alpha\in\Alpha}\|h_\alpha\|^2<\infty\}$ is a Hilbert space with inner product $\langle k,h\rangle \defeq\sum_{\alpha\in\Alpha}\langle k_\alpha,h_\alpha\rangle$ for $k=(k_\alpha)_{\alpha\in\Alpha}$ and $h=(h_\alpha)_{\alpha\in\Alpha}$.

\begin{proposition}\cite[Proposition II.3.3]{takesaki:oa1}\label{prop:product of von Neumann algebras}
	Let $M_\alpha$ be a von Neumann algebra on a Hilbert space $H_\alpha$ for each $\alpha\in A$. Then $\prod_{\alpha\in\Alpha}M_\alpha \defeq \{ (x_\alpha)_{\alpha\in\Alpha}\mid \sup_{\alpha\in\Alpha}\|x_\alpha\|<\infty\}$
	is a von Neumann algebra on $H\defeq\bigoplus_{\alpha\in\Alpha}H_\alpha$, where $xh\defeq (x_\alpha h_\alpha)_{\alpha\in\Alpha}\in H$ for $x=(x_\alpha)_{\alpha\in\Alpha}\in \prod_{\alpha\in\Alpha}M_\alpha$ and $h=(h_\alpha)_{\alpha\in\Alpha}\in H$.
\end{proposition}

\begin{definition}
  We call  a von Neumann algebra  $M$ \emph{hereditarily atomic}, or simply an \emph{HA-algebra}, if $M$  is isomorphic to $\prod_{\alpha\in\Alpha}M_\alpha$, where each $M_\alpha$ is $*$-isomorphic to some matrix algebra.
\end{definition}

In order to make the correspondence between HA-algebras and the types of our language clearer, we overload notation and often write $\bigoplus_{\alpha \in A} M_\alpha \eqdef \prod_{\alpha\in\Alpha}M_\alpha.$

\begin{example}
  All of the following are HA-algebras and we indicate to which type they correspond.
  The complex numbers $\mathbb C$ correspond to type $\mathbf I$; the algebra $\mathrm M_2(\mathbb C)$ corresponds to type $\qbit$; the algebra $\bigoplus_{n \in \mathbb N} \mathbb C$ corresponds to type $\mathbf{QNat} \equiv \mu \mathbf X. \mathbf I \oplus \mathbf X$;
  the algebra $\bigoplus_{n \in \mathbb N} \mathrm{M}_{2^n}(\mathbb C)$ corresponds to type $\mathbf{List}(\qbit) \equiv \mathbf{\mu X. I \oplus (\qbit \otimes X)}.$
  Moreover, if $X$ is a set, then $\ell^\infty(X) \cong \bigoplus_{x\in X}\mathbb C$ is an HA-algebra and we will see it corresponds to \emph{observable} quantum types, in general.
\end{example}

\subsection{Quantum Computation with Hereditarily Atomic von Neumann Algebras}
\label{sub:w*-category-definition}

In this subsection, we define the appropriate notion of morphism that is computationally relevant.

If $M$ is a von Neumann algebra on a Hilbert space $H$, we say $x\in M$ is \emph{self adjoint} if $x^*=x$, and \emph{positive} if $x=y^*y$ for some $y\in M$; equivalently $\langle h,xh\rangle\geq 0$ for each $h\in H$ \cite{kadisonringrose:oa1}. 
Given self-adjoint elements $x$ and $y$ in $M$, we write $x\leq y$ iff $y-x$ is positive. The relation $\leq$ is a partial order on the set $M_{\mathrm{sa}}$ of self-adjoint elements in $M$ under which $M_{\mathrm{sa}}$ is \emph{bounded directed complete:} if $A$ is directed and $\alpha\mapsto x_\alpha\in M_{\mathrm{sa}}$ is a monotone ascending net that is bounded (i.e., $x_\alpha\leq  y\in M_{\mathrm{sa}}\forall \alpha\in \Alpha$), then  
$(x_\alpha)_{\alpha\in A}$ has a supremum $\sup_{\alpha\in \Alpha} x_\alpha= x\in M_{\mathrm{sa}}$. The partial order $\leq$ is often called the \emph{L\"owner order}. 

In fact, $x$  is the limit of $(x_\alpha)$ in the \emph{strong operator topology} on $M$: i.e., $\lim_{\alpha\in\Alpha} x_\alpha h= xh$ for each $h\in H$; this also implies convergence with respect to the \emph{weak operator topology} on $M$, i.e., $\lim_{\alpha\in\Alpha}\langle k,x_\alpha h\rangle=\langle k,xh\rangle$ for each $h,k\in H$ \cite[Proposition I.3.2.5 \& Corollary I.3.2.6]{Blackadar}. As a consequence, the \emph{unit interval of $M$}, $[0,1]_M = \{ x\in M_{\mathrm{sa}}\mid 0\leq x\leq 1_M\}$ is a dcpo. 

A linear function $\varphi:M\to N$   between von Neumann algebras is  \emph{unital} if $\varphi(1_M)=1_N$, \emph{subunital} if $\varphi(1_M)\leq 1_N$; $\varphi$ is \emph{positive} if it preserves positive elements, equivalently, if $\varphi$ is monotone with respect to $\leq$. A  positive and subunital $\varphi\colon M\to N$ restricts to a monotone map $[0,1]_M\to[0,1]_N$, which by linearity completely determines $\varphi$. We call $\varphi$ \emph{normal} if it preserves the suprema of bounded increasing nets, i.e., if it is Scott continuous with respect to $\leq$. 

We denote by $\mathrm{M}_n(M)$ the von Neumann algebra  of all $n\times n$-matrices with entries in $M$. Any linear map $\varphi:M\to N$ between von Neumann algebras induces a linear map $\varphi^{(n)}:\mathrm{M}_n(M)\to\mathrm{M}_n(N)$ obtained by applying $\varphi$ entrywise. We say that $\varphi$ is \emph{completely positive} if $\varphi^{(n)}$ is positive for each $n\in\mathbb N$. In particular, any completely positive map is positive.
Finally, we say that a linear map $\varphi : M \to N$ is an \emph{NCPSU} map, if $\varphi$ is normal completely positive and subunital.
We note that every normal unital $*$-homomorphism is an NCPSU map, but the converse is not true, in general.

\begin{definition}
	We denote the category of von Neumann algebras and NCPSU maps by $\vN$.
	Its full-on-objects subcategory having normal unital $*$-homomorphisms as morphisms is denoted by $\vNs$.
  The category of HA-algebras and NCPSU maps is denoted by $\HA$ and we denote the full-on-objects subcategory of $\HA$ with normal unital $*$-homomorphisms by $\HAs$.
  The categories relevant for our semantics are their formal duals given by $\QQ \eqdef (\HA)^\op$ and $\QQs \eqdef (\HAs)^\op$.
\end{definition}

\begin{remark}
  When working with von Neumann algebras, it is customary to adopt the
  Heisenberg picture of quantum mechanics, rather than the Schr{\"o}dinger one.
  This corresponds to working in the category $\QQ$ which is the formal dual of
  $\HA.$ In fact, the program of non-commutative geometry is based
  on dualities between categories of operator algebras and ``formal dual''
  categories. Furthermore, this approach is also established in
  quantum programming semantics \cite{qpl-fossacs,quantum-von-neumann}
  and necessary for the appropriate categorical structure (\secref{sub:quantum-submodel}).
\end{remark}

From now on, all $*$-homomorphisms we work with are implicitly assumed to be unital and normal.
We interpret quantum values in $\QQs$ and quantum terms in $\QQ$. Next, we describe maps between HA-algebras that are crucial for quantum computation and that we use in our semantics.

\begin{figure}
\small{
\centering
$
\begin{array}{l|l|l|l}
    \mathrm{tr} : \mat n(\mathbb C) \to \mathbb C \ & \ \mathrm{state}_\rho : \mathbb C \to \mat{2^n}(\mathbb C) \ & \ \mathrm{meas} : \mat 2(\mathbb C) \to \mathbb C \oplus \mathbb C  \ & \ \mathrm{unitary}_U : \mat{2^n}(\mathbb C) \to \mat{2^n}(\mathbb C) \\
    \mathrm{tr} :: x \mapsto \sum_i x_{i,i} \ &
\ \mathrm{state}_\rho :: \alpha \mapsto \alpha \rho \ & \ \mathrm{meas} :: \begin{pmatrix} \alpha & \beta \\ \gamma & \delta \end{pmatrix} \mapsto \begin{pmatrix} \alpha & \delta \end{pmatrix} \ & \ \mathrm{unitary}_U :: x \mapsto UxU^* \\
    \mathrm{tr}^* : \mathbb C \to \mat n(\mathbb C) \ &
\ \mathrm {state}_\rho^* : \mat{2^n}(\mathbb C) \to \mathbb C \ & \ \mathrm{meas}^* : \mathbb C \oplus \mathbb C  \to \mat 2(\mathbb C) \ & \ \mathrm{unitary}^*_U : \mat{2^n}(\mathbb C) \to \mat{2^n}(\mathbb C) \\
    \mathrm{tr}^* :: \alpha \mapsto \alpha 1_{\matrixAlg n} \ &
\ \mathrm{state}_\rho^* :: x \mapsto \trace{x \rho} \ & \ \mathrm{meas}^* :: \begin{pmatrix} \alpha & \delta \end{pmatrix} \mapsto \begin{pmatrix} \alpha & 0 \\ 0 & \delta \end{pmatrix} & \ \mathrm{unitary}_U^* :: x \mapsto U^* xU
\end{array}
$
}
\caption{Maps in the Schr{\"o}dinger picture $(\varphi \colon M \to N)$ and their Hermitian adjoints $(\varphi^* \colon N \to M)$.}
\label{fig:quantum-pictures}
\end{figure}

The maps in the upper half of Figure \ref{fig:quantum-pictures} describe NCPSU
maps between HA-algebras that are well-known in the quantum computing
literature.  The map "$\mathrm{tr}$" computes the trace of a matrix; the map
"$\mathrm{state}_\rho$" prepares a new (mixed) quantum state that is described by
the density matrix $\rho$; the map "$\mathrm{meas}$" performs a destructive
quantum measurement on a qubit and returns a bit as outcome; the map
"$\mathrm{unitary}_U$" applies the unitary matrix $U$ of arity $n$ to an
$n$-dimensional quantum state.  These are the appropriate maps to take in the
Schr{\"o}dinger picture of quantum mechanics, but as explained above, the
Heisenberg picture is more appropriate for our denotational semantics, so it is
the Hermitian adjoints of these maps (bottom half of Figure
\ref{fig:quantum-pictures}) that are relevant to us. By writing $\varphi^\ddagger
\eqdef (\varphi^*)^\op$ for $\varphi \in \{\mathrm{tr}, \mathrm{state}_\rho,
\mathrm{meas}, \mathrm{unitary}_U \}$, these maps are then morphisms of $\QQ.$
In particular, $\mathrm{meas}^\ddagger : \mathrm{M}_2(\mathbb C) \to \mathbb C
\oplus \mathbb C$ is the morphism of $\QQ$ which represents quantum
measurement.

We also define a morphism $\mathrm{new}^\ddagger \defeq \mathrm{meas}^\op \in \QQ(\mathbb C \oplus \mathbb C, \mathrm{M}_2(\mathbb C))$.
The way to think of this map (in $\QQ$) is that given a bit $i \in \{0,1\}$, the map would prepare the density matrix $\ket i \bra i$.
Indeed, notice that (in $\QQ$)
we have $\mathrm{meas}^\ddagger \circ \mathrm{new}^\ddagger = (\mathrm{meas} \circ \mathrm{meas}^*)^\op = \id_{\mathbb C \oplus \mathbb C}$, as expected.


\subsection{Continuous Domain Enrichment of $\QQ$}
\label{subsection:enrichment vN}
The category $\vN$ is enriched over $\DCPO_{\perp!}$~\cite{cho:semantics}, where $\varphi\leq \psi$  iff $\psi-\varphi$ is completely positive,   for $\varphi,\psi\in \vN(M,N)$.
The $\DCPO_{\perp!}$-enrichment of $\vN$ immediately implies that of $\HA$ and $\QQ$. 

While $\DCPO_{\perp!}$-enrichment is important, it is insufficient for our purposes. In particular,
the crucial connection between the semantics of the quantum and classical 
probabilistic effects in our language requires $\QQ(A,B)$ to be an $\MM$-algebra for $A, B$
objects of $\QQ$ (Theorem~\ref{thm:M-alg}). We can show this is the case when $\QQ$ is
enriched in a much stronger sense, namely, when $\QQ$ is enriched over continuous dcpo's (see Theorem \ref{thm:emcategory}).
Explaining all this is our next goal.

We begin with continuous dcpo's. For $x, y$ in a dcpo 
$D$, $x$ is \emph{way-below} $y$ (in symbols, $x\ll y$)
if and only if for every directed set $A$ with $y\leq \sup A$, there is some
$a\in A$ such that $x\leq a$. A dcpo $D$ is \emph{continuous}, or simply a
\emph{domain}, if every element $x\in D$ is the supremum of a directed set of
elements that are way-below $x$. We use $\DOM$ to denote the category of domains 
and Scott-continuous maps. 

It was noted in \cite[Example 2.7]{selinger:higher-order} that the unit interval 
of $\mathrm{M}({\mathbb C}^n)$ is a continuous dcpo, from which it is easy to 
show $[0,1]_A$ is a domain for every HA-algebra $A$; in fact, this is an if-and-only-if~\cite{furber-continuousdcpos}. 
We conclude this section with a much stronger result.

\begin{theorem}\label{thm:continuous enrichment}
  The category $\QQ$ is enriched over $\DOM.$
\end{theorem}
\begin{proof} The proof (Appendix \ref{app:proofs-vN-algebras}) starts with the fact that $[0,1]_A$ is a domain for each HA-algebra $A$ and then makes extensive use of the representation theory of von Neumann algebras.
\end{proof}

\section{Probabilistic Effects, Quantum Effects and Kegelspitzen}
\label{sec:relationship}

Our language shows that quantum effects induce probabilistic effects on the
classical side (via the "run" term) and, vice-versa, probabilistic effects on
the classical side can also influence the quantum dynamics (via the
"$\textbf{init}$" and "$m\textbf{q}$" terms). In this section, we describe the
mathematical structure we use to interpret this correspondence.

In particular, we show there is a strong relationship between the Kleisli
category $\KL$ of $\MM$ (where we interpret classical programs) and our
category $\QQ$ of hereditarily atomic von Neumann algebras  (where we interpret
quantum programs). The link between the two categories is provided by the
theory of (continuous) Kegelspitzen \cite{keimelplotkin17}. The relationship we
identify is crucial for the interpretation of the mixed classical/quantum
judgements of Figure \ref{fig:syntax-mixed},  and it links the classical theory of
valuations monads to the quantum theory of von Neumann algebras.

\subsection{(Continuous) Kegelspitzen}
\label{sub:kegelspitzen}

We begin by recalling the definition of Kegelspitzen \cite{keimelplotkin17}.

\begin{definition}
  A Kegelspitze is a dcpo equipped with a convex structure. More precisely:\\[1ex]
$\bullet$\quad  
A \emph{barycentric algebra} is a set $A$ endowed with binary operations $(a,b)\mapsto a+_{r}b\colon A\times A\to A$ indexed by $r\in [0, 1]$ such that for all $a, b, c \in A$ and $r, p\in [0, 1]$, the following equations hold:
\[ a+_{1} b = a;\quad a+_{r}b = b+_{1-r}a;\quad a+_{r} a = a; \quad (a+_{p}b)+_{r}c = a+_{pr}(b+_{\frac{r-pr}{1-pr}} c)~~\text{provided}~r, p<1. \]
$\bullet$\quad  A \emph{pointed barycentric algebra} is a barycentric algebra~$A$ with a distinguished element~$\bot$. 
For $a\in A$ and $r\in [0, 1]$, we define scalar multiplication $r\cdot a \defeq a+_{r} \bot$. 
A map $f\colon A\to B$ between pointed barycentric algebras is \emph{linear} if $f(\bot_{A}) = \bot_{B}$ and $f(a+_{r}b) = f(a)+_{r}f(b)$ for all $a, b\in A, r\in [0, 1]$. \\[1ex]
$\bullet$\quad  A \emph{Kegelspitze} is a pointed barycentric algebra $K$ equipped with a directed-complete partial 
order such that scalar multiplication $(r, a)\mapsto r\cdot a \colon [0,1]\times K\to K$ and the binary operation 
$(a, b)\mapsto a+_{r}b\colon K\times K\to K$, for $r\in [0, 1]$, 
are Scott-continuous (in both arguments). A \emph{continuous Kegelspitze} is a Kegelspitze 
that is a domain in the equipped order. 
\end{definition}

\begin{example}
\label{exa:MDiskegel}
For each dcpo $D$, $\MM D$ is Kegelspitze: for $\nu_{i}\in \MM D$ and $r_{i}\geq 0, i=1,\ldots, n$ with $\sum_{i=1}^{n}r_{i}\leq 1$, the convex sum $\sum_{i=1}^{n}r_{i}\nu_{i}$ is again in $\MM D$. 
Then, if $\nu_{1}, \nu_{2} \in \MM D$ and $r\in [0, 1]$, we define $\nu_{1} +_{r} \nu_{2}\defeq r \nu_{1} + (1-r ) \nu_{2}$. The zero valuation $\mathbf 0_{D}$ is the distinguished least element. 
If, in addition, $D$ is a domain, then $\MM D = \VV D$ is a continuous Kegelspitze \cite{m-monad}.
For each Scott-continuous map $f\colon D\to E$, the map $\MM(f) \colon \MM D\to \MM E $ is Scott-continuous and linear. 
\end{example}

\begin{example}
\label{exa:unit-interval}
The real unit interval $[0,1]$ is obviously a continuous Kegelspitze. More generally, the unit interval $[0,1]_A$ of any von Neumann algebra $A$ is a Kegelspitze.
If $A$ also is hereditarily atomic, then $[0,1]_A$ is a continuous Kegelspitze by \cite[Example 2.7]{selinger:higher-order}.
Moreover, any NCPSU map $f : A \to B$ between von Neumann algebras $A$ and $B$, is Scott-continuous and linear when (co)restricted to the unit intervals of $A$ and $B$.
\end{example}

\subsection{Correspondence between Observable Quantum/Probabilistic Effects}

Our next result describes a bijective correspondence between \emph{observable} quantum/probabilistic effects 
that allows us to interpret the terms dealing with observable primitives.
A semantic observation (which we make precise later) shows that:
any quantum observable type $\OO$ is interpreted as a commutative HA-algebra that is $*$-isomorphic to $\ell^\infty(X)$ for some set $X$;
its classical observable counterpart $|\OO|$ is interpreted as the \emph{discrete domain} with underlying set $X$.
Moreover, quantum values correspond to $*$-homomorphisms and classical values to Dirac valuations.

\begin{theorem}
\label{thm:r-iso-categorical}
  Let $X$ be an arbitrary set. Then, there exists an isomorphism of (continuous) Kegelspitzen
  $r_{X} \colon \QQ(\mathbb C, \ell^\infty(X)) \cong \MM (X, \sqsubseteq) \colon r_{X}^{-1}$, where $\sqsubseteq$ is the discrete order on $X$.
  Furthermore, this isomorphism restricts to a 1-1 correspondence between the $*$-homomorphisms of $\QQ(\mathbb C, \ell^\infty(X))$ and the Dirac valuations of $\MM (X, \sqsubseteq)$.
\end{theorem}
\begin{proof}
See Appendix \ref{app:r-isomorphism}.
\end{proof}

Combined with the above semantic observation, this theorem shows there is a 1-1
correspondence between the quantum and classical probabilistic states of
\emph{observable} types, and also a 1-1 correspondence between the
interpretations of quantum and classical observable values.  This isomorphism
is used for the interpretations of the "run" and "\textbf{init}" terms.

Next, we construct an isomorphism that we use for the interpretation of dynamic lifting (the "\textbf{lift}" term).
This is similar to a construction first reported in \cite{ewire}.

\begin{proposition}
\label{prop:lift-categorical}
Given a dcpo $X$, HA-algebras $A,B$, and a \emph{discrete} dcpo $Y$, there exists a Scott-continuous and linear bijection
$ \widehat{(-)} \colon \dcpo(X \times Y, \QQ(A, B)) \cong  \dcpo(X, \QQ(\ell^\infty(Y) \otimes A, B) ) , $ natural in all components. 
\end{proposition}
\begin{proof}
See Appendix \ref{app:lift-isomorphism}.
\end{proof}


\subsection{Combining Probabilistic and Quantum Effects}
\label{sub:quantum-probabilistic-combination}

In the previous subsection we considered \emph{observable} effects. In the
present subsection, we show how to combine \emph{arbitrary} quantum and
probabilistic effects \emph{into} quantum ones.

We begin by noting that on any Kegelspitze, the binary operations $a +_r b$ generalize to convex sums. We then use these convex sums in order to define \emph{barycentre} maps.
\begin{definition}
\label{def:convex-sums}
In each pointed barycentric algebra~$K$, given $a_{i}\in K, r_{i}\in [0, 1], i=1,\ldots, n$ with $\sum_{i=1}^{n}r_{i}\leq 1$, we inductively define the convex sum by
\[\sum_{i=1}^{n}r_{i}a_{i} \defeq \begin{cases}
a_{1} & \text{, if } r_{1}=1,\\
a_{1}+_{r_{1}}(\sum_{i=2}^{n} \frac{r_{i}}{1-r_{1}}a_{i}) &  \text{, if }r_{1}< 1.
						\end{cases}
\]
This sum  is invariant under index-permutation: for $\pi$ a permutation of $\{1, \ldots, n\}$, $\sum_{i=1}^{n}r_{i}a_{i} = \sum_{i=1}^{n}r_{\pi(i)}a_{\pi(i)}$ \cite[Lemma 5.6]{jones90}. If $K$ is a Kegelspitze, then the expression $\sum_{i=1}^{n}r_{i}a_{i}$ is Scott-continuous in each $r_{i}$ and $a_{i}$. 
A \emph{countable} convex sum also can be defined: if $a_i \in K$ and $r_i \in [0,1]$, for $i \in \Na$, with $\sum_{i \in\Na} r_i \leq 1$, define
$ \sum_{i \in \Na} r_i a_i \defeq \sup \{ \sum_{j \in J} r_j a_j \ |\ J \subseteq \Na \text{ and } J \text{ is finite} \} . $
\end{definition}

\begin{definition}
\label{def:barycentre}
Let $K$ be a Kegelspitze and $s=\sum_{i=1}^{n}r_{i}\delta_{x_{i}}$ be a simple valuation on~$K$. The \emph{barycentre} of $s$ is defined as $\beta_*(s) \defeq \sum_{i=1}^{n}r_{i}x_{i}$. Furthermore, if $K$ is a continuous Kegelspitze and $\nu \in \MM K$, the \emph{barycentre} of $\nu$ is defined as $\beta(\nu) \defeq \sup\{\beta_{*}(s)\mid s \in \SSS K \text{ and }s\ll \nu\}$.
\end{definition}

When $K$ is a \emph{continuous} Kegelspitze, the barycentre map $\beta \colon \MM K\to K :: \nu \mapsto \beta(\nu)$ is well-defined, unique, Scott-continuous and linear \cite{m-monad}. We emphasise that \emph{continuity} is crucial for establishing this
and it is unclear if this holds otherwise. Moreover, the following also is true.

\begin{theorem}[{\cite{m-monad}}]
\label{thm:emcategory}
The Eilenberg-Moore category $\DOM^{\MM}$ of $\MM$ over $\DOM$ is isomorphic to the category of continuous Kegelspitzen and Scott-continuous linear maps. In particular:
\begin{enumerate}
\item Each \emph{continuous} Kegelspitze $K$ admits a linear barycentre map $\beta\colon \MM K\to K$ (as in Definition~\ref{def:barycentre}) for which the pair $(K, \beta)$ is an Eilenberg-Moore algebra of $\MM$ over $\DOM$.
\item Conversely, on each $\MM$-algebra $(K, \beta)$ on $\DOM$, define $a+_{r}b \defeq \beta (\delta_{a} +_{r} \delta_{b})$ for $a, b\in K$ and $r\in[0,1]$.  Then with the operations~$+_{r}$,  $K$ is a continuous Kegelspitze and $\beta\colon \MM K\to K$ is linear. 
\end{enumerate}
\end{theorem}

In \secref{sec:operator-algebras} we saw that $\QQ$ is enriched over $\DOM$. We now further strengthen that result.

\begin{theorem}
\label{thm:q-kegelspitze}
The category $\QQ$ is enriched over \emph{continuous} Kegelspitzen in the following sense: for all objects $A$, $B$ in $\QQ$, the homset $\QQ(A,B)$ is a \emph{continuous} Kegelspitze,  and for any morphism $\varphi:A\to B$ in $\QQ$ and any object $C$ in $\QQ$, the following maps are Scott-continuous and linear: 
\[\QQ(C,\varphi):\QQ(C,A)\to \QQ(C,B) :: \psi\mapsto \varphi\circ\psi\quad \text{and}\quad \QQ(\varphi,C):\QQ(B,C)\to \QQ(A,C) :: \psi\mapsto\psi\circ\varphi . \] 
\end{theorem}
\begin{proof}
  See Appendix \ref{app:proofs-vN-algebras}.
\end{proof}

Combining these two theorems gives the main result of this section.

\begin{theorem}\label{thm:M-alg}
  For any HA-algebras $A$ and $B$, there exists a (unique) Scott-continuous and linear barycentre map $\beta \colon \MM \QQ(A,B) \to \QQ(A,B)$ that is also an Eilenberg-Moore algebra of $\MM$.
\end{theorem}

The above properties of $\beta$ are exactly what is needed to interpret the "$m\qq$" term from Figure \ref{fig:syntax-mixed}, which allows us to combine
classical probabilistic computation with quantum computation.

\section{Categorical Model}
\label{sec:categorical-model}

In this section we organise the relevant mathematical data into several
categories that we later use to describe our denotational semantics.
A diagrammatic summary is provided in Figure \ref{fig:categorical-model} ($\secref{sec:semantics}$).

\subsection{The Kleisli Category of $\MM$}
\label{sub:model}

This subsection provides a summary of the development in~\cite{m-monad} of the Kleisli category of the monad 
$\MM : \dcpo \to \dcpo$, which we denote  $\KL$.
In order to distinguish between the categorical primitives of $\DCPO$ and $\KL$,
we adopt the notation of~\cite{m-monad}, indicating  the morphisms of $\KL$ by $f: A \kto B$, and
using $f \kcirc g \eqdef \mu \circ \MM(f) \circ g$ to denote the Kleisli composition of morphisms in $\KL$
(where $\mu$ is the multiplication of $\MM$).
We write $\kid_A : A \kto A$ with $\kid_A = \eta_A  : A \to \MM A$ for the identity morphisms in $\KL.$
The adjunction $\JJ \dashv \mathcal U : \KL \to \dcpo$ that factorises $\MM$ is determined by the assignments:
\begin{align*}
\JJ A \defeq A , \quad  \JJ f \defeq \eta \circ f  , \quad
\UU A \defeq \MM A , \quad \UU f \defeq \mu \circ \MM f .
\end{align*}
\subsubsection{Coproducts} $\KL$ inherits (small) coproducts from $\dcpo$ in the standard way
\cite[pp. 264]{jacobs-coalgebra} and we write $A_1 \kplus A_2 \eqdef A_1 + A_2$ for the induced
(binary) coproduct. The induced coprojections are given by $\JJ(\emph{in}_1) \colon A_1 \kto A_1 \kplus A_2$
and $\JJ(\emph{in}_2) \colon A_2 \kto A_1 \kplus A_2.$
Then for $f\colon A\kto C$ and $g\colon B\kto D$, $f\kplus g = [\MM(\emph{in}_{C}) \circ f, \MM(\emph{in}_{D}) \circ g]$
and the functor $\JJ$ strictly preserves coproducts.

\subsubsection{Symmetric monoidal structure}
Because $\MM$ is \emph{commutative}, it induces a canonical
symmetric monoidal structure on $\KL$ making $\JJ$ a strict monoidal functor \cite{premonoidal}.
The induced tensor product is $A \ktimes B \eqdef A \times B$ with
Kleisli projections 
$\JJ(\pi_A) : A \ktimes B \kto A$ and $\JJ(\pi_B) : A \ktimes B \kto B$; but these projections do \emph{not}
satisfy the universal property of a product.
The tensor product of $f\colon A\kto C$ and $g\colon B\kto D$ is denoted by
$f\ktimes g$ and it is defined as usual.
It follows that Kleisli products distribute over Kleisli coproducts and we write $d_{A,B,C} : A \ktimes (B \kplus C) \cong (A \ktimes B) \kplus (A \ktimes C)$ for this natural isomorphism.

\subsubsection{Kleisli Exponential}
The adjunction $\JJ \dashv \mathcal U$ also contains the structure of a \emph{Kleisli-exponential}.
Following \cite{moggi-monads}, we use this to interpret higher-order function types.

For each dcpo $B$, we use $[B \kto -] \defeq [B \to \UU(-)] : \KL \to \DCPO$ to  denote the right adjoint of the 
functor $J(-) \ktimes B : \DCPO \to \KL$. Therefore, on objects, $[B \kto C] = [B \to \MM C].$ This determines a
family of Scott-continuous bijections
$
\lambda : \KL(\JJ A \ktimes B, C) \cong \DCPO(A, [B \kto C]) ,
$
natural in $A$ and $C$, often called \emph{currying}. We also denote the counit of these adjunctions
by $\epsilon : \JJ [B \kto -] \ktimes B \naturalto \Id$, which is often called \emph{evaluation}.
Since this family of adjunctions is parameterised by objects $B$ of $\KL$, standard categorical results \cite[\S IV.7]{maclane} 
imply the assignment $[B \kto -] : \KL \to \DCPO$ can be extended uniquely to a bifunctor
$ [- \kto -] : \KL^\op \times \KL \to \DCPO , $
such that $\lambda$ is natural in all three components.

\subsubsection{Enrichment Structure}
The Kleisli category $\DCPO_{\MM}$ is enriched over $\dcpobs:$ for dcpo's~$A, B$ and $C$, the Kleisli
exponential $[A\kto B] = [A\to \MM B] = \KL(A,B)$ is a pointed dcpo and the Kleisli composition 
$
\kcirc \colon [A\kto B]\times [B\kto C] \to [A\kto C]
$
is strict and Scott continuous. Furthermore, the adjunction $\JJ \dashv \UU$ also is $\DCPO$-enriched, as are the bifunctors $(- \ktimes -),  (- \kplus -) $ and $[- \kto - ] .$

The category $\KL$ also has a convex structure:
for each dcpo~$B$, $\MM B$ is a Kegelspitze in the stochastic order by Example~\ref{exa:MDiskegel}, from which 
it follows that $[A\kto B] = \KL(A,B)$ also is a Kegelspitze in the pointwise order.
This convex structure is preserved by Kleisli composition~$\kcirc$, Kleisli coproduct $\kplus$ and Kleisli product $\ktimes$ \cite[Lemma 38]{m-monad}.

\subsubsection{The Subcategories $\TD$ and $\PD$}
\label{subsub:subcategories}
We identify two important subcategories of $\KL$: one for the interpretation of classical values ($\TD$) and one for solving recursive domain equations ($\PD$).

\begin{definition}
\label{def:T}
The subcategory $\TD$ of \emph{deterministic total maps} is the full-on-objects subcategory of $\KL$ whose morphisms $f \colon X \kto Y$  admit a factorisation $f = \JJ(f') = \eta_y \circ f',$ for some $f'$.
\end{definition}
Each map $f: X \kto Y$ in $\TD$ satisfies
$f(x) = \delta_y$ for some $y \in Y$, by definition. We call such maps \emph{deterministic} because 
they carry no interesting convex structure, and they are \emph{total} in that they map all inputs 
$x \in X$ to non-zero valuations. $\TD$ is important because all \emph{classical values} of
our language are interpreted in $\TD$. In fact, $\DCPO \cong \TD$ \cite[Proposition 40]{m-monad}.

The canonical copy map at an object $A$ in our model is given by the map $\JJ
\langle \id_A, \id_A \rangle \colon A  \kto A \ktimes A$; likewise, the canonical
discarding map at $A$ is the map $\JJ(1_A) \colon A \kto 1,$ where $1_A \colon
A \to 1$ is the terminal map of $\dcpo$.  Because maps in $\TD$ are in the
image of $\JJ$, they are compatible with the copy and discard maps, and hence 
also with weakening and contraction \cite{benton-small}.


\begin{definition}
\label{def:D}
The subcategory of \emph{deterministic partial maps}, denoted $\PD$, is the full-on-objects subcategory of $\KL$ each of whose morphisms $f \colon X \kto Y$ admits a factorisation $f = \left( X \xrightarrow{f'} Y_\perp \xrightarrow{\phi_Y} \MM Y \right),$
where $Y_\perp$ is the dcpo obtained from $Y$ by freely adding a least element $\perp$, and where $\phi_Y$ is the map
$
\phi_Y \colon Y_\perp \to \MM Y :: y \mapsto
  \begin{cases}
    {\mathbf 0}_{Y}        & \text{, if } y = \perp \\
    \delta_y & \text{, if } y \neq \perp 
  \end{cases} .
$
\end{definition}
These maps are  \emph{partial} because they map some inputs to $\mathbf 0$; they also are deterministic, because the convex structure is trivial in both cases.
This is justified by the fact that  $\PD \cong \dcpo_{\mathcal T} \cong \dcpobs,$ where $\dcpo_{\mathcal T}$ is the Kleisli category of the lift monad $\mathcal T : \DCPO \to \DCPO$ \cite{m-monad}.

\subsubsection{Solving Recursive Domain Equations}
\label{subsub:domain-equations}
The standard method for interpreting recursive types is 
to construct \emph{parameterised initial algebras} \cite{fiore-thesis,fiore-plotkin}.
We employ this approach in $\PD$ using the limit-colimit coincidence theorem \cite{smyth-plotkin:domain-equations}.

\begin{definition}[see {\cite[\S 6.1]{fiore-thesis}}]
\label{def:initial-algebra}
  Given a category $\CC$ and a functor $\TTT \colon \CC^{n+1} \to \CC,$ a \emph{parameterised initial algebra}
  for $\TTT$ is a pair $(\TTT^\sharp, \iota^\TTT),$ such that:
  \begin{itemize}
    \item $\TTT^\sharp \colon \CC^n \to \CC$ is a functor;
    \item $\iota^\TTT \colon \TTT \circ \langle \Id, \TTT^\sharp \rangle \naturalto \TTT^\sharp : \CC^n \to \CC$ is a natural transformation;
    \item For every $\vec C \in \Ob(\CC^n)$, the pair $(\TTT^\sharp \vec C, \iota^\TTT_{\vec C})$ is an initial $\TTT(\vec C, -)$-algebra.
  \end{itemize}
\end{definition}

The usual notion of an initial algebra arises in the case that $n=1$.

\begin{proposition}[see {\cite[\S 4.3]{lnl-fpc-lmcs}}]
\label{prop:par-initial-algebra}
Let $\CC$ be a category with an initial object and all $\omega$-colimits, and let $\TTT \colon \CC^{n+1} \to \CC$ be an $\omega$-cocontinuous functor. Then $\TTT$ has a  parameterised initial algebra $(\TTT^\sharp, \iota^\TTT)$ and the functor $\TTT^{\sharp} \colon \CC^n \to \CC $ is also $\omega$-cocontinuous.
\end{proposition}

In fact, the subcategory $\PD$ has sufficient structure to solve recursive domain equations, because it is \emph{$\DCPO$-algebraically compact} \cite{m-monad}.
Therefore, every $\dcpo$-enriched \emph{covariant} functor on $\KL$ that
restricts to $\PD$ has a parameterised initial algebra (whose inverse is a parameterised final coalgebra).
Solving equations 
involving \emph{mixed-variance} functors (induced by function types)
can be done using the limit-colimit coincidence theorem
\cite{smyth-plotkin:domain-equations}. An important observation
made in \cite{smyth-plotkin:domain-equations} is that
all type expressions (including function spaces) can be interpreted 
as covariant functors on \emph{subcategories of embeddings}.
For more details on this, see \cite{icfp19,lnl-fpc-lmcs}; here we
also follow this approach.

\begin{definition}
If  $\CC$ is a $\DCPO$-enriched category, a morphism $e \colon X \to Y$ is an \emph{embedding} if there exists a
(necessarily unique)  \emph{projection} $e^p \colon Y \to X$, i.e., a morphism satisfying $e^p \circ e = \id_X$ and $e \circ e^p \leq \id_Y.$ 
$\CC_e$ denotes the full-on-objects subcategory of $\CC$ whose morphisms are the embeddings.
\end{definition}

\begin{proposition}[{\cite[Proposition 47]{m-monad}}]
\label{prop:omega-functors}
The category $\PD_e$ has an initial object and all $\omega$-colimits, and the assignments:
\begin{align*}
  & \ktimes_e \colon \PD_e \times \PD_e \to \PD_e \text{ defined by } X \ktimes_e Y \defeq X \ktimes Y\ \text{and}\ e_1 \ktimes_e e_2 \defeq e_1 \ktimes e_2  \\
  & \kplus_e \colon \PD_e \times \PD_e \to \PD_e \text{ defined by } X \kplus_e Y \defeq X \kplus Y\ \text{and}\ e_1 \kplus_e e_2 \defeq e_1 \kplus e_2  \\
  & [\kto]_e^\JJ \colon \PD_e \times \PD_e \to \PD_e \text{ defined by } [X \kto Y]_e^\JJ \defeq \JJ [ X \kto Y]\ \text{and}\ [e_1 \kto e_2]_e^\JJ \defeq \JJ [e_1^p \kto e_2] 
\end{align*}
are \emph{covariant} $\omega$-cocontinuous bifunctors on $\PD_e$.
\end{proposition}

Thus Propositions \ref{prop:par-initial-algebra} and \ref{prop:omega-functors} show
we can solve recursive domain equations induced by all well-formed type expressions
within $\PD_e$, notably with no restrictions on the admissible logical polarities of the types.
However, our classical judgements support weakening and contraction, so we have
an extra proof obligation: proving each isomorphism that is a solution
to a recursive domain equation can be copied and discarded.
This is true, because every isomorphism of $\PD$ (and $\PD_e$) also are isomorphisms of $\TD$ \cite[Propoistion 48]{m-monad}.

\subsection{The Quantum Category $\QQ$}
\label{sub:quantum-submodel}

We now describe the categorical structure of $\QQ$ and its subcategory $\QQs$. We interpret quantum terms in $\QQ$ and quantum values in $\QQs$.

\subsubsection{Coproducts}
Proposition \ref{prop:product of von Neumann algebras} describes the categorical product on $\vNs$, which restricts to a categorical product on $\HAs$ since the product of hereditarily atomic von Neumann algebras clearly is hereditarily atomic.
Moreover, the product on $\HAs$ extends to a product on $\HA$. As a consequence, $\QQ$ and $\QQs$ have small coproducts and we write $A \oplus B$ to denote the coproduct in both categories.
We write the coprojections as $\mathbf{in}_1 \colon A \to A \oplus B$ and $\mathbf{in}_2 \colon B \to A \oplus B.$ Note that the initial object $\mathbf 0$ of $\QQs$, given by the 1 element HA-algebra, is a zero object in $\QQ$ (but not in $\QQs$).

\subsubsection{Symmetric Monoidal Structure}
Given two von Neumann algebras $M$ and $N$ on Hilbert spaces $H$ and $K$, respectively, the algebraic tensor product $M\odot N$ acts in a natural way on the Hilbert space tensor product $H\otimes K$. The weak operator closure of $M\odot N$ in $B(H\otimes K)$ is a von Neumann algebra, usually denoted $M\bar{\otimes} N$, and called the \emph{spatial tensor product} of $M$ and $N$. The construction in \cite[III.2.2.5]{Blackadar} shows the spatial tensor product of von Neumann algebras induces a symmetric monoidal product on both $\HAs$ and $\HA$, hence on $\QQs$ and $\QQ$. We write $A \otimes B$ for the tensor product in both $\QQ$ and $\QQs$.
Moreover, $\QQs$ is symmetric monoidal closed \cite[Theorem 9.1]{Kornell18} and therefore there exists a natural isomorphism $\mathbf d_{A,B,C} : A \otimes (B \oplus C) \cong (A \otimes B) \oplus (A \otimes C)$.

\subsubsection{Adjunctions}
The subcategory inclusion $\II:\QQs\to\QQ$ corresponds to an embedding $\HAs\to\HA$ that is shown to have a left adjoint in \cite[Section 4.3.4]{Westerbaan-thesis}. Therefore $\II$ has a right adjoint.
Moreover, the adjunction between $\QQs$ and $\QQ$ is Kleislian \cite{Westerbaan-thesis} and the subcategory inclusion $\II:\QQs\to\QQ$ is a strict monoidal functor that strictly preserves coproducts.

The assignment $\ell^\infty(-)$ extends to a functor $\ell^\infty:\Set\to\HA^\op_*$ whose action on functions $f:X\to Y$ between sets is a normal $*$-homomorphism $\ell^\infty(f):\ell^\infty(Y)\to\ell^\infty(X) :: k\mapsto k\circ f$.
Hence we obtain a functor $\ell^\infty:\Set\to\QQ_*$, which is fully faithful; its essential image is the full subcategory of $\QQ_*$ consisting of all commutative hereditarily atomic von Neumann algebras.

\subsubsection{Affine Structure}

The monoidal unit $\mathbb C$ is \emph{initial} in $\HA_*$ and therefore it is \emph{terminal} in $\QQ_*$, but the same is not true for $\HA$ and $\QQ$.
The terminal map of $\QQs$ at $\mathrm{M}_n(\mathbb C)$ is actually $\mathrm{tr}^\ddagger$ (see \secref{sub:w*-category-definition}).
This allows us to define suitable discarding maps. The map $\mathrm{drop_k}^\ddagger \eqdef (x \mapsto x \otimes 1)^\op \in \QQ( A \otimes \matrixAlg{2^k},  A)$
should be thought of (in $\QQ$) as discarding $k$ auxiliary qubits; this map is used for the interpretation of the "run" term when we execute a non-total quantum configuration
that has $k$ auxiliary qubits (which may be safely discarded at the end of the computation). Indeed, notice that in the category $\QQ$, we have $\mathrm{drop}_k^\ddagger \circ (\id_A \otimes \mathrm{state}_\rho^\ddagger) = \id_A,$
for any density matrix $\rho$ in $\matrixAlg{2^k}$, as one would expect (for brevity, we implicitly suppress the isomorphism $A \otimes \mathbb C \cong A$).

\subsubsection{Solving Recursive Domain Equations}

We now show that the category $\QQs$ has sufficient structure to construct parameterised initial algebras for polynomial functors. On the quantum side, this covers all recursive domain equations that have to be solved.

\begin{proposition}
\label{prop:quantum-omega-functors}
The category $\QQs$ is cocomplete and the functors $\otimes \colon \QQs \times \QQs \to \QQs$ and $\oplus \colon \QQs \times \QQs \to \QQs$ are cocontinuous.
\end{proposition}
\begin{proof}
	Cocompleteness of $\QQs$ is shown in \cite[Proposition 8.6]{Kornell18}; the coproduct bifunctor $\oplus$ is obviously cocontinuous; the functor $\otimes$ is cocontinuous because $\QQs$ is monoidal closed. 
\end{proof}

\section{Denotational Semantics}
\label{sec:semantics}

\begin{figure}
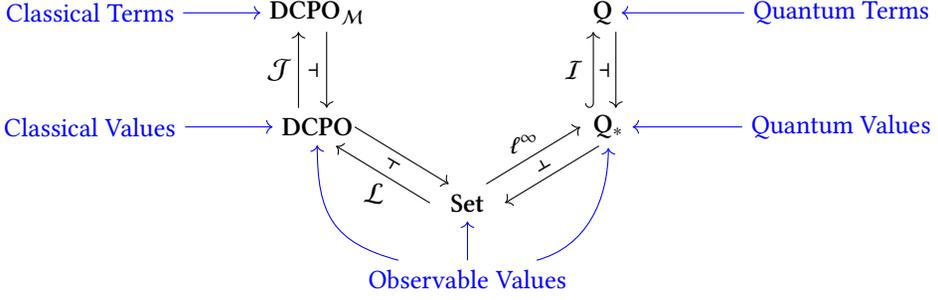

  \centering
  \stikz{model.tikz}
  \caption{Overview of the categorical model.}
  \label{fig:categorical-model}
\end{figure}

We now give the denotational semantics of our language. In Figure
\ref{fig:categorical-model}, we summarise the overall structure of the
interpretation. The blue arrows show where the indicated programming
primitives can be interpreted. Every such primitive may also be interpreted in
a category above it by following the corresponding left adjoint. All depicted
categories are symmetric monoidal and so are the adjunctions between them. The
left adjoints $\mathcal L$ and $\ell^\infty$ are both fully faithful and this
allows us to interpret the \emph{observable} values in the bottom three
categories and also to coherently relate these interpretations.

\subsection{Interpretation of Types}
\label{sub:type-interpretation}

We begin with the interpretation of (open) types, which is described in Figure \ref{fig:type-interpretation}.
Open quantum types are interpreted as functors $\qlrb{\Theta \vdash A} \colon\ \QQ_*^{|\mathbf \Theta|} \to \QQ_*$ and
open classical types are interpreted as functors $\lrb{\Theta \vdash P} \colon\ \PDe^{|\Theta|} \to \PDe$.
Given closed types $\cdot \vdash \mathbf A$ and $\cdot \vdash P$, we write $\qlrb A \defeq \qlrb{\cdot \vdash A}(*) \in \Ob(\QQ_*) = \Ob(\QQ)$ and $\lrb P \defeq \lrb{\cdot \vdash P}(*) \in \Ob(\PDe) = \Ob(\DCPO).$

\begin{figure*}
\small{
\centering
\begin{align*}
& \qlrb{\Theta \vdash A} \colon\ \QQ_*^{|\Theta|} \to \QQ_* \quad
\lrb{\mathbf \Theta \vdash \mathbf \Theta_i} \defeq \Pi_i \quad
\qlrb{\Theta \vdash I} \defeq\ K_{\mathbb C}  \quad
\qlrb{\Theta \vdash \qbit} \defeq\ K_{\mathrm{M}_2(\mathbb C)} \quad
\qlrb{\Theta \vdash \mu X.A} \defeq \qlrb{\Theta, X \vdash A}^\sharp \\
& \qlrb{\Theta \vdash A \oplus B } \defeq \oplus \circ \langle \qlrb{\Theta \vdash A}, \qlrb{\Theta \vdash B} \rangle \qquad
\qlrb{\Theta \vdash A \otimes B } \defeq \otimes \circ \langle \qlrb{\Theta \vdash A}, \qlrb{\Theta \vdash B} \rangle\\[1ex]
& \lrb{\Theta \vdash P} \colon\ \PDe^{|\Theta|} \to \PDe \qquad
\lrb{\Theta \vdash \Theta_i} \defeq \Pi_i \qquad
\lrb{\Theta \vdash 1} \defeq\ K_1  \qquad
\lrb{\Theta \vdash \mu X.P} \defeq \lrb{\Theta, X \vdash P}^\sharp \\
& \lrb{\Theta \vdash P + R } \defeq\ \kplus_e \circ\ \langle \lrb{\Theta \vdash P}, \lrb{\Theta \vdash R} \rangle \qquad
\lrb{\Theta \vdash P \times R } \defeq\ \ktimes_e \circ\ \langle \lrb{\Theta \vdash P}, \lrb{\Theta \vdash R} \rangle \\
& \lrb{\Theta \vdash P \to R } \defeq [\kto]_e^\JJ \circ \langle \lrb{\Theta \vdash P}, \lrb{\Theta \vdash R} \rangle \qquad
\lrb{\Theta \vdash Q(\mathbf A, \mathbf B) } \defeq K_{\QQ( \qlrb{A}, \qlrb{B})}
\end{align*}
\vspace{-4ex}\caption{Interpretation of types. $K_X$ is the constant-$X$-functor.}
\label{fig:type-interpretation}
\[
\begin{array}{l l l l l}
 \qlrb{I} = \mathbb C                   & \qlrb{A \otimes B} = \qlrb A \otimes \qlrb B          & \qlrb{A \oplus B} = \qlrb A \oplus \qlrb B              & \lrb{\mu \mathbf X. \AAA} \cong  \lrb{\AAA[\mu \mathbf X. \AAA/ \mathbf X]}  \\[1ex]
 \lrb{1} = 1                            & \lrb{P \times R} = \lrb P \times \lrb R               & \lrb{P+R} = \lrb P + \lrb R                             & \lrb{\mu X. P} \cong  \lrb{P[\mu X. P/ X]}                                   \\[1ex]
 \qlrb{\qbit} = \mathrm{M}_2(\mathbb C) & \lrb{P \to R} = \left[ \lrb P \to \MM \lrb R \right]  & \lrb{Q(\mathbf A, \mathbf B)} = \QQ(\qlrb A, \qlrb B)   & 
\end{array}
\]
\vspace{-2ex}\caption{Derived equations for closed types.}
\label{fig:derived-type-semantics}\smallbreak
\hspace{-.35in}\begin{minipage}{.4\textwidth}
\[
\begin{array}{l}
\lrb{\Phi \vdash m : P} : \lrb \Phi \kto \lrb P \text{ in $\KL$} \\
\lrb{\Phi, x:P \vdash x: P} \eqdef \JJ \pi_2 \\
\lrb{\Phi \vdash (): 1} \eqdef \JJ 1 \\
\lrb{\Phi \vdash (m, n) : P \times R} \eqdef (\lrb m \ktimes \lrb n) \kcirc \JJ \langle \id, \id \rangle \\
\lrb{\Phi \vdash \pi_i m : P_i} \eqdef \JJ \pi_i \kcirc \lrb m   , \text{ for } i \in \{ 1, 2\} \\
\lrb{\Phi \vdash \emph{in}_{i} m : P_1+P_2} \eqdef \JJ \emph{in}_i \kcirc \lrb{m}, \text{ for } i \in \{ 1, 2\} \\
\lrb{\Phi \vdash ( \text{case}\ m\ \text{of}\ \emph{in}_1 x \Rightarrow n_1\ |\ \emph{in}_2 y \Rightarrow n_2 ) : R}\\
\quad \eqdef [\lrb{n_1}, \lrb{n_2}] \kcirc d \kcirc (\id \ktimes \lrb m) \kcirc \JJ \langle \id, \id \rangle \\
\lrb{\Phi \vdash \lambda x^P . m : P \to R} \eqdef \JJ \lambda(\lrb m) \\
\lrb{\Phi \vdash mn : R} \eqdef \epsilon \kcirc (\lrb m \ktimes \lrb n) \kcirc \JJ \langle \id, \id \rangle  \\
\lrb{\Phi \vdash  \lambdaq (\xx_1, \ldots, \xx_n) . \qq : Q( \AAA_1 \otimes \cdots \otimes \AAA_n, \BBB) }\\
\quad \eqdef \JJ \qlrb q  \\
\lrb{\Phi \vdash \text{new} : Q(\bit, \qbit)} \eqdef \JJ(x \mapsto \text{new}^\ddagger) \\
\lrb{\Phi \vdash U : Q(\qbit^{\otimes n}, \qbit^{\otimes n})} \eqdef \JJ(x \mapsto \text{unitary}_U^\ddagger) \\
\lrb{\Phi \vdash \text{meas} : Q(\qbit, \bit)} \eqdef \JJ(x \mapsto \text{meas}^\ddagger)  \\
\lrb{\Phi \vdash \fold{}\ m: \mu X. P} \eqdef \sfold \kcirc \lrb m  \\
\lrb{\Phi \vdash \unfold\ m : P[\mu X. P / X]} \eqdef \sunfold \kcirc \lrb m  \\
\lrb{\Phi \vdash \text{run } \mathcal C : |\mathbf O|} \eqdef x \mapsto r ( \mathrm{drop}_k^\ddagger \circ \lrb{\mathcal C}_x)
\end{array}
\]
\vspace{-3ex}\caption{Interpretation of classical term judgements.}
\label{fig:classical-term-semantics}
\end{minipage}
\qquad\qquad\begin{minipage}{.4\textwidth}
\[
\begin{array}{l}
\lrb{\Phi; \mathbf{ \Gamma \vdash q : A}} \colon \lrb \Phi \to \QQ(\lrb{\mathbf \Gamma}, \lrb{\mathbf A}) \text{ in $\dcpo$} \\[.5ex]
\lrb{\Phi; \mathbf{ y: A \vdash y: A}} \eqdef x \mapsto \id_{\qlrb A} \\
\lrb{\Phi; \cdot \vdash *: \mathbf I } \eqdef x \mapsto \id_{\mathbb C} \\[1ex] 
\lrb{\Phi; \mathbf{ \Gamma_1, \Gamma_2 \vdash q ; r : \mathbf A}} \eqdef x \mapsto (\cong \circ (\qlrb q_x \otimes \qlrb r_x)) \\
\lrb{\Phi; \mathbf{ \Gamma_1, \Gamma_2 \vdash q \otimes r : A \otimes B} } \eqdef x \mapsto \qlrb{q}_x \otimes \qlrb{r}_x \\[1ex]
\lrb{\Phi; \mathbf{ \Gamma_1, \Gamma_2 \vdash let\ x \otimes y = q\ in\ r : B }}\\
\quad \eqdef x \mapsto \qlrb{r}_x \circ ( \id_{\qlrb{\Gamma_2}} \otimes \qlrb{q}_x ) \circ \text{swap}  \\[.5ex]
\lrb{\Phi; \BGamma \vdash \mathbf{in}_{i}\ \qq : \AAA \oplus \BBB}  \eqdef x \mapsto \mathbf{in}_i \circ \qlrb{q}_x \\[1ex]
\lrb{\Phi; \mathbf{ \Gamma_1, \Gamma_2 \vdash (case\ q\ of\ in_1 x \Rightarrow r_1\ |\ in_2 y \Rightarrow r_2) : B}}\\
\quad \eqdef x \mapsto [\qlrb{r_1}_x, \qlrb{r_2}_x] \circ \mathbf d \circ (\id \otimes \qlrb{q}_x) \circ \text{swap} \\[.5ex]
\lrb{ \Phi; \mathbf \Gamma \vdash m \mathbf q : \mathbf B } \eqdef x \mapsto \beta(\lrb m_x) \circ  \qlrb q_x \\
\lrb{\Phi; \mathbf{ \Gamma \vdash \bfold\ q : \mu X. A } } \eqdef x \mapsto \qfold \circ \qlrb{q}_x \\
\lrb{\Phi; \mathbf{ \Gamma \vdash \bunfold\ q : A[\mu X. A / X]}}\\
\quad \eqdef x \mapsto \qunfold \circ \qlrb q_x \\[.5ex]
\lrb{ \Phi; \cdot \vdash \mathbf{init}\ m : \mathbf O } \eqdef x \mapsto r^{-1} ( \lrb{m}_x ) \\
\lrb{ \Phi; \mathbf{ \Gamma_1, \Gamma_2} \vdash \mathbf{let}\  x = \text{lift}\ \mathbf q\ \mathbf{ in\ r} : \mathbf A }\\
\quad \eqdef x \mapsto \widehat{\qlrb r}_x \circ (\qlrb q_x \otimes \id_{\qlrb{\Gamma_2}}) \\
\end{array}
\]
\vspace{-2ex}\caption{Interpretation of quantum term judgements.}
\label{fig:quantum-term-semantics}
\end{minipage}\smallbreak
\[
  \lrb{\Phi \vdash [\ket \psi, \ell, \qq] : \mathbf A ; \qbit^k} \colon \lrb{\Phi} \to \QQ(\mathbb C, \lrb \AAA \otimes \lrb{\qbit^k}) :: x \mapsto (\qlrb q_x \otimes \id_{\lrb{\qbit^{k}}}) \circ \sigma_\ell \circ \text{state}_{\ket \psi \bra \psi }^\ddagger
\]
\vspace{-4ex}\caption{Interpretation of quantum configurations.}
\label{fig:quantum-configuration}
\begin{align*}
\qfold_{\mathbf{\mu X.A}} &: \qlrb{A[\mu X. A/ X]}  = \qlrb{X \vdash A} \qlrb{\mu X. A}   \cong \qlrb{\mu X. A} : \qunfold_{\mathbf{\mu X.A}} \\
\sfold_{\mu X. P} &: \lrb{P[\mu X. P/ X]}  = \lrb{X \vdash P} \lrb{\mu X. P}   \cong \lrb{\mu X. P} : \sunfold_{\mu X. P}
\end{align*}
\vspace{-3ex}\caption{Definition of the folding/unfolding isomorphisms.}
\label{fig:folding}
}
\end{figure*}

\begin{proposition}
The assignments $\qlrb{\Theta \vdash A} \colon\ \QQ_*^{|\mathbf \Theta|} \to \QQ_*$ and $\lrb{\Theta \vdash P} \colon\ \PDe^{|\Theta|} \to \PDe$ are well-defined $\omega$-cocontinuous functors.
\end{proposition}
\begin{proof}
By induction using Propositions \ref{prop:par-initial-algebra}, \ref{prop:omega-functors} and \ref{prop:quantum-omega-functors}.
\end{proof}

\begin{lemma}[Substitution]
\label{lem:type-substitution}
Given quantum types $\mathbf{\Theta, X \vdash A}$ and $\mathbf{\Theta \vdash B}$ and classical types $\Theta, X \vdash P$ and $\Theta \vdash R$, then:
\begin{align*}
\qlrb{\Theta \vdash A[B/X]} = \qlrb{\Theta, X \vdash A} \circ \langle \Id, \qlrb{\Theta \vdash B} \rangle \quad \text{and} \quad
\lrb{\Theta \vdash P[R/X]} = \lrb{\Theta, X \vdash P} \circ \langle \Id, \lrb{\Theta \vdash R} \rangle.
\end{align*}
\end{lemma}

For closed recursive types, the folding/unfolding isomorphisms are defined in Figure \ref{fig:folding},
where the equalities are from Lemma~\ref{lem:type-substitution} and the unnamed isomorphisms are the initial algebras.
Note that $\sfold_{\mu X. P}$ and $\qfold_{\mathbf{\mu X.A}}$ are isomorphisms in $\TD$ and $\QQ_*$, respectively.
Now, the derived equations in Figure \ref{fig:derived-type-semantics} follow immediately.

\subsubsection{Relationship Between Observable Types}
Quantum/classical observable types play a special role in our language and they also satisfy a special denotational relationship that we now describe.

\begin{proposition}
  \label{prop:observable-types}
  Let $\OO$ be a closed quantum observable type with $|\OO|$ its classical counterpart. Then there exists a canonical \emph{set} $\flrb{\OO}$, defined by induction on the derivation of $\cdot \vdash \OO$, such that: (1) $\lrb{\OO} \cong \ell^\infty \flrb{\OO}$ in $\QQ_*$; and (2) $\lrb{|\OO|} = \mathcal L \flrb{\OO},$
  where $\mathcal L : \Set \to \DCPO$ is the obvious left adjoint functor that equips a set $X$ with the discrete order.
\end{proposition}
\begin{proof}
  Both of these statements follow as special cases of the abstract categorical semantics in \cite[Section 6]{lnl-fpc-lmcs,icfp19}, where the source category is $\Set$ and the target one is $\QQ_*$ and $\TD$, respectively (note that $\TD \cong \DCPO,$ so $\DCPO$ can also be taken).
\end{proof}

\begin{remark}
\label{rem:identification}
The coherence conditions outlined in \cite{lnl-fpc-lmcs,icfp19} are very strong and the functor $\ell^\infty : \Set \to \QQ_*$ is a fully faithful strong symmetric monoidal left adjoint (just like $\mathcal L)$.
To avoid notational overhead, we  treat the  $*$-isomorphism in (1) above as an equality.
\end{remark}

Proposition \ref{prop:observable-types} shows that the interpretation of any classical observable type is a \emph{discrete domain} and the set $\flrb{\OO}$ is simply its underlying set.
We may safely extend the action of the functor $\ell^\infty$ to discrete dcpo's and then by the above remark it follows that observable types are related by the following strong relationship:
\begin{equation}
  \label{eq:observable-types}
  \lrb{\OO} = \ell^\infty(\lrb{|\OO|} ) .
\end{equation}

This shows the interpretation $\lrb \OO$ of a quantum observable type $\OO$ is a \emph{commutative} HA-algebra.

\subsection{Interpretation of Terms and Quantum Configurations}
\label{sub:term-interpretation}
A classical context $\Phi = x_1 \colon P_1, \ldots, x_n \colon P_n$ is interpreted as the dcpo $\lrb \Phi \eqdef \lrb{P_1} \times \cdots \times \lrb{P_n},$
and a quantum context $\mathbf{ \Gamma = x_1 \colon A_1, \ldots, x_n \colon A_n}$ is interpreted as the HA algebra $\mathbf{ \lrb{\Gamma} \defeq \lrb{A_1} \otimes \cdots \otimes \lrb{A_n} }$.
The interpretation of classical/quantum term judgements and quantum configurations is defined by mutual induction in Figures \ref{fig:classical-term-semantics}, \ref{fig:quantum-term-semantics} and \ref{fig:quantum-configuration}. Next, we explain some of the notation used therein.

The interpretation of a classical term judgement $\Phi \vdash m \colon P$ is a morphism $\lrb{\Phi \vdash m \colon P} \colon \lrb \Phi \kto \lrb P$ in $\KL$ 
that we often abbreviate by writing $\lrb m$.
Likewise, a quantum term judgement $\Phi; \mathbf{ \Gamma \vdash q : A}$ is interpreted as
a morphism $\lrb{\Phi; \mathbf{ \Gamma \vdash q : A}} \colon \lrb \Phi \to \QQ(\lrb{\mathbf \Gamma}, \lrb{\mathbf A})$ in $\dcpo$
that we often abbreviate by writing $\qlrb{q}$.
For an element $x \in \lrb \Phi$, we also write $\qlrb q_x$ and $\lrb m_x$ as a shorthand for $\qlrb q(x)$ and $\lrb m(x)$, respectively.
In the special case that $\Phi = \cdot$, we can also regard $\qlrb q$ as a morphism $\qlrb q \colon \qlrb \Gamma \to \qlrb A$ in $\QQ.$
The unnamed isomorphism in Figure \ref{fig:quantum-term-semantics} is the left monoidal unitor $\mathbb C \otimes A \cong A$ and "swap" is the monoidal symmetry in $\QQ$.

The interpretation of a configuration $\Phi \vdash [\ket \psi, \ell, \qq] : \mathbf A ; \qbit^k$ is given by a morphism in $\DCPO$ of type $\lrb{\Phi \vdash [\ket \psi, \ell, \qq] : \mathbf A; \qbit^k} \colon \lrb \Phi \to \QQ(\mathbb C, \qlrb A \otimes \lrb{\qbit^{k}})$ defined in Figure \ref{fig:quantum-configuration}.
The notation $\sigma_\ell$ used there denotes any permutation $\sigma_\ell : \lrb{\qbit^n} \to \lrb{\qbit^n}$ that maps the $i$-th component to $\ell(\xx_i)$, defined in full analogy to $\sigma$ from Figure \ref{fig:quantum-information-operational}, and where $n = \dim(\ket \psi).$
In the special case that $\mathcal C = [\ket \psi, \ell, \qq]$ is total, its interpretation can be seen as a morphism $ \lrb {\mathcal C} \colon \lrb \Phi \to \QQ(\mathbb C, \qlrb A)$ in $\DCPO$, given by
$\lrb{\mathcal C} = \left( x \mapsto \qlrb q_x \circ \sigma_\ell  \circ \text{state}_{\ket \psi \bra \psi}^\ddagger \right).$ If the linking function $\ell$ is given by $\ell(\xx_i) = i$ (which we may usually assume), then its interpretation is equivalently given by
$\lrb{\mathcal C} = \left( x \mapsto \qlrb q_x \circ \text{state}_{\ket \psi \bra \psi}^\ddagger \right).$
If, in addition, $\mathcal C$ is closed, i.e., $\Phi = \cdot,$ then we can regard it as a map $\lrb{\mathcal C} \colon \mathbb C \to \qlrb A$ in $\QQ$, which is just a state in $\QQ$ (and also a state in the operator-algebraic sense).

We now comment on the terms that are of primary interest to us. The interpretation of the "new", "meas" and "U" terms is determined by the constant function on the appropriate $\QQ$-morphism from \secref{sub:w*-category-definition} which is then injected via $\JJ$ into $\TD$. For quantum lambda abstractions, $\lrb{\qq}$ is, by construction,
a Scott-continuous function $\lrb \qq \colon \lrb \Phi \to \QQ(\lrb{\mathbf \Gamma}, \lrb{\mathbf A})$, so we may see it as a morphism of $\TD$ via the $\JJ: \DCPO \cong \TD$ isomorphism. For the interpretation of the "run" term, let us consider the special case when $\mathcal C$ is total. Then the semantics is equivalently given by the Kleisli morphism
$\lrb{\mathrm{run}\ \mathcal C} = \left(x \mapsto r(\lrb{\mathcal C}_x) \right)$, where $r$ is the isomorphism from Theorem \ref{thm:r-iso-categorical}. When $\mathcal C$ is not total, the $\mathrm{drop}_k^\ddagger$ morphism is used to get rid of the remaining auxiliary qubits in accordance with affine principles. The interpretation of the "$\mathbf{init}$"
term is done by simply taking the inverse isomorphism $r^{-1}$. Dynamic lifting is interpreted using the natural bijection $\widehat{(-)}$ from Proposition \ref{prop:lift-categorical}.

Finally, the interpretation of the "$m\qq$" term, representing quantum function application, makes use of the barycentre maps from Theorem \ref{thm:M-alg}.
In our view, this is the term of highest interest discussed here and which required the most effort to interpret. Notice that its interpretation is unique in that it \emph{combines} two different
Kegelspitzen structures living in two different categories.

\subsection{Interpretation of (Observable) Values}
\label{sub:values-interpretation}

The interpretation of values in our language enjoys additional structural properties, as usual.

\begin{proposition}
For any classical value $\Phi \vdash v : P$ and quantum value $\Phi ; \mathbf{\Gamma \vdash v : A}$, we have:
\begin{enumerate}
  \item $\lrb v \colon \lrb \Phi \kto \lrb P$  also is a morphism of $\TD$. Equivalently, it is in the image of $\JJ$.
  \item $\qlrb v \colon \lrb \Phi \to \QQ(\qlrb \Gamma, \qlrb A)$ corestricts to $\QQ_*(\qlrb \Gamma, \qlrb A)$. That is, $\forall x \in \lrb \Phi. \qlrb v_x \in \QQ_*(\qlrb \Gamma, \qlrb A)$.
\end{enumerate}
\end{proposition}

This means that $\lrb{v}_x$ is a Dirac valuation and $\qlrb{v}_x$ is a $*$-homomorphism, for any $x \in \lrb \Phi$. Note that the functor $\JJ$ restricts to an isomorphism of categories $\JJ \colon \DCPO \cong \TD$ \cite{m-monad}, so for a classical value $v$, we can define $\flrb v \eqdef \JJ^{-1} \lrb v : \lrb \Phi \to \lrb P$ in $\DCPO.$

The interpretation of classical/quantum \emph{observable} values enjoys even stronger structural properties and they also are strongly related to each other, as we show next.
If $\Phi \vdash v \colon P$ is an \emph{observable} value, then $\lrb \Phi$ and $\lrb P$ are discrete, so we can safely regard $\flrb v : \lrb \Phi \to \lrb P$ as a morphism in $\Set.$ 
For the next proposition, we identify $\mathbb C$ with $\ell^\infty(1)$ (see Remark \ref{rem:identification}).

\begin{proposition}
  \label{prop:observable-value-interpretation}
  Let $\cdot; \cdot \vdash \vq : \OO$ be an observable quantum value and let $\cdot \vdash |\vq| : |\OO|$ be its classical counterpart.
  Then $\lrb{\vq} = \ell^\infty\flrb{|\vq|}$ and $\lrb{|\vq|} = \JJ \flrb{|\vq|}.$
  Furthermore, $r (\lrb{\vq}) = \lrb{|\vq|}_*.$
\end{proposition}
\begin{proof}
  By combining Proposition \ref{prop:observable-types} and \cite[Proposition 6.15]{lnl-fpc-lmcs}. The final statement follows by Theorem \ref{thm:r-iso-categorical} and by construction of the isomorphism $r$, which relies on the fact that the functors $\ell^\infty$ and $\JJ$ (and implicitly $\mathcal L$) are fully faithful.
\end{proof}

The above proposition is used for establishing soundness for the "run" and
$"\mathbf{init}"$ terms.  The natural isomorphism $\widehat{(-)}$ from
Proposition \ref{prop:lift-categorical} satisfies an additional coherence
condition (see Appendix \ref{app:lift-isomorphism}) w.r.t $\ell^\infty$, which, 
combined with Proposition \ref{prop:observable-value-interpretation},
is used for the soundness proof of the $``\mathbf{lift}"$ term.
Also, by Remark \ref{rem:translation}, it suffices to establish
Proposition \ref{prop:observable-value-interpretation} for \emph{closed}
observable values, as we have done.

\subsection{Soundness and Computational Adequacy}
\label{sub:soundness-adequacy}

Our final contirubtion is to show our semantic interpretation is sound and (strongly) adequate.

\begin{lemma}[Substitution]
\label{lem:term-substitution}
Let $\Phi \vdash v : P$ be a classical value and $\Phi; \mathbf{ \Sigma \vdash v : A}$ a quantum value.
\begin{enumerate}
\item If $\Phi, x : P \vdash m : R$, then $\lrb{ m[v/x] } = \lrb m \kcirc (\id_{\lrb \Phi} \ktimes \lrb v) \kcirc \JJ \langle \id_{\lrb \Phi}, \id_{\lrb \Phi} \rangle .$
\item If $\Phi; \mathbf{\Gamma,  y : A \vdash q : B}$ and $x \in \lrb \Phi$, then $\qlrb{ q[v/y] }_x = \qlrb q_x \circ (\id_{\qlrb \Gamma} \otimes \qlrb v_x) .$
\item If $\Phi, z : P; \mathbf{\Gamma \vdash q : B}$ and $x \in \lrb \Phi$, then $\lrb{ \mathbf q[v/z] }(x) = \qlrb q(x, \flrb v(x)).$
\end{enumerate}
\end{lemma}

Soundness is the statement that our interpretation is invariant under single-step reduction in a probabilistic sense. In both equations sums of morphisms are defined pointwise using the convex structure of the codomain (which is a Kegelspitze).

\begin{theorem}[Soundness]
\label{thm:soundness}
For any classical term $\Phi \vdash m : P$ and configuration $\Phi \vdash \mathcal C: \mathbf A; \qbit^k$:
\[
\lrb m = \sum_{m \probto{p} m'} p \lrb{m'} \qquad \qquad \qquad \lrb{\mathcal C} = \sum_{\mathcal C \probto{r} \mathcal C'} r \lrb{\mathcal C'} 
\]
assuming $m \probto{p} m'$ and $\mathcal C \probto{r} \mathcal C'$ for some rules from the operational semantics (\secref{sec:syntax}) and where the convex sums range over all such rules.
\end{theorem}

In the above theorem, both sums have at most two summands. Our next theorem is much stronger, because it involves reductions spanning an arbitrary number of steps, and the convex sums can be countably infinite (these can be defined in any Kegelspitze by Definition \ref{def:convex-sums}).

\begin{theorem}[Strong Adequacy]
\label{thm:strong-adequacy}
  Let $\cdot \vdash m : P$ be a closed classical term and $\cdot \vdash \mathcal C: \mathbf A ; \qbit^k$ a closed quantum configuration. Then:
  \begin{align*}
  \lrb m = \sum_{v \in \text{Val}(m)} P(m \probto{}_* v) \lrb{v} \qquad \qquad
  \lrb{\mathcal C} = \sum_{\mathcal V \in \text{ValC}(\mathcal C)} P(\mathcal C \probto{}_* \mathcal V) \lrb{\mathcal V} .
  \end{align*}
\end{theorem}
\begin{proof}
See Appendix \ref{app:adequacy}.
\end{proof}

\begin{remark}
  As mentioned previously, the \emph{progress} property holds for quantum configurations that are closed (and not necessarily total). The above theorem indeed holds for precisely this class of configurations. Of course, it also holds for closed total configurations as a special case.
\end{remark}

\begin{corollary}[Adequacy]
\label{cor:adequacy}
Let $\cdot \vdash m : 1$ be a term and $\cdot \vdash \mathcal C: \mathbf I$ a quantum configuration. Then:
  \[ \lrb m_*(\{*\}) = \mathrm{Halt}(m) \quad \text{ and } \quad \lrb{\mathcal C}(1) = \mathrm{Halt}(\mathcal C) . \]
\end{corollary}
\begin{proof}
  If $v : 1$ is a value, then $v = ()$ and we have $\lrb{()}_*(\{*\}) = 1 \in \mathbb R$.
  Similarly, if $\cdot \vdash \mathcal V : \mathbf I$ is a value configuration, then $\mathcal V = [1, \varnothing, *]$ and it follows $\lrb{\mathcal V}(1) = 1 \in \mathbb R$. The proof then follows immediately by Theorem \ref{thm:strong-adequacy} and by definition of Halt in \eqref{eq:halt}.
\end{proof}

\section{Conclusion and Related Work}
\label{sec:related}

The work closest to ours is QWIRE \cite{qwire} and EWire \cite{ewire}. They are
related languages that have a classical and non-linear host language together
with a separate small quantum circuit language. However, neither language is
suitable for variational quantum programming, because quantum function
application is restricted to pure and deterministic (i.e. non-probabilistic)
classical programs.  Furthermore, both languages have only limited support for
recursion and repeat-until-success patterns are not expressible within the host
languages.  $\VQPL$ does not have such restrictions. Furthermore, neither of
these papers proves an adequacy result and the only soundness statements that
are proven are with respect to strongly-normalising fragments of the languages.
The focus in QWIRE and EWire is on quantum \emph{circuits} and not on effectful
quantum/probabilistic programming, which is the focus of our paper.

Other related work includes adequate and even fully-abstract semantics for the
quantum lambda calculus \cite{quant-semantics,quantum-games1,quantum-games2}.
In this version of the quantum lambda calculus, classical information is
represented by types of the form $!(A \multimap B)$, but arbitrary types $!A$
are not allowed. This makes it difficult to understand the connections to
classical probabilistic programming.  In another version of the quantum lambda
calculus \cite{quantum-von-neumann}, the authors allow types of the form $!A$,
but their model does not support recursion. In all of these papers, there is no
denotational analysis that relates the models of the quantum lambda calculus to
models of classical probabilistic programming. Indeed, this is the main focus
of our paper and we are the first to present such an analysis.

To conclude, we described a mixed linear/non-linear quantum/probabilistic
programming language which is suitable for programming hybrid quantum-classical
algorithms. Our language, $\VQPL$, is equipped with a type system and a
type-safe operational semantics that makes it suitable for variational quantum
programming (\secref{sec:syntax}). We showed how to interpret the (induced)
classical probabilistic effects using a commutative probabilistic monad on
dcpo's (\secref{sec:probabilistic}). We then showed how to interpret quantum
resources and effects in the category of hereditarily atomic von Neumann
algebras, which we proved is enriched over continuous domains
(\secref{sec:operator-algebras}). The relationship between quantum and
classical probabilistic effects is modelled via novel semantic methods
(\secref{sec:relationship}). Most notably, we use the barycentre maps of our
model that are well-behaved under the strong sense of enrichment we established
to combine probabilistic and quantum effects (\secref{sub:quantum-probabilistic-combination}).
Finally, we organised all the relevant mathematical structure in a
categorical model (\secref{sec:categorical-model}), and we described a sound
and strongly adequate denotational semantics for $\VQPL$ within that model
(\secref{sec:semantics}).

A natural question to ask is how to extend $\VQPL$ with higher-order quantum
lambda abstractions and to find a suitable model.  We have shown that our model
supports first-order quantum lambda abstractions, but we do not believe it
supports higher-order quantum ones. The problem is that there is no known order
on the category $\QQ_*$ that makes the adjunction
$\stikz{quantum-adjunction.tikz}$ $\dcpo$-enriched. Even though $\QQs$ is
monoidal closed, the induced Kleisli exponential is not $\dcpo$-enriched, and
so it does not behave well under (type) recursion. On the other hand, the
category $\QQ$ is enriched in a very strong sense (over domains), which is
crucial for establishing the strong correspondence between classical
probabilistic and quantum effects that we presented here. Resolving this
conundrum is left for future work.



\bibliography{refs}


\appendix

\newpage
\section{Domain Enrichment of Hereditarily Atomic von Neumann Algebras}
\label{app:proofs-vN-algebras}

In this appendix we prove that the category $\HA$ (and thus also $\QQ$) is enriched over
continuous dcpo's. Thus we have to show that $\HA(M,N)$ is a continuous dcpo
for each $M,N\in\HA$. We will rely heavily on the following lemma, which
follows from \cite[Example 2.7]{selinger:higher-order}:
\begin{lemma}\label{lem:HAimpliesContinuous}
	Let $M$ be a von Neumann algebra. If $M$ is hereditarily atomic, then $[0,1]_M$ is a continuous dcpo.
\end{lemma} 

We note that the converse of this lemma is shown in \cite{furber-continuousdcpos}. Basically, our proof strategy is to show that all principal downsets in $\HA(M,N)$ are order isomorphic to $[0,1]_R$ for some hereditarily atomic von Neumann algebra. It then follows from the above lemma that $\HA(M,N)$ is continuous as well.

We will rely on some topologies on von Neumann algebras. We already mentioned the strong and the weak operator topologies on a von Neumann algebra $M$ on $H$ in Section \ref{sub:w*-category-definition}, where the former is the locally convex topology induced by the seminorms $a\mapsto\|ah\|$ for $h\in H$, and the latter is the locally convex topology induced by the seminorms $a\mapsto|\langle k,ah\rangle|$ for $h,k\in H$. We further note that $M$ has an intrinsic topology, which is known under several names: the \emph{$\sigma$-weak operator topology}, the \emph{ultraweak operator topology} or the \emph{weak*-topology}. It is the locally convex topology on $M$ induced by the seminorms $a\mapsto\left|\sum_{n=1}^\infty\langle k_n,ah_n\rangle\right|$, where $(h_n)_n$ and $(k_n)_n$ are sequences in $H$ such that $\sum_{n=1}^\infty\|h_n\|^2<\infty$ and $\sum_{n=1}^\infty\|k_n\|^2<\infty$. Any completely positive map between von Neumann algebras is normal if and only if it is continuous with respect to the $\sigma$-weak operator topology \cite[Proposition III.2.2.2]{Blackadar}. The $\sigma$-weak operator topology is stronger than the weak operator topology, but these topologies coincide on norm-bounded subsets \cite[I.3.1.4]{Blackadar}. We note that the bicommutant theorem of von Neumann states that a unital $*$-subalgebra of $B(H)$ is a von Neumann algebra if and only if it is closed with respect to either the weak operator topology, the strong operator topology or the $\sigma$-weak operator topology (and hence w.r.t. all of them).

 Given a von Neumann algebra $M$ and a Hilbert space $H$, any completely positive map $\varphi:M\to B(H)$ can be decomposed as $\varphi(x)=v^*\pi(x)v$, where $\pi:M\to B(K)$ is a \emph{representation} of $M$ on some Hilbert space $K$, i.e., a unital $*$-homomorphism, and $v:H\to K$ is a linear map such that $\|v\|^2\leq\|\varphi(1)\|$. We say that $(\pi,v,K)$ is a \emph{Stinespring representation} of $\varphi$. Moreover, $K$ can be chosen to be \emph{minimal}, i.e.,  $\pi[M]vH$ is norm dense in $K$; up to unitary equivalence, the minimal Stinespring decomposition of $\varphi$ is unique. A proof of the existence of the minimal Stinespring representation is given in  in \cite[Theorem II.6.9.7]{Blackadar} and in \cite[Theorem 1.2.7]{stormer}. In the proof of \cite[Theorem III.2.2.4]{Blackadar} it is asserted that $\pi$ is normal when $\varphi$ is normal.

\begin{proposition}\label{prop:Stinespring order iso}
	Let $M$ be a von Neumann algebra, let $H$ be a Hilbert space, and let $\varphi:M\to B(H)$ be a normal completely positive map with minimal Stinespring representation $(\pi,v,K)$. Then we have an order isomorphism $[0,1]_{\pi[M]'}\to\downarrow\varphi$, $t\mapsto \varphi_t$, where $\varphi_t(x)=v^*t\pi(x)v$ for each $x\in M$ and where the downset is taken in $\vN(M,B(H))$. 
\end{proposition}
\begin{proof} 
	It follows from \cite[Theorem 3.5.3]{stormer} and the paragraph preceding it that the assignment $t\mapsto\varphi_t$ is a bijection from $[0,1]_{\pi[M]'}$ to the set $S$ of all completely positive maps $\psi:M\to B(H)$ such that $\varphi-\psi$ is completely positive. We show that $S=\downarrow\varphi$. Using the bijection, any $\psi\in S$ is of the form $v^*t\pi v$ for some $t\in[0,1]_{\pi[M]'}$. 
	Since $\varphi$ is normal, so is $\pi$, whence also $\psi=v^*t\pi v$ is normal. Moreover, $\varphi$ is subunital, and $\varphi-\psi$ is completely positive, so positive, hence $\psi(1_M)\leq\varphi(1_M)\leq 1_H$, so also $\psi$ is subunital. Thus $\psi\in\vN(M,B(H))$. Now by definition of the order $\leq$ on $\vN(M,B(H))$ (cf. Section \ref{subsection:enrichment vN}), $\varphi-\psi$ is completely positive expresses that $\psi\leq\varphi$, so the assignment $t\mapsto\varphi_t$ is indeed a bijection between $[0,1]_{\pi[M]'}$ and $\downarrow\varphi$. We have to show that it is an order isomorphism. Let $t_1\leq t_2$ in $[0,1]_{\pi[M]'}$. Let $x\in M$ be positive, so $x=y^*y$ for some $y\in M$.
	Then for each $i=1,2$, we have that $\pi(y^*)$ commutes with $t_i$ (since the latter is an element of $\pi[M]'$, the commutant of $\pi[M]$ in $B(K)$), whence
	$\varphi_{t_i}(x)=v^*t_i\pi(y^*)\pi(y)v=v^*\pi(y)^*t_i\pi(y)v$. Since for any operator $b$ the assignment $t\mapsto b^*tb$ preserves the order (cf. \cite[Corollary 4.2.7]{kadisonringrose:oa1}), it now follows that $\varphi_{t_1}(x)\leq\varphi_{t_2}(x)$ in $B(H)$, and since $x$ is an arbitrary positive element of $M$, we conclude that hence $\varphi_{t_1}\leq\varphi_{t_2}$.
	
	Next assume that $\varphi_{t_1}\leq\varphi_{t_2}$. We have to show that $t_1\leq t_2$. 
	Write $\varphi_i:=\varphi_{t_i}$ for $i=1,2$, and $\varphi_3:=\varphi$.
	For each $i=1,2,3$, let $(\pi_i,v_i,K_i)$ be the minimal Stinespring representation of $\varphi_i$.
	For $i\leq j$ in $\{1,2,3\}$, \cite[Lemma 3.5.2]{stormer} assures the existence of operators $s_{ij}:K_j\to K_i$ with $\|s_{ij}\|\leq 1$ such that
	\begin{itemize}
		\item[(i)] $s_{ij}v_j=v_i$;
		\item[(ii)] $s_{ij}\pi_j(x)=\pi_i(x)s_{ij}$ for each $x\in M$.
	\end{itemize}
Note that $\pi_3=\pi$ for $\varphi_3=\varphi$. As a consequence, for $i,j=1,2,3$ with $i\leq j$ and for each $x\in M$ we have
	\begin{equation}\label{eq:representations1}
	v_j^*s_{ij}^*s_{ij}\pi_j(x)v_j=v_i^*s_{ij}\pi_j(x)v_j=v_i^*\pi_i(x)s_{ij}v_j=v_i^*\pi_i(x)v_i=\varphi_i(x)
	\end{equation}
and for each $x\in M$ we have 
\begin{equation}\label{eq:representations2}
	s_{ij}^*s_{ij}\pi_j(x)=s_{ij}^*\pi_i(x)s_{ij}=s_{ij}^*\pi_i(x^*)^*s_{ij}=(\pi_i(x^*)s_{ij})^*s_{ij}=(s_{ij}\pi_j(x^*))^*s_{ij}=\pi_j(x)s_{ij}^*s_{ij}
\end{equation} so $s_{ij}^*s_{ij}$ is an element of $\pi_j[M]'$, the commutant of $\pi_j[M]$. In particular, we find that $s_{i3}^*s_{i3}\in\pi_3[M]'=\pi[M]'$ such that $v^*s_{i3}^*s_{i3}\pi v=v_3^*s_{i3}^*s_{i3}\pi_3v_3=\varphi_i$.
Since by the bijection $t_i$ is the unique element in $\pi[M]'$ such that $\varphi_i(x)=v^*t_i\pi(x)v$, it follows that $t_i=s_{i3}^*s_{i3}$.
Moreover, for each $x\in M$ we have
\begin{equation*} v^*s_{23}^*s_{12}^*s_{12}s_{23}\pi(x)v= v_3^*s_{23}^*s_{12}^*s_{12}s_{23}\pi_3(x)v_3=v_3^*s_{23}^*s_{12}^*s_{12}\pi_2(x)s_{23}v_3=v_2^*s_{12}^*s_{12}\pi_2(x)v_2=\varphi_1(x),
	\end{equation*}
where we used (\ref{eq:representations1}) in the last equality, and for each $x\in M$, we have
\begin{align*}   s_{23}^*s_{12}^* s_{12}s_{23}\pi_3(x) & =s_{23}^*s_{12}^* s_{12}\pi_2(x)s_{23}=s_{23}^*\pi_2(x)s_{12}^* s_{12}s_{23}=s_{23}^*\pi_2(x^*)^*s_{12}^* s_{12}s_{23} =(\pi_2(x^*)s_{23})^*s_{12}^* s_{12}s_{23}\\
	& =(s_{23}\pi_3(x^*))^*s_{12}^* s_{12}s_{23}=\pi_3(x^*)^*s_{23}^*s_{12}^* s_{12}s_{23}=\pi_3(x)s_{23}^*s_{12}^* s_{12}s_{23},
	\end{align*}
where we used (\ref{eq:representations2}) in the second equality. Thus $ s_{23}^*s_{12}^* s_{12}s_{23}$ is an element of $\pi_3[M]'=\pi[M]'$ and $ \varphi_1=v^*s_{23}^*s_{12}^* s_{12}s_{23}\pi v$, which allows us to conclude that $t_1=s_{23}^*s_{12}^*s_{12}s_{23}$. Since $\|s_{ij}\|\leq 1$, we have $\|s_{ij}^*s_{ij}\|=\|s_{ij}\|^2\leq 1$ using the C*-identity. It now follows from \cite[Proposition 4.2.3(ii)]{kadisonringrose:oa1} that $s_{ij}^*s_{ij}\leq 1$, so $1-s_{ij}^*s_{ij}$ is positive, whence $s_{23}^*(1-s_{12}^*s_{12})s_{23}$ is positive. But \[s_{23}^*(1-s_{12}^*s_{12})s_{23}=s_{23}^*s_{23}-s_{23}^*s_{12}^*s_{12}s_{23}=t_2-t_1,\] so also $t_2-t_1$ is positive, i.e., $t_1\leq t_2$. 
\end{proof}

It follows that $\vN(M,B(H))$ is a continuous dcpo if $\pi[M]'$ is hereditarily atomic. Thus we have to understand the minimal Stinespring representation $\pi$ of $M$ better. In case $M$ is a product of von Neumann algebras (as in the case of hereditarily atomic von Neumann algebras), it turns out we can describe $\pi$ in terms of the minimal Stinespring representations of the product factors of $M$ as is shown in the following two lemmas.

\begin{lemma}\label{lem:product of representations}
	Let $M=\prod_{\alpha\in \Alpha}M_\alpha$ be the product of a collection $(M_\alpha:\alpha\in \Alpha)$ of von Neumann algebras, for each $\alpha\in\Alpha$ let $\pi_\alpha:M_\alpha\to B(K_\alpha)$ be a representation on some Hilbert space $K_\alpha$, and let $K\defeq\bigoplus_{\alpha\in \Alpha}K_\alpha$ be the sum of the Hilbert spaces $K_\alpha$. Then $\pi:M\to B(K)$ defined by 
	$\pi(x)k=(\pi_\alpha(x_\alpha)k_\alpha)_{\alpha\in \Alpha}$ for each $x=(x_\alpha)_{\alpha\in \Alpha}$ in $M$ and each $k=(k_\alpha)_{\alpha\in \Alpha}$ in $K$ is a representation of $M$ on $K$ such there is an injective $*$-homomorphism $\rho:\pi[M]'\to\prod_{\alpha\in\Alpha}B(K_\alpha)$.
\end{lemma}
\begin{proof}
	Clearly $\pi$ is a $*$-homomorphism, hence a representation of $M$ onto $K$. For $\alpha\in\Alpha$, let $e_\alpha:K_\alpha\to K$ be the embedding and $p_\alpha:K\to K_\alpha$ be the projection, which are easily seen to be bounded with norm at most $1$, and to be each other's adjoints: for each $k\in K$ and each $h_\alpha\in K_\alpha$, we have $\langle p_\alpha k,h_\alpha\rangle=\langle k,e_\alpha h_\alpha\rangle$. A straigtforward application of \cite[2.5.8]{kadisonringrose:oa1} gives us that $\sum_{\alpha\in\Alpha}e_\alpha p_{\alpha}$ converges strongly to $1_K$.
		For each $y\in B(K)$, and each $\alpha,\beta\in\Alpha$ let $y_{\alpha\beta}=p_\beta y e_\alpha:K_\alpha\to K_\beta$, so $y_{\alpha\beta}\in B(K_\alpha,K_\beta)$. Then $y$ is completely determined by the $y_{\alpha\beta}$. Indeed, if $z\in B(K)$ and for each $\alpha,\beta\in\Alpha$, we have $y_{\alpha\beta}=z_{\alpha\beta}$, then we have $e_{\beta}p_{\beta}ye_{\alpha}p_{\alpha}=e_{\beta}p_{\beta}ze_{\alpha}p_{\alpha}$, and since multiplication with a fixed operator is strongly continuous (cf. \cite[Remark 2.5.10]{kadisonringrose:oa1}), and $\sum_{\alpha\in\Alpha}e_{\alpha}p_{\alpha}$ converges strongly to $1$, we obtain $y=z$.

	For $y,z\in B(K)$ and $\beta,\gamma\in\Alpha$, multiplication can be described in a way similar to matrix multiplication: $(yz)_{\beta\gamma}=\sum_{\alpha\in\Alpha}y_{\beta\alpha}z_{\alpha\gamma}$.
	Moreover, for $x=(x_\alpha)_{\alpha\in \Alpha}$ in $M$, $\beta,\gamma\in\Alpha$ and each $a\in K_\gamma$ we have 
	\[(\pi(x))_{\beta\gamma}a=p_\gamma\pi(x)e_\beta(a)=p_\gamma\pi(x)(\delta_{\alpha\beta}a)_{\alpha\in\Alpha}=p_\gamma (\pi_\alpha(x_\alpha)\delta_{\alpha\beta}a)_{\alpha\in\Alpha}=\pi_\gamma(x_\gamma)(\delta_{\gamma\beta}a)=\delta_{\beta\gamma}\pi_\gamma(x_\gamma)a,\]
	hence 
	\begin{equation}\label{eq:pirepresentation}
	\pi(x)_{\beta\gamma}=\delta_{\beta\gamma}\pi_\gamma(x_\gamma).
	\end{equation}
	Let $D=\{y\in B(K):y_{\alpha\beta}=0\text{ for each }\alpha\neq\beta\text{ in }\Alpha\}.$
	We define a map $\rho:D\to\prod_{\alpha\in\Alpha}B(K_\alpha)$ by $y\mapsto (y_{\alpha\alpha})_{\alpha\in\Alpha}$. In order to show that $\rho$ is well defined, let $y\in D$. Since for each $\alpha\in\Alpha$, we have
	$\|y_{\alpha\alpha}\|=\|p_\alpha y e_\alpha\|\leq\|p_\alpha\|\|y\|\|e_\alpha\|\leq \|y\|,$
	because $\|p_\alpha\|,\|e_\alpha\|\leq 1$, it follows that $\sup_{\alpha\in\Alpha}\|y_{\alpha\alpha}\|\leq\|y\|<\infty,$
	which shows that $\rho(y)$ is a well-defined element of $\prod_{\alpha\in\Alpha}B(K_\alpha)$.
	Clearly $\rho$ is a $*$-homomorphism (where preservation of the involution $(-)^*$ follows since $p_\alpha$ and $e_\alpha$ are each other's adjoints). It is injective, because each $y\in B(K)$ is determined by $y_{\beta\gamma}$ for $\beta,\gamma\in\Alpha$, and by definition of $D$, we have $y_{\beta\gamma}=0$ for each $\beta\neq\gamma$.
	
	Let $y\in \pi[M]'$, i.e., $y\in B(K)$ such that $\pi(x)y=y\pi(x)$ for each $x=(x_\alpha)_{\alpha\in \Alpha}$ in $M$. Let $\beta\neq\gamma$ in $\Alpha$, so in particular $x_\gamma=0$. 
	Let $x_\beta=1_{M_\beta}$ and let $x_\alpha=0$ for each $\alpha\neq\beta$. Then $x=(x_\alpha)_{\alpha\in\Alpha}$ is an element of $M$ and since $y\in\pi[M]'$, we find:
	\begin{align*}y_{\beta\gamma}=\pi_\beta(1_{M_\beta})y_{\beta\gamma} & =\pi_\beta(x_\beta)y_{\beta\gamma}=\sum_{\alpha\in\Alpha}\delta_{\beta\alpha}\pi_\alpha(x_\alpha)y_{\alpha\gamma}=\sum_{\alpha\in\Alpha}\pi(x)_{\beta\alpha}y_{\alpha\gamma}=(\pi(x)y)_{\beta\gamma}\\
	& = (y\pi(x))_{\beta\gamma}=\sum_{\alpha\in\Alpha}y_{\beta\alpha}\pi(x)_{\alpha\gamma}=\sum_{\alpha\in\Alpha}y_{\beta\alpha}\delta_{\alpha\gamma}\pi_\gamma(x_\gamma)=y_{\beta\gamma}\pi_\gamma(x_\gamma)=0,
	\end{align*}
where we used (\ref{eq:pirepresentation}) in the fourth and eighth equality.
	Hence $y\in D$, from which follows that $\rho$ restricts to an injective $*$-homomorphism  $\pi[M]'\to \prod_{\alpha\in\Alpha}B(K_\alpha)$.
\end{proof}

\begin{lemma}\label{lem:Stinespring representation decomposition}
	Let $(M_\alpha)_{\alpha\in\Alpha}$ be a collection of von Neumann algebras, and let $M=\prod_{\alpha\in\Alpha}M_\alpha$, and for $\beta\in\Alpha$, denote the embedding $M_\beta\to M$ by $\iota_\beta$. Let $H$ be a Hilbert space, and let $\varphi:M\to B(H)$ be a normal completely positive subunital map. Then for each $\alpha\in\Alpha$, the map $\varphi_\alpha:=\varphi\circ\iota_\alpha$ is a completely positive map $M_\alpha\to B(H)$, and if $(\pi_\alpha,v_\alpha,K_\alpha)$ denotes the minimal Stinespring representation corresponding to $\varphi_\alpha$, then the representation $(\pi,v,K)$ constructed in Lemma \ref{lem:product of representations} from the representations $(\pi_\alpha,v_\alpha,K_\alpha)$ is the minimal Stinespring representation corresponding to $\varphi$.
\end{lemma}
\begin{proof}
	Let $\alpha\in\Alpha$. Then  $\iota_\alpha$ is a (non-unital) $*$-homomorphism, hence completely positive, whence $\varphi_\alpha$ is completely positive. 
	The identity element in $M_\alpha$ corresponds to a projection $r_\alpha$ in $M$, and $\sum_{\alpha\in\Alpha}r_\alpha=1_M$, where the convergence is with respect to the strong operator topology. Since multiplication with a fixed element is continuous with respect to the strong operator topology, it follows that $x=\sum_{\alpha\in\Alpha}xr_\alpha$ for each $x=(x_\alpha)_{\alpha\in\Alpha}$ in $M$. It is easy to see that $xr_\alpha=\iota_{\alpha}(x_\alpha)$: if $\bar x_\alpha$ denotes the element in $M$ whose $\alpha$-component is $x_\alpha$ and with all other components vanishing, then clearly $\bar x_\alpha=\iota_{\alpha}(x_\alpha)$ and $xr_\alpha=\bar x_\alpha r_\alpha$, hence $xr_\alpha=\iota_\alpha(x_\alpha)\iota_\alpha (1_{M_\alpha})=\iota_{\alpha}(x_\alpha 1_{M_\alpha})=\iota_{\alpha}(x_\alpha)$. Since convergence with respect to the strong operator topology implies convergence with respect to the weak operator topology, and the latter topology coincides with the $\sigma$-weak operator topology on the unit ball of an operator algebra \cite[Lemma II.2.5]{takesaki:oa1}, it follows that $x=\sum_{\alpha\in\Alpha}\iota_\alpha(x_\alpha)$, where the sum converges with respect to the $\sigma$-weak operator topology.
	
	Since $\varphi$ is normal, so continuous with respect to the $\sigma$-weak operator topology, and since we previously found that $x=\sum_{\alpha\in\Alpha}\iota_{\alpha}(x_\alpha)$ for each $x=(x_\alpha)_{\alpha\in\Alpha}$ in $M$, where the sum converges with respect to the $\sigma$-weak operator topology, we obtain \[\varphi(x)=\varphi\left(\sum_{\alpha\in\Alpha}\iota_{\alpha}(x_\alpha)\right)=\sum_{\alpha\in\Alpha}\varphi\circ\iota_\alpha(x_\alpha)=\sum_{\alpha\in\Alpha}\varphi_\alpha(x_\alpha).\]
	
	For each $\alpha\in\Alpha$, let $(\pi_\alpha,v_\alpha,K_\alpha)$ be the minimal Stinespring representation of $\varphi_\alpha$, so $v_\alpha:H\to K_\alpha$ satisfies $\|v_\alpha\|^2\leq\|\varphi_\alpha(1_{M_\alpha})\|$ and  $\varphi_\alpha(x)=v_\alpha^*\pi_\alpha(x)v_\alpha$ for each $x\in M_\alpha$, and $\overline{\pi_\alpha[M_\alpha] v_\alpha H}=K_\alpha$. Let $K=\bigoplus_{\alpha\in \Alpha}K_\alpha$. We want to define  $v:H\to K$ by $vh=(v_\alpha h)_{\alpha\in\Alpha}$, but we have to show that each $vh$ is indeed an element of $K$, i.e., we have to show that $\sum_{\alpha\in\Alpha}\|v_\alpha h\|^2<\infty$. Since each Stinespring representation $\pi_\alpha$ is unital, and since $\varphi$ is subunital, we obtain
	\begin{align*}
	1_H & \geq \varphi(1_M) =\sum_{\alpha\in\Alpha}\varphi_\alpha(1_{M_\alpha})=\sum_{\alpha\in\Alpha} v_\alpha^*\pi_\alpha(1_{M_\alpha})v_\alpha=\sum_{\alpha\in\Alpha}v_\alpha^*v_\alpha.
	\end{align*}
	Hence for each $h\in H$, we obtain
	\[\sum_{\alpha\in\Alpha}\|v_\alpha h\|^2=\sum_{\alpha\in\Alpha}\langle v_\alpha h,v_\alpha h\rangle=\sum_{\alpha\in\Alpha}\langle h,v_\alpha^*v_\alpha h\rangle=\left\langle h,\left(\sum_{\alpha\in\Alpha}v_\alpha^* v_\alpha\right) h\right\rangle \leq \langle h,h\rangle=\|h\|^2,\]
	where the third equality is due to the fact convergence with respect to the $\sigma$-weak operator topology implies convergence with respect to the weak operator topology. Thus $v:H\to K$ given by $vh=(v_\alpha h)_{\alpha\in\Alpha}$ is indeed a well defined operator.
	
	Then for $h,k\in H$, and $x=(x_\alpha)_{\alpha\in\Alpha}$ in $M$, we have
	\begin{align*}\langle k, v^*\pi(x)vh\rangle_H & =\langle vk,\pi(x)vh\rangle_K=\langle (v_\alpha k)_{\alpha\in\Alpha},(\pi_\alpha(x_\alpha)v_\alpha  h)_{\alpha\in\Alpha}\rangle_K\\
	&  =\sum_{\alpha\in\Alpha}\langle v_\alpha k,\pi_\alpha(x_\alpha)v_\alpha h\rangle_{K_{\alpha}}=\sum_{\alpha\in\Alpha}\langle h,v_\alpha^*\pi_\alpha(x_\alpha)v_{\alpha}h\rangle_H  \\
	& = \sum_{\alpha\in\Alpha}\langle k,\varphi_\alpha(x_\alpha)h\rangle_H= \left\langle k,\sum_{\alpha\in\Alpha}\varphi_\alpha(x_\alpha)h\right\rangle \\
	& = \langle k,\varphi(x)h\rangle,
	\end{align*}
	where the penultimate equality is because $\sum_{\alpha}\varphi_\alpha(x_\alpha)$ converges to $\varphi(x)$ in the $\sigma$-weak operator topology, hence also in the weak operator topology. Since we can take $h$ and $k$ in $H$ arbitrary, it follows that $v^*\pi(x)v=\varphi(x)$.

	Finally, we show that $K=\overline{\pi[M]vH}$. Let $e_\alpha:K_\alpha\to K$ denote the embedding and denote by $G_\alpha$ its image in $K$. Let $k=(k_\alpha)_{\alpha\in\Alpha}$ be an element of $K$. Then for each $\alpha\in\Alpha$ it follows that $k_\alpha\in K_\alpha=\overline{\pi_\alpha[M_\alpha]v_\alpha H}$.
	
	Let $\alpha\in\Alpha$, $x\in M_\alpha$ and $h\in H$. Then $\iota_{\alpha}(x)_\beta=\delta_{\alpha\beta} x$ for each $\beta\in\Alpha$, hence \[(\pi(\iota_{\alpha}(x))v(h))_\beta=\pi_\beta(\iota_\alpha(x)_\beta)(v(h))_\beta=\pi_\beta(\delta_{\alpha\beta}x)v_\beta h=\delta_{\alpha\beta}\pi_\alpha(x)v_\alpha h=(e_\alpha\pi_\alpha(x)v_\alpha h)_\beta,\] so $\pi(\iota_\alpha(x))v(h)=e_\alpha \pi_\alpha(x)v_\alpha(h)$, where we recall that $e_\alpha:K_\alpha\to K$ is the embedding, whose image is $G_\alpha\subseteq K$.
	As a consequence, we have for each $\alpha\in\Alpha$:
	\[G_\alpha=e_\alpha[K_\alpha]=e_\alpha[\overline{\pi_\alpha[M_\alpha]v_\alpha H}]\subseteq \overline{e_\alpha [\pi_\alpha[M_\alpha]v_\alpha H}] ]=\overline{\pi[\iota_\alpha[M_\alpha]]vH}\subseteq\overline{\pi[M]vH},\]
	where we used that $e_\alpha$ is bounded, so continuous in the first inclusion. Since clearly $\bigvee_{\alpha\in\Alpha}G_\alpha=K$ in the lattice of closed subspaces of $K$, it follows that $K\subseteq\overline{\pi[M]vH}$, which implies $K=\overline{\pi[M]vH}$. We conclude that $\pi$ is the minimal Stinespring representation corresponding to $\varphi$.
\end{proof}

\begin{lemma}\label{lem:product of homsets}
	Let $M$ be a von Neumann algebra, and let $(N_\alpha)_{\alpha\in\Alpha}$ be a collection of von Neumann algebras. Let $N=\prod_{\alpha\in\Alpha}N_\alpha$, and let $\pi_\alpha:N\to N_\alpha$ be the projection on the $\alpha$-th coordinate. Then $\iota:\vN(M,N)\to\prod_{\alpha\in\Alpha}\vN(M,N_\alpha)$, $\varphi\mapsto(\pi_\alpha\circ\varphi)_{\alpha\in\Alpha}$ is an order isomorphism.
\end{lemma}
\begin{proof}
	Firstly, $\pi_\alpha$ is a unital $*$-homomorphism, so certainly completely positive and subunital. Let $H_\alpha$ be the Hilbert space such that $N_\alpha$ is a von Neumann algebra on $H_\alpha$, hence $N$ is a von Neumann algebra on $H\defeq\bigoplus H_\alpha$. 
	Let $x=(x_\alpha)_{\alpha\in\Alpha}$ be an element of $N$. Then $x$ is positive if and only if for each $h=(h_\alpha)_{\alpha\in\Alpha}$ in $H$ we have $0\leq \langle h,xh\rangle$, i.e., if and only if $0\leq \sum_{\alpha\in\Alpha}\langle h_\alpha,x_\alpha h_\alpha\rangle$, from which it is clear that the positivity of each $x_\alpha$ is sufficient for $x$ to be positive. Since for fixed $\beta\in\Alpha$, we can choose $h$ in such a way that $h_\alpha=0$ for each $\alpha\neq\beta$, it follows that $x$ positive implies that $\langle h_\beta,x_\beta h_\beta\rangle\geq 0$, so it is also necessary for $x$ to be positive that each $x_\beta$ is positive. Thus $x$ is positive if and only if each $x_\alpha$ is positive. As a consequence, if $(x^d)_{d\in D}$ is a monotonically ascending net in $N_{\mathrm{sa}}$ with supremum $x$, and if we write $x^d=(x_\alpha^d)_{\alpha\in\Alpha}$ and $x=(x_\alpha)_{\alpha\in\Alpha}$, we have $x_\alpha=\sup_{d\in D}x_\alpha^d$, whence $\pi_\alpha(x)=\sup_{d\in D}\pi_\alpha(x^d)$, so $\pi_\alpha$ is normal.

As a consequence, $\iota$ is well defined, and if $\theta:M\to N$ is a normal completely positive subunital map, then for each $\alpha\in\Alpha$, the map $\pi_\alpha\circ\theta$ is normal, completely positive, and subunital. Moreover, since we found that $x=(x_\alpha)_{\alpha\in\Alpha}$ is positive if and only if each $x_\alpha$ is positive, it follows that the product $\prod$ on $\vNs$ extends to a product on $\vN_+$, the category of von Neumann algebras and normal positive maps. 

Now assume that $\theta:M\to N$ is a map such that $\pi_\alpha\circ \theta$ is a normal completely positive subunital map for each $\alpha\in\Alpha$. Clearly it follows that $\theta$ is a normal positive subunital map. We assert it is also completely positive, so we need to show that for fixed $n\in\mathbb N$, $\theta^{(n)}$ is positive. In order to see this, we first mention that from \cite[Proposition IV.1.6]{takesaki:oa1} it follows $\mathrm{M}_n(\mathbb C)\bar\otimes(-):\vNs\to\vNs$ is natural isomorphic to the functor $\mathrm{M}_n(-):\vNs\to\vNs$, which acts on morphisms via $\rho\mapsto\rho^{(n)}$. Moreover, since $(\vNs)^\op$ is monoidal closed \cite[Theorem 9.5]{Kornell17}, it follows that its monoidal product preserves coproducts, hence $\mathrm{M}_n(\mathbb C)\bar\otimes(-)$ preserves products in $\vNs$ (see also \cite[Corollary 9.6]{Kornell17}). As a consequence, we have a $*$-isomorphism $\kappa:\mathrm{M}_n(N)\to\prod_{\alpha\in\Alpha}\mathrm{M}_n(N_\alpha)$. Since $\pi_\alpha\circ\theta$ is completely positive, it follows that $(\pi_\alpha\circ\theta)^{(n)}:\mathrm{M}_n(M)\to\mathrm{M}_n(N_\alpha)$ is positive, and since $\prod$ is the product in $\vN_+$, it follows that $((\pi_\alpha\circ\theta)^{(n)} )_{\alpha\in\Alpha}:\mathrm{M}_n(M)\to\prod_{\alpha\in\Alpha}\mathrm{M}_n(N_\alpha)$ is positive. Now, $\vNs$ and $\vN_+$ share the same isomorphisms, so $\kappa$ is also an isomorphism in $\vN_+$. Since both $\mathrm{M}_n(N)$ and $\prod_{\alpha}\mathrm{M}_n(N_\alpha)$ are the product of the $\mathrm{M}_n(N_\alpha)$ in $\vN_+$, it follows that $\kappa\circ\theta^{(n)}=( (\pi_\alpha\circ\theta)^{(n)})_{\alpha\in\Alpha}$, so $\theta^{(n)}=\kappa^{-1}\circ ( (\pi_\alpha\circ\theta)^{(n)})_{\alpha\in\Alpha}$ is the composition of two positive functions, hence positive. We conclude that $\theta$ is completely positive if and only if $\pi_\alpha\circ\theta$ is completely positive for each $\alpha\in\Alpha$. 

It now follows that $\varphi\leq\psi$ in $\vN(M,N)$ if and only if $\theta:=\psi-\varphi$ is completely positive,

 if and only if $\pi_\alpha\circ\theta=\pi_\alpha\circ\psi-\pi_\alpha\circ\varphi$ is completely positive for each $\alpha\in\Alpha$,
 
  if and only if $\pi_\alpha\circ\varphi\leq\pi_\alpha\circ\psi$ in $\vN(M,N_\alpha)$ for each $\alpha\in\Alpha$, 
  
  if and only if $\iota(\varphi)=(\pi_\alpha\circ\varphi)_{\alpha\in\Alpha}\leq (\pi_\alpha\circ\psi)_{\alpha\in\Alpha}=\iota(\psi)$ in $\prod_{\alpha\in\Alpha}\vN(M,N_\alpha)$.\\ Thus $\iota$ is an order embedding. The map $\iota$ also is surjective: If $\psi\in\prod_{\alpha\in\Alpha}\vN(M,N_\alpha)$, then $\psi=(\psi_\alpha)_{\alpha\in\Alpha}$ for some completely positive maps $\psi_\alpha:M\to N_\alpha$. Then the properties of the categorical product, $N = \prod_{\alpha\in A} N_\alpha$, imply there is a map $\varphi:M\to N$ with $\pi_\alpha\circ \varphi=\psi_\alpha$ for each $\alpha\in\Alpha$, i.e., $\iota(\varphi)=
  (\pi_\alpha\circ \varphi)_{\alpha\in A} = (\psi_\alpha)_{\alpha\in A} = \psi$.
\end{proof}

\begin{proof}[Proof of Theorem \ref{thm:continuous enrichment}]
	Let $M$ and $N$ be hereditarily atomic von Neumann algebras. We show that the pointed dcpo $\vN(M,N)$ is continuous. We first assume that $N=B(H)$ for some finite-dimensional Hilbert space $H$. 	
	Since $M$ is hereditarily atomic, we can write $M=\prod_{\alpha\in\Alpha}M_\alpha$, where $M_\alpha$ is a matrix algebra. Let $\varphi\in\vN(M,N)$. It follows from combining Lemmas \ref{lem:product of representations} and \ref{lem:Stinespring representation decomposition} that the minimal Stinespring representation $(\pi,v,K)$ of $\varphi$ can be obtained from the minimal Stinespring representations $(\pi_\alpha,v_\alpha,K_\alpha)$ of $\varphi_\alpha:=\varphi\circ\iota_\alpha$, and that $\pi[M]'$ embeds into $\prod_{\alpha\in\Alpha}B(K_\alpha)$. Since $(\pi_\alpha,v_\alpha,K_\alpha)$ is minimal, $\pi_\alpha[M_\alpha]v_\alpha H$ is dense in $K_\alpha$. Since both $M_\alpha$ and $H$ are finite-dimensional, it follows that $K_\alpha$ is finite-dimensional, too, so $\prod_{\alpha\in\Alpha}B(K_\alpha)$ is hereditarily atomic. Since $\pi[M]'$ embeds in this algebra, it is a hereditarily atomic von Neumann algebra, too, hence its unit interval $[0,1]_{\pi[M]'}$ is a continuous dcpo by Lemma \ref{lem:HAimpliesContinuous}. It now follows  from Proposition \ref{prop:Stinespring order iso} that $\downarrow\varphi$ is a continuous dcpo. Thus all principal downsets in $\vN(M,N)$ are continuous, so $\vN(M,N)$ is continuous by \cite[Proposition 2.2.17]{abramskyjung:domaintheory}. 
	
	If $N$ is an arbitrary hereditarily atomic von Neumann algebra, then $N=\prod_{\beta\in\mathrm{B}}B(H_\beta)$ for some finite-dimensional Hilbert spaces $H_\beta$. By Lemma \ref{lem:product of homsets} we have an order isomorphism $\vN(M,N)\cong\prod_{\beta\in\mathrm{B}}\vN(M,B(H_\beta))$, and since the product of pointed continuous dcpos is continuous \cite[Exercise I-2.18]{gierzetal:domains}, it follows that $\vN(M,N)$ is continuous. Hence $\vN$ is enriched over continuous dcpos, so  its dual $\QQ$ also is enriched over continuous dcpos.
\end{proof}

Recall the definition of a Kegelspitze in Section \ref{sub:kegelspitzen}. We denote the category of Kegelspitzen and Scott continuous linear maps by $\mathbf{KS}$. We proceed by showing that $\vN$ is enriched over $\mathbf{KS}$, whence $\HA$ and $\QQ$ are also enriched over $\mathbf{KS}$.


\begin{lemma}
	Let $M$ and $N$ be von Neumann algebras. Then $\vN(M,N)$ is a barycentric algebra if we define $\varphi+_r\psi:=r\varphi+(1-r)\psi$ for each $\varphi,\psi\in\vN(M,N)$ and each $r\in[0,1]$.
\end{lemma}
\begin{proof}
	Clearly $\varphi+_r\psi$ is normal. We show that it is completely positive. First assume that $R$ is a von Neumann algebra and $x,y\in R$ are positive. By \cite[Theorem 4.2.2]{kadisonringrose:oa1} $x+y$ is positive, and $rx$ is positive for each $r\in[0,\infty)$. Hence if $r\in[0,1]$ it follows that $rx+(1-r)y$ is a positive element of $R$. Now let $\varphi,\psi:M\to N$ be positive maps. So $\varphi(x)$ and $\psi(x)$ are positive for each positive $x\in M$, hence for $r\in[0,1]$, we also have that $(r\varphi+(1-r)\psi)(x)=r\varphi(x)+(1-r)\psi(x)$ is positive, so $r\varphi+(1-r)\psi$ is a positive map. 
	
	Recall that if $\omega:M\to N$ is a map and $n\in\mathbb{N}$, then $\omega^{(n)}:\mathrm{M}_n(M)\to\mathrm{M}_n(N)$ is the map $[x_{ij}]_{i,j=1}^n\mapsto[\varphi(x_{ij})]_{i,j=1}^n$, and that $\omega$ is completely positive if and only if $\omega{(n)}$ is positive for each $n\in\mathbb{N}$. Let $n\in\mathbb{N}$. Then for each $r,s\in\mathbb C$ and each $[x_{ij}]_{i,j=1}^n$ in $\mathrm{M}_n(M)$, we have
	\begin{align*}
	(r\varphi+s\psi)^{(n)}([x_{ij}]_{i,j=1}^n) & = [ (r\varphi+s\psi)(x_{ij})]_{i,j=1}^n=[r\varphi(x_{ij})+s\psi(x_{ij})]_{i,j=1}^n\\
	& =r[\varphi(x_{ij})]_{i,j=1}^n+s[\psi(x_{ij})]_{i,j=1}^n=r\varphi^{(n)}([x_{ij}]_{i,j=1}^n)+s\psi^{(n)}([x_{ij}]_{i,j=1}^n)\\
	& =(r\varphi^{(n)}+s\psi^{(n)})([x_{ij}]_{i,j=1}^n),
	\end{align*}
	hence $(r\varphi+s\psi)^{(n)}=r\varphi^{(n)}+s\psi^{(n)}$. Assume that $\varphi,\psi:M\to N$ are normal completely positive subunital maps. Then $\varphi^{(n)}$ and $\psi^{(n)}$ are positive, hence $r\varphi^{(n)}+(1-r)\psi^{(n)}$ is positive, which equals $(r\varphi+(1-r)\psi)^{(n)}$, and since $n\in\mathbb N$ is arbitrary, we conclude that $r\varphi+(1-r)\psi$ is completely positive. Moreover, since both $\varphi$ and $\psi$ are subunital, we have
	\[ (r\varphi+(1-r)\psi)(1_M)=r\varphi(1_M)+(1-r)\psi(1_M)\leq r 1_N+(1-r) 1_N=1_N,\]
	so $r\varphi+(1-r)\psi$ is subunital. Finally, $\varphi$ and $\psi$ are normal, so continuous with respect to the $\sigma$-weak operator topology, hence so is their convex combination $r\varphi+(1-r)\psi$. Thus $+_r$ defined by $\varphi+_r\psi:=r\varphi+(1-r)\psi$ for each $r\in[0,1]$ is a well-defined binary operation on $\vN(M,N)$.
	We now have $\varphi+_1\psi=1\varphi+(1-1)\psi=\varphi$,  $\varphi+_r\varphi=r\varphi+(1-r)\varphi=\varphi$, and
	$\varphi+_r\psi=r\varphi+(1-r)\psi=(1-r)\psi+(1-(1-r))\varphi=\psi+_{1-r}\varphi$. Let $\omega\in\vN(M,N)$ and $p,r\in[0,1)$. Then
	\begin{align*}(\varphi+_p\psi)+_r\omega & =r(\varphi+_p\psi)+(1-r)\omega=r\big(p\varphi+(1-p)\psi)\big)+(1-r)\omega\\
	& =rp\varphi+r(1-p)\psi+(1-r)\omega =rp\varphi+(1-rp)\left(\frac{r-rp}{1-rp}\psi+\frac{1-r}{1-rp}\omega\right)\\
	& = \varphi+_{rp}\left(\frac{r-rp}{1-rp}\psi+\frac{1-r}{1-rp}\omega\right)=\varphi+_{rp}\left(\frac{r-rp}{1-rp}\psi+\frac{1-rp-(r-rp)}{1-rp}\omega\right)\\
	& = \varphi+_{rp}\left(\frac{r-rp}{1-rp}\psi+\left(1-\frac{r-rp}{1-rp}\right)\omega\right)=\varphi_{rp}+\left(\psi+_{\frac{r-rp}{1-rp}}\omega\right).
	\end{align*}
	We conclude that $\vN(M,N)$ is indeed a barycentric algebra.
\end{proof}

Given von Neumann algebras $M$ and $N$, the map $0:M\to N$ defined $x\mapsto 0$ is normal, completely positive and subunital, hence we can take it as the distinguished element of  $\vN(M,N)$.

\begin{lemma}\label{lem:homsets in WStar are linear}
	Let $M$, $N$ and $R$ be von Neumann algebras and let $\varphi:M\to N$ be a normal completely positive subunital map. Then
	\[\vN(R,\varphi):\vN(R,M)\to \vN(R,N),\quad \psi\mapsto \varphi\circ\psi\] and 
	\[\vN(\varphi,R):\vN(N,R)\to \vN(M,R),\quad  \psi\mapsto\psi\circ\varphi\] are linear maps between barycentric algebras.
\end{lemma}
\begin{proof}
	Since $\varphi$ is linear in the sense of linear algebra and $\varphi(0)=0$, it follows that $\vN(R,\varphi)$ is linear. 
	
	Write $f=\vN(\varphi,R)$. Then for $\psi,\omega\in \vN(N,R)$, we have
	\[ f(\psi+_r\omega)=(\psi+_r\omega)\circ\varphi,\]
	hence for each $x\in M$ we have
	\begin{align*} f(\psi+_r\omega)(x) & =((\psi+_r\omega)\circ\varphi)(x)=((r\psi+(1-r)\omega)\circ\varphi)(x)=(r\psi+(1-r)\omega)(\varphi(x))\\
	& =r\psi(\varphi(x))+(1-r)\omega(\varphi(x))=(r\psi\circ\varphi)(x)+((1-r)\omega\circ\varphi)(x)\\
	& =(r\psi\circ\varphi+(1-r)\omega\circ\varphi)(x) =(\psi\circ\varphi+_r\omega\circ\varphi)(x)
	=( f(\psi)+_rf(\omega))(x)
	\end{align*}
	hence $f(\psi+_r\omega)=f(\psi)+_rf(\omega)$, so $f$ is linear.
	Furthermore, we have for each $x\in M$:
	$f(0)(x)=(0\circ\omega)(x)=0(\omega(x))=0,$
	so $f(0)=0$, expressing that $f$ is linear.
\end{proof}



\begin{proposition}\label{prop:WStar homsets are KS}
	Let $M$ and $N$ be von Neumann algebras. Then $\vN(M,N)$ is a Kegelspitze.
\end{proposition}
\begin{proof}
	
	Upon inspecting the proof of \cite[Proposition 5.2]{cho:semantics} that shows that $\vN(M,N)$ is a pointed dcpo, the supremum $\varphi$ of any directed set $(\varphi_{\alpha})_{\alpha\in\Alpha}$ in $\vN(M,N)$ is calculated pointwise:
	$\varphi:M\to N$ is the normal completely positive subunital map such that
	$\varphi(x)=\bigvee_{\alpha\in\Alpha}\varphi_\alpha(x)$ for each $x\in M$, where the supremum of the $\varphi_{\alpha}(x)$ is calculated in $N_{\mathrm{sa}}$, the $\mathbb{R}$-vector space of all self-adjoint elements in $N$ and where this supremum is the limit of the $\varphi_\alpha(x)$ with respect to the strong operator topology on $N$ \cite[Lemma 5.1.4]{kadisonringrose:oa1}. This means that if $H$ is a Hilbert space such that $N$ is a von Neumann algebra on $B(H)$, we have that $\|\varphi_\alpha(x)h-\varphi(x)h\|$ converges to $0$ for each $h\in H$.

	We proceed by showing that $+_r$ is Scott continuous for each $r\in[0,1]$. Fix $\psi\in\vN(M,N)$. We start by showing that $(-)+_r\psi$ is monotone for each $\psi\in\vN(M,N)$, i.e., if $\varphi\leq\omega$ in $\vN(M,N)$, then $\varphi+_r\psi\leq\omega+_r\psi$. Fix $n\in\mathbb N$. Then $\varphi\leq\omega$ implies that $\omega-\varphi$ is completely positive, so $(\omega-\varphi)^{(n)}$ is positive. Then also $((\omega+_r\psi)-(\varphi+_r\psi))^{(n)}=(r\omega+(1-r)\psi-(r\varphi+(1-r)\psi))^{(n)}=(r\omega-r\varphi)^{(n)}=r(\omega-\varphi)^{(n)}$ is positive, so $(\omega+_r\psi)-(\varphi+_r\psi)$ is completely positive since $n$ is arbitrary. Hence indeed $\varphi+_r\psi\leq \omega+_r\psi$.
	
	Let $(\varphi_\alpha)_{\alpha\in\Alpha}$ be a directed set in $\vN(M,N)$ with supremum $\varphi$. For each $h\in H$, we have
	\[ \| (r\varphi_\alpha+(1-r)\psi)(x)h-(r\varphi+(1-r)\psi)(x)h\|=\|r\varphi_\alpha(x)h-r\varphi(x)h\|=|r|\|\varphi_\alpha(x)h-\varphi(x)h\|,\]
	which clearly converges to $0$, hence $\bigvee_{\alpha\in\Alpha}(r\varphi_\alpha(x)+(1-r)\psi(x)=r\varphi(x)+(1-r)\psi(x)$, whence $\bigvee_{\alpha\in\Alpha}(\varphi_\alpha+_r\psi)=\varphi+_r\psi$. We conclude that $+_r:\vN(M,N)\times \vN(M,N)\to \vN(M,N)$ is Scott continuous in the first variable, and since $\varphi+_r\psi=\psi+_{1-r}\varphi$, it also follows that $+_r$ is Scott continuous in the second variable, hence Scott continuous overall.
	
	Since $r\cdot \varphi=\varphi+_r0$, it follows that scalar multiplication is Scott continuous in $\varphi$, so we only have to check that it is Scott continuous in $r$. So fix $\varphi\in\vN(M,N)$, and let $D\subseteq[0,1]$ be a directed set with supremum $s$. 
	We need to show that $\bigvee_{d\in D}d\varphi=s\varphi$, hence for each $x\in M$, we need to show that $(d\varphi(x))_{d\in D}$ converges to $s\varphi(x)$ in the strong operator topology. Thus we need to show for each $x\in M$ and each $h\in H$ that
	$\|d\varphi(x)h-s\varphi(x)h\|$ converges to zero. But $\|d\varphi(x)-s\varphi(x)\|=|d-s|\|\varphi(x)\|$ and obviously $|d-s|$ converges to $0$ since $s$ is the supremum of $D$ in $[0,1]$, which show that scalar multiplication is also Scott continuous in both variables, hence Scott continuous overall.
\end{proof}

\begin{theorem}
	The category $\vN$ is enriched over $\mathbf{KS}$.
\end{theorem}
\begin{proof}
	By Proposition \ref{prop:WStar homsets are KS}, any homset in $\vN$ is a Kegelspitze. We have to verify that for any von Neumann algebras $M$, $N$, and $R$, and any normal completely positive subunital map $\varphi:M\to N$, the maps
	\[\vN(R,\varphi):\vN(R,M)\to \vN(R,N),\quad \psi\mapsto \varphi\circ\psi\] and 
	\[\vN(\varphi,R):\vN(N,R)\to \vN(M,R),\quad  \psi\mapsto\psi\circ\varphi\] 
	are morphisms in $\mathbf{KS}$. It follows from \cite[Theorem 5.3]{cho:semantics} that $\vN$ is enriched over $\mathbf{DCPO}_{\perp!}$, hence the morphisms above are Scott continuous. By Lemma \ref{lem:homsets in WStar are linear}, the maps are linear, hence indeed morphisms in $\mathbf{KS}$. \end{proof}

\newpage
\section{The isomorphism $r_{X}$}
\label{app:r-isomorphism}

\begin{definition}
Let $\mathcal L: \Set \to \dcpo$ be the functor defined by
\begin{align*}
\mathcal L X &\eqdef (X, \sqsubseteq), \text{ where } \sqsubseteq \text{ is the discrete order on $X$} \\
\mathcal L f &\eqdef f
\end{align*}
\end{definition}

\begin{definition}
	A complex Banach algebra $A$ is called a \emph{$*$-algebra} if it equipped with an idempotent conjugate-linear map $*:A\to A$ such that $(ab)^*=b^*a^*$ for each $a,b\in A$. If, in addition, $\|a^*a\|=\|a\|^2$ for each $a\in A$, we call $A$ a \emph{C*-algebra}.  
\end{definition}
In particular any von Neumann algebra is a C*-algebra. More generally, for each Hilbert space $H$, every norm-closed $*$-subalgebra of $B(H)$ is a C*-algebra. The converse holds as well: if $A$ is a C*-algebra, any $*$-homomorphism $\pi:A\to B(H)$ for some Hilbert space $H$ is called a \emph{representation}. The Gelfand-Naimark Representation Theorem \cite[Corollary II.6.4.10]{Blackadar} states that any C*-algebra has a faithful representation. 
\begin{definition}
	Any C*-algebra that is $*$-isomorphic to a von Neumann algebra is called a \emph{W*-algebra}.
\end{definition}
 Let $V$ be a Banach space. The space of all continuous linear maps $V\to\mathbb C$ is denoted by $V^*$.
 

The following theorem gives an alternative characterization of W*-algebras:
\begin{theorem}\cite[III.2.1.8, Theorems III.2.4.1 \& III.2.4.2]{Blackadar}
	Let $A$ be a C*-algebra. Then $A$ is a W*-algebra if and only if there is a Banach space $F$ such that $A$ is isometrically isomorphic to the dual $F^*$ of $F$. Under the isometric embedding of $F$ into $F^{**}$, and by the isometric isomorphism $A^*\cong F^{**}$, $F$ is isometrically isomorphic to $A_*$, the subspace of $A^*$ consisting of all normal functionals on $A$. 
\end{theorem}	
Hence if $M$ is a W*-algebra, then up to isometric isomorphism $M_*$ is the unique Banach space whose dual is isometrically isomorphic to $M$. We call $M_*$ the \emph{predual} of $M$.


Let $X$ be a set. From Example \ref{ex:ellinfty} we already know that $\ell^\infty(X)$ is a von Neumann algebra. Its predual is the the Banach space of all functions $f:X\to\mathbb C$ such that $\sum_{x\in X}|f(x)|<\infty$, which we denote by $\ell^1(X)$, see for instance Table B.1, p. 547 of \cite{landsman}. By the same table, we have a \emph{separating} pairing $\langle -,-\rangle:\ell^1(X)\times\ell^\infty(X)\to\mathbb C$, $(f,g)\mapsto\sum_{x\in X}f(x)g(x)$, i.e., for each different $f,f'\in\ell^1(X)$ there exists a $g\in \ell^\infty(X)$ such that $\langle f,g\rangle\neq\langle f',g\rangle$, and for each different $g,g'\in\ell^\infty(X)$ there is an $f\in\ell^1(X)$ such that $\langle f,g\rangle\neq\langle f,g'\rangle$. Using \cite[Theorem B.47]{landsman}, this pairing induces an isomorphism of vector spaces $\ell^1(X)\to\ell^\infty(X)_*$, which is isometric as one can see as follows. Let $g\in\ell^1(X)$, so $\|g\|=\sum_{x\in X}|g(x)|<\infty$. Let $f\in\ell^\infty(X)$ such that $\|f\|=1$, i.e., $\sup_{x\in X}|f(x)|=1$.
Then
\[ \|g\|=\|f\|\|g\|=\sup_{x\in X}\sum_{x'\in X}|f(x)||g(x')|\geq\sum_{x\in X}|f(x)||g(x)|\geq\left|\sum_{x\in X}f(x)g(x)\right|=|\zeta(g)(f)|, \]
hence 
$\|\zeta(g)\|=\sup\{|\zeta(g)(f)|:f\in\ell^\infty(X),\|f\|=1\}\leq\|g\|$. On the other hand, for ach $x\in X$, we can write $g(x)=|g(x)|e^{i\arg(g(x))}$. Let $f:X\to\mathbb C$ be given by $f(x)=e^{-i\arg(g(x))}$. Then $|f(x)|=1$ for each $x\in X$, hence $f\in\ell^\infty(X)$ with $\|f\|=1$. Then we have
\[ \|\zeta(g)\|\geq|\zeta(g)(f)|=\left|\sum_{x\in X}f(x)g(x)\right|=\left|\sum_{x\in X}|g(x)|\right|=\|g\|, \]
so $\|\zeta(g)\|=\|g\|$, showing that $\zeta$ is an isometry. We conclude:


\begin{lemma}\label{lem:r-map}
	We have an isometric isomorphism $\zeta:\ell^1(X)\to \ell^\infty(X)_*$ defined by 
	\[ \zeta(g)(f)=\sum_{x\in X}f(x)g(x)\] for each $g\in\ell^1(X)$ and each $f\in\ell^\infty(X)$.
\end{lemma}

Let $\mathcal D(X)$ be the set of all $g\in\ell^1(X)$ such that $g(x)\geq 0$ for each $x\in X$ and $\sum_{x\in X}g(x)\leq 1$. We order $\mathcal D(X)$ by $g\leq g'$ if and only if $g(x)\leq g'(x)$ for each $x\in X$.

\begin{proposition}\label{prop:r-map-1}
	The isometric isomorphism $\zeta:\ell^1(X)\to \ell^\infty(X)_*$ restricts to an isomorphism of Kegelspitzen $\mathcal D(X)\to\vN(\ell^\infty(X),\mathbb C)$.
\end{proposition}
\begin{proof}
Let $g\in\ell^1(X)$. Then $\zeta(g)\in\ell^\infty(X)^*$ is positive if and only if $\zeta(g)(f)$ is positive for each positive $f\in\ell^\infty(X)$, i.e., $\sum_{x\in X}f(x)g(x)\geq 0$ for each positive $f\in\ell^\infty(X)$. Since $f$ is positive if and only if $f(x)\geq 0$ for each $x\in x$, and for each $y\in X$ the function $e_y:X\to\mathbb C$ defined by $e_y(x)=\delta_{x,y}$ for each $x\in X$ is an element of $\ell^\infty(X)$ that clearly is positive, it follows that $\zeta(g)(f)$ is positive if and only if $g(x)\geq 0$ for each $x\in X$.  It also follows that if $g'\in\ell^1(X)$ is another element such that $\zeta(g')$ is positive, then $\zeta(g)\leq\zeta(g')$ if and only if $\zeta(g'-g)$ is positive if and only if $g(x)\leq g'(x)$ for each $x\in X$. 

By \cite[Proposition II.6.2.5]{Blackadar}, any positive $\omega\in\ell^\infty(X)^*$ is bounded with norm $\|\omega\|=\omega(1)$. Thus $\omega$ is subunital if and only if $\|\omega\|\leq 1$. Hence, if $\zeta(g)$ positive, it is bounded with norm $\|\zeta(g)\|=\zeta(g)(1)$, so $\|\zeta(g)\|=\sum_{x\in X}g(x)$. Thus $\zeta(g)$ is positive and subunital if and only if $\|g\|=\|\zeta(g)\|=\sum_{x\in X}g(x)\leq 1$, where we used that $\zeta$ is an isometry. It follows that $\zeta$ restricts to an isometric isomorphism between $\vN(\ell^\infty(X),\mathbb C)$ and $\mathcal D(X)$, which is also an order isomorphism. Since $\zeta$ is an isomorphisms of vector spaces, it in particular preserve the convex structure, whence its restriction is a Kegelspitze isomorphism.
\end{proof}

\begin{proposition}
\label{prop:l-infty}
For any set $X$, there exists an isomorphism of Kegelspitzen $r_{X} \colon \QQ(\ell^\infty (1), \ell^\infty (X)) \cong \KL(1, \mathcal L X) \colon r_X^{-1}$.
Moreover, $r_X$ restricts to an isomorphism of sets $r_{X} \colon \QQ_*(\ell^\infty (1), \ell^\infty (X)) \cong \TD(1, \mathcal L X) \colon r_X^{-1}$.
\end{proposition}
\begin{proof}
Firstly, we have an isomorphism \[\alpha:\KL(1,\mathcal L X)=\DCPO(1,\mathcal{ML}X)\to \mathcal {ML}(X)=\mathcal D(X),\qquad f\mapsto f(1).\] 
Furthermore, we have an isomorphism $i:\mathbb C\to\ell^\infty(1)$ that to each $\lambda\in\mathbb C$ assigns the function $c_\lambda\in\ell^\infty(1)$ defined by $c_\lambda(*)=\lambda$. This isomorphism induces an isomorphism \[\iota:\vN(\ell^\infty(X),\mathbb C)\to\vN(\ell^\infty(X),\ell^\infty(1))=\QQ(\ell^\infty(1),\ell^\infty(X)), \qquad f\mapsto i\circ f.\] It follows that the required isomorphism $r_X^{-1}:\KL(1, \mathcal L X)\to\QQ(\ell^\infty (1), \ell^\infty (X))$ is given by $\iota\circ\zeta\circ\alpha$, where $\zeta:\mathcal D(X)\to\vN(\ell^\infty(X),\mathbb C)$ is the isomorphism from Proposition \ref{prop:l-infty}.

Recall that $\ell^\infty:\Set\to\QQs$ is a functor if for each function $f:X\to Y$ between sets, we define $\ell^\infty(f):\ell^\infty(Y)\to\ell^\infty(X)$ by $g\mapsto g\circ f$. By \cite[Theorem 7.4]{Kornell18}, this functor is fully faithful. 
Now consider the following diagram:
\cstikz{bert-diagram.tikz}

Since it consists only of isomorphisms and two injective functors (note that $\JJ$, which was introduced in Section \ref{sub:model}, is fully faithful by construction of $\TD$ in Definition \ref{def:T}), commutativity of the diagram implies that $r_X$ restricts to an isomorphism $\QQs(\ell^\infty(1),\ell^\infty(X))\to\TD(1,\mathcal LX)$. In order to show that the diagram indeed is commutative, let $f\in\Set(1,X)$. We will write $f(*)=x$. Then $\ell^\infty(f)\in \QQs(\ell^1(X),\ell^\infty(X))=\HAs(\ell^\infty(X),\ell^\infty(1))$ is given by $\ell^\infty(f)(g)=g\circ f$, i.e, $(\ell^\infty(f)(g))(*)=(g\circ f)(*)=g(x)$, hence $\ell^\infty(f)(g)=c_{g(x)}$.

In the other direction, we have $\JJ(f)(*)=(\eta_{\mathcal L X}\circ f)(*)=\eta_{\mathcal L X}(x)=\delta_x$, hence 
\[(r_{X}^{-1}\circ \JJ(f)) =(\iota\circ\zeta\circ\alpha)(\JJ(f))=(\iota\circ\zeta)(\JJ(f)(*))=(\iota\circ\zeta)(\delta_x),\]
which is a function $\ell^\infty(X)\to\ell^\infty(1)$, hence for $g\in \ell^\infty(X)$, we have
\[(r_{X}^{-1}\circ \JJ(f))(g)=(\iota\circ\zeta)(\delta_x)(g)=\iota\left(\sum_{y\in X}\delta_x(y)g(y)\right)=\iota(g(x))=i(g(x))=c_{g(x)}.\]
Thus $(r_X^{-1}\circ \JJ(f))(g)=c_{g(x)}=\ell^\infty(f)(g)$, so $\JJ(f)=r_X\circ\ell^\infty(f)$, which shows that the diagram indeed commutes.
\end{proof}

Theorem \ref{thm:r-iso-categorical} now follows immediately, because $\MM(X, \sqsubseteq) \cong \KL(1, \mathcal L X)$ and $\TD(1, \mathcal L X)$ consists precisely of the Dirac valuations of $\MM (X, \sqsubseteq).$

\newpage
\section{Construction of the lift isomorphism}
\label{app:lift-isomorphism}

We start by proving a general proposition whose second statement is a coherence property used in the proof of soundness/adequacy.

\begin{proposition}
\label{prop:lift-more-general}
  Given a dcpo $X$, hereditarily atomic von Neumann algebras $A,B$, and a \emph{discrete} dcpo $Y$ (equivalently, a set), there exists a natural Scott-continuous and linear bijection
\begin{align*}
  \widehat{(-)} \colon \dcpo(X \times Y, \QQ(A, B)) \xrightarrow{\cong}  \dcpo(X, \QQ(\ell^\infty(Y) \otimes A, B) ).
\end{align*}
  Furthermore, if $f \in \dcpo(X \times Y, \QQ(A, B)) $ and $g : X \to [1 \to Y]$ is a Scott-continuous function, then
  \[ \forall x \in X. \widehat f(x) \circ (\ell^\infty(g(x)) \otimes \id_A) = \left( f(x,g(x)(*)) \circ \cong \right) : \ell^\infty(1) \otimes A \to B , \]
as morphisms in $\QQ$, where the unnamed *-isomorphism is $\ell^\infty(1) \otimes A \cong A.$ 
\end{proposition}
\begin{proof}
This is established via the following sequence of Scott-continuous bijections:


\[ 
\begin{bprooftree}
\AxiomC{$ \dcpo(X \times Y, \QQ(A, B)) $}
\doubleLine
\RightLabel{(currying)}
\UnaryInfC{$ \dcpo(X, [Y \to \QQ(A, B)] ) $}
\doubleLine
\RightLabel{($Y$ is discrete)}
\UnaryInfC{$ \dcpo(X, \prod_{|Y|}  \QQ(A, B) ) $}
\doubleLine
\RightLabel{($\QQ$ has $\dcpo$-enriched coproducts)}
\UnaryInfC{$ \dcpo(X, \QQ(\coprod_{|Y|}   A, B) ) $}
\doubleLine
\UnaryInfC{$ \dcpo(X, \QQ(\ell^\infty(Y) \otimes A, B) ) $}
\end{bprooftree}
\]
where $|Y|$ indicates the cardinality of $Y$.
Here the last isomorphism is induced by an isomorphism $\ell^\infty(Y)\otimes A\to\coprod_{|Y|}A$ in $\QQ$ (and in $\QQ_*$, which has the same isomorphisms as $\QQ$) as follows. 
From Proposition \ref{prop:product of von Neumann algebras} recall that in $\HA_*$ the $|Y|$-fold product of $A$ is given by $\prod_{|Y|}A=\{(a_y)_{y\in Y}:\sup_{y\in Y}\|a_y\|<\infty\}$, from which it easily follows that $\prod_{|Y|}\mathbb C=\ell^\infty(Y)$. Since $\QQ_*$ is monoidal closed, coproducts commute with the monoidal product $\otimes$, hence in $\HA_*$ products commute with the spatial tensor product $\bar\otimes$, i.e., \[\prod_{|Y|}A\cong\prod_{|Y|}(\mathbb C\bar{\otimes}A)\cong\left(\prod_{|Y|}\mathbb C\right)\bar\otimes A=\ell^\infty(Y)\bar\otimes A.\] The resulting $*$-isomorphism $\theta:\ell^\infty(Y)\bar\otimes A\to \prod_{|Y|}A$ is given explicitly by $\theta(a)=\sum_{y\in Y}e_y\otimes a_y$ for each  $a=(a_y)_{y\in Y}$ in $\prod_{|Y|}A$, where $e_y\in\ell^\infty(Y)$ is the function $Y\to\mathbb C$ whose $y$-component is $1$, while all other components vanish. The sum converges in the $\sigma$-weak operator topology.
Let $f\in\DCPO(X\times Y,\QQ(A,B))$, so $f$ is a Scott continuous function $X\times Y\to\HA(B,A)$. We can now construct $\widehat f$ as follows. Fix $x\in X$. Then we obtain a function $f(x,-):Y\to \HA(B,A)$, hence for $b\in B$, we have $f(x,-)(b):Y\to A$. This corresponds with an element $( f(x,y)(b))_{y\in Y}$ in $\prod_{|Y|}A$, hence using $\theta$, to an element $\sum_{y\in Y}e_y\otimes f(x,y)(b)$ in $\ell^\infty(Y)\bar\otimes A$. Thus $\widehat f(x)(b)=\sum_{y\in Y}e_y\otimes f(x,y)(b)$.


Now let $g\in\DCPO(X,[1\to Y])$. Then $g(x)$ is a function $1\to Y$, 
hence $\ell^\infty(g(x))$ is a morphism $\ell^\infty(1)\to\ell^\infty(Y)$ in $\QQ$, corresponding to the $*$-homomorphism $\ell^\infty(Y)\to\ell^\infty(1)$, $k\mapsto k\circ (g(x))$ in $\HA_*$. 
Thus we obtain a morphism \[\QQ(\ell^\infty(g(x))\otimes\id_A,B):\QQ(\ell^\infty(Y)\otimes A,B)\to\QQ(\ell^\infty(1)\otimes A,B),\qquad h\mapsto h\circ (\ell^\infty(g(x))\otimes\id_A).\] In particular, it follows from choosing $h=\widehat f(x)$ that $\widehat f(x)\circ (\ell^\infty(g(x))\otimes\id_A)$ is a morphism $\ell^\infty(1)\otimes A\to B$ in $\QQ$, which in the opposite category $\HA$ corresponds to the map \[(\ell^\infty(g(x))\otimes\id_A)\circ\widehat f(x):B\to \ell^\infty(1)\otimes A.\] For fixed $b\in B$, we then have \[(\ell^\infty(g(x))\otimes\id_A)\circ \widehat f(x)(b)=(\ell^\infty(g(x))\otimes\id_A)\circ\sum_{y\in Y}e_y\otimes f(x,y)(b)=\sum_{y\in Y}e_y(g(x))\otimes f(x,y)(b),\] where we used that $\ell^\infty(g(x))\otimes\id_A$ is continuous with respect to the $\sigma$-weak operator topology. Now, the inverse of the isomorphism $\cong:\ell^\infty(1)\otimes A\to A$ in $\QQ$ corresponds to the $*$-isomorphism $\zeta:\ell^\infty(1)\bar\otimes A\to A$ in $\HA$, which acts on elementary tensors by $k\otimes a\mapsto k(*)a$. By continuity of $\zeta$ with respect to the $\sigma$-weak operator topology, we obtain \[\zeta\circ (\ell^\infty(g(x))\otimes\id_A)\circ \widehat f(x)(b)=\sum_{y\in Y}e_y(g(x)(*))f(x,y)(b)=\sum_{y\in Y}\delta_{y,g(x)(*)}f(x,y)(b)=f(x,g(x)(*))(b).\] Thus $(\ell^\infty(g(x))\otimes\id_A)\circ \widehat f(x)=\zeta^{-1}\circ f(x,g(x)(*))$ in $\HA$, whence the expression in the statement holds in the opposite category $\QQ$.
\end{proof}



\newpage
\section{Proof of Strong Adequacy}
\label{app:adequacy}

In this appendix we provide a proof of Theorem \ref{thm:strong-adequacy}. We begin by stating a corollary for the soundness theorem.

\begin{corollary}
\label{cor:soundness-inequality}
  Let $\cdot \vdash m : P$ be a closed classical term and $\cdot \vdash \mathcal C: \mathbf A ; \qbit^k$ a closed quantum configuration. Then:
  \begin{align*}
  \lrb m \geq \sum_{v \in \text{Val}(m)} P(m \probto{}_* v) \lrb{v} \qquad \qquad
  \lrb{\mathcal C} \geq \sum_{\mathcal V \in \mathrm{ValC}(\mathcal C)} P(\mathcal C \probto{}_* \mathcal V) \lrb{\mathcal V} .
  \end{align*}
\end{corollary}
\begin{proof}
  The classical statement is identical to \cite[Corollary 67]{m-monad} and the quantum statement is fully analogous.
\end{proof}

The remainder of the appendix is dedicated to showing the converse inequalities, which are much more difficult to prove.

\subsection{Overview of the Proof Strategy}

Our proof strategy uses logical relations to establish strong adequacy.  Our
logical relations are described in Theorem \ref{thm:formal-relations} and their
design follows that of \cite{m-monad} which is in turn based on the logical
relations of Claire Jones in her thesis \cite{jones90}.  We establish some
useful closure properties for these relations in Subsection
\ref{sub:closure-properties} and this allows us to prove the Fundamental Lemma
(Lemma \ref{lem:fundamental}). Once the Fundamental Lemma is proved, strong
adequacy follows easily.

A large part of the effort in proving Strong Adequacy lies in the proof of
Theorem \ref{thm:formal-relations}. The classical logical relations there are
defined via non-well-founded induction.  The proof of the existence of these
relations is not obvious. We use methods from
\cite{m-monad,icfp19,lnl-fpc-lmcs} (which are in turn based on ideas from
\cite{fiore-thesis}) to show the existence of these logical relations.

The quantum logical relations are actually easier to define and we do this
first (Subsection \ref{sub:quantum-logical}). The reason for this is that the
quantum subsystem is \emph{first-order} and all the \emph{quantum values}
depend only on themselves and not on classical terms or quantum terms. This is
not the case for the classical values, because they may depend on terms that
are not values (e.g. lambda abstractions). Existence of the quantum logical
relations is clear and immediate from their definition.

The main idea for the proof of existence of the classical relations is the
following. For every type $P$, we define a category $\RR(P)$ of logical
relations with a suitable notion of morphism.  We show that every such category
has sufficient structure to construct parameterised initial algebras
(Proposition \ref{prop:logical-colimits}). It follows that we may define
functors on these categories (Proposition \ref{prop:logical-functors1}) which
construct logical relations as they are needed in Theorem
\ref{thm:formal-relations}. All of these functors are $\omega$-cocontinuous
(Proposition \ref{prop:logical-functors2}) and therefore we may form
(parameterised) initial algebras using them. This, in turn, allows us to define
\emph{augmented interpretation of types} on the categories $\RR(P)$. These
interpretations satisfy important coherence conditions with respect to the
standard interpretation of types (Corollary \ref{cor:cool-form}). These
coherence conditions are important, because they show that every augmented
interpretation $\elrbs P$ of a type $P$ contains the standard interpretation
$\lrb P$, together with the logical relation that we need for the proof, as
shown in Theorem \ref{thm:formal-relations}.

The proof stategy that we use to define the classical logical relations and
their existence is heavily based on \cite{m-monad}, so we use the same notation
as there.

\subsection{Notation for Reduction Paths}
\label{sub:reduction-paths}

Before we may define our logical relations, we have to introduce some auxiliary definitions for reduction paths.

\begin{assumption}
Throughout this appendix, we assume that all types are closed, unless otherwise noted.
\end{assumption}

\begin{definition}
For each classical type $P$ and quantum type $\AAA$ we write:
\begin{itemize}
  \item $\Val(P) \defeq \{ V \ |\ V \text{ is a classical value and } \cdot \vdash V : P\}.$
  \item $\Prog(P) \defeq \{ M \ |\ M \text{ is a classical term and } \cdot \vdash M : P\}.$
  \item $\ValC(\AAA; \qbit^k) \defeq \{ \mathcal V \ |\ \mathcal V \text{ is a value configuration and } \cdot \vdash \mathcal V : \AAA; \qbit^k \}.$
  \item $\mathrm{C}(\AAA; \qbit^k) \defeq \{ \mathcal C \ |\ \mathcal C \text{ is a configuration and } \cdot \vdash \mathcal C : \AAA; \qbit^k \}.$
\end{itemize}
\end{definition}

\begin{definition}
\label{def:paths}
Let $M \colon P$ and $N \colon P$ be closed classical terms of the same type. We define
\[ \mathrm{Paths(M, N)} \defeq \left\{ \pi\ |\ \pi = \left( M = M_0  \probto{p_0} M_1 \probto{p_1} M_2 \probto{p_2} \cdots \probto{p_n} M_n = N \right)  \text{ is a reduction path} \right\} . \]
In other words, $\mathrm{Paths(M, N)}$ is the set of all reduction paths from $M$ to $N$. The \emph{probability weight} of a path $\pi \in \Paths(M,N)$ is $P(\pi) \eqdef \prod_{i=0}^n p_i,$ i.e., it is simply the product of all the probabilities of single-step reductions within the path.
The \emph{set of terminal reduction paths of $M$} is
\[ \mathrm{\TPaths(M)} \defeq \bigcup_{V \in \Val(P)} \Paths(M,V) . \]
Thus the endpoint of any path $\pi \in \TPaths(M)$ is a value. If $\pi \in \Paths(M,W)$, where $W$ is a value, then we shall write $V_\pi \eqdef W.$ That is, for a path $\pi \in \TPaths(M)$, the notation $V_\pi$ indicates the endpoint of the path $\pi$ which is indeed a value.
\end{definition}

Reduction paths for quantum configurations are defined in the same way.

\begin{definition}
\label{def:paths-quantum}
Let $\cdot \vdash \mathcal C \colon \AAA;\qbit^k$ and $\cdot \vdash \mathcal D \colon \AAA; \qbit^k$ be closed quantum configurations of the same type. We define
\[ \mathrm{Paths(\mathcal C, \mathcal D)} \defeq \left\{ \pi\ |\ \pi = \left( \mathcal C = \mathcal C_0  \probto{p_0} \mathcal C_1 \probto{p_1} \mathcal C_2 \probto{p_2} \cdots \probto{p_n} \mathcal C_n = \mathcal D \right)  \text{ is a reduction path} \right\} . \]
In other words, $\mathrm{Paths(\mathcal C, \mathcal D)}$ is the set of all reduction paths from $\mathcal C$ to $\mathcal D$. The \emph{probability weight} of a path $\pi \in \Paths(\mathcal C, \mathcal D)$ is $P(\pi) \eqdef \prod_{i=0}^n p_i,$.
The \emph{set of terminal reduction paths of $\mathcal C$} is
\[ \mathrm{\TPaths(\mathcal C)} \defeq \bigcup_{\mathcal V \in \ValC(\AAA; \qbit^k)} \Paths(\mathcal C, \mathcal V) . \]
Thus the endpoint of any path $\pi \in \TPaths(\mathcal C)$ is a value configuration. If $\pi \in \Paths(\mathcal C, \mathcal W)$, where $\mathcal W$ is a value configuration, then we shall write $\mathcal V_\pi \eqdef \mathcal W.$ That is, for a path $\pi \in \TPaths(\mathcal C)$, the notation $\mathcal V_\pi$
indicates the endpoint of the path $\pi$ which is indeed a value configuration.
\end{definition}

\begin{remark}
  We also note that for each closed classical term $M$ and each closed quantum configuration $\mathcal C$ the sets $\TPaths(M)$ and $\TPaths(\mathcal C)$ are both countable.
\end{remark}

\subsection{The Quantum Logical Relation}
\label{sub:quantum-logical}

Next, we define a logical relation between quantum configurations and their semantic domain.

\begin{definition}
\label{def:logical-configuration}
  For each closed quantum type $\AAA$ and $k \in \mathbb N$, we define a \emph{logical relation}
  \begin{align*}
    \btleq_\AAA^k &\subseteq \QQ(\mathbb C, \lrb{\AAA \otimes \qbit^k}) \times \mathrm{C}(\AAA ; \qbit^k) \qquad \text{by}\\
    c \btleq_\AAA^k & \mathcal C \text{ iff } c \in \mathcal S(\btleq_\AAA^k; \mathcal C),  \text{ where } \mathcal S(\btleq_\AAA^k; \mathcal C) \text{ is the Scott-closure in $\QQ(\mathbb C, \lrb{\AAA \otimes \qbit^k})$ of the set } \\
      &\mathcal S_{0}(\btleq_\AAA^k; \mathcal C) \defeq \left\{ \sum_{\pi \in F} P(\pi) \lrb{\mathcal V_\pi} \ |\ F \subseteq \TPaths(\mathcal C),\ F \text{ is finite} \right \} .
  \end{align*}
\end{definition}

\begin{lemma}
  \label{lem:value-configuration}
  If $\cdot \vdash \mathcal V : \AAA; \qbit^k$ is a value configuration, then $\lrb{\mathcal V} \btleq^k_{\AAA} \mathcal V.$ 
\end{lemma}
\begin{proof}
  This is immediate by Definition \ref{def:logical-configuration}.
\end{proof}

\subsection{Classical Logical Relations}

\begin{assumption}
  In this section and the next two, all types are assumed to be classical, unless otherwise noted.
\end{assumption}

We define sets of relations that are parameterised by dcpo's $X$ from our semantic category, types $P$ from our language and partial deterministic embeddings $e_X : X \kto \lrb P$ which show how $X$ approximates $\lrb P$.
We shall write relation membership in infix notation, that is, for a binary relation $\tleq$, we write $v \tleq V$ to indicate $(v, V) \in \tleq.$

\begin{definition}
\label{def:logical-relations}
For any dcpo $X$, type $P$ and morphism $e \colon X \kto \lrb P$ in $\PDe$, let:
\begin{align*}
  \ValRel(X, P, e)  &= \{ \tleq_{X,P}^e \subseteq \TD(1, X) \times \Val(P) \ |\ \forall V \in \Val(P).\ (-) \tleq_{X,P}^e V \text{ is a Scott closed subset of }\\
  & \TD(1,X) \text{ and} \forall V \in \Val(P).\ v \tleq_{X,P}^e V \Rightarrow e \kcirc v \leq \lrb V \}. 
\end{align*}
\end{definition} 

\begin{remark}
In the above definition, relations $\tleq_{X,P}^e \in \ValRel(X,P,e)$ can be seen as ternary relations $\tleq_{X,P}^e \subseteq \TD(1,X) \times \Val(P) \times \{e\}$. However, since there is no choice for the third component, we prefer to see them as binary relations that are parameterised by the embeddings $e$. Indeed, this leads to a much nicer notation.
We shall also sometimes indicate the parameters $X, P$ and $e$ of the relation in order to avoid confusion as to which set $\ValRel(X,P,e)$ it belongs to.
\end{remark}

The relations we need for the adequacy proof inhabit the sets $\ValRel(\lrb P, P, \kid_{\lrb P})$. In the remainder of the appendix, we will show how to choose exactly one relation (the one we need) from each of those sets.

The next definition we introduce is crucial for the proof of strong adequacy.

\begin{definition}\label{def:SM}
Given a relation $\tleqd \in \ValReld$ and a term $\cdot \vdash M : P$, let $\mathcal S(\tleqd; M)$ be the Scott-closure in $\KL(1,X)$ of the set
\begin{equation}
\label{eq:logical-sums}
\mathcal S_{0}(\tleqd; M) \defeq \left\{ \sum_{\pi \in F} P(\pi) v_\pi \ |\ F \subseteq \TPaths(M),\ F \text{ is finite and }  \text{$v_\pi \tleqd V_\pi$ for each $\pi \in F$} \right\} .
\end{equation}
In other words, $\mathcal S(\tleqd; M)$ is the smallest Scott-closed subset of $\KL(1,X)$ which contains all morphisms of the form in \eqref{eq:logical-sums}.
For a subset $U \subseteq \KL(1,X),$ we write $\overline U$ to indicate its Scott-closure in $\KL(1,X)$.
\end{definition}

\begin{lemma}
\label{lem:semantically-dense-value}
For any value $V$, we have $\mathcal S(\tleqd; V) = \ol{\{ v\ |\ v \tleqd V\}} \cup \{ 0\} = \ol{\{ v\ |\ v \tleqd V\} \cup \{ 0\}} .$
\end{lemma}
\begin{proof}
This is because all of the sums in \eqref{eq:logical-sums} are singleton sums or the empty sum.
\end{proof}

\begin{lemma}[{\cite[Lemma 8.4]{jones90}}]
\label{lem:topological-goodness}
Let $Y$ be a dcpo and let $\{X_i\}_{i \in F}$ be a finite collection of dcpo's.
Let $f \colon \prod_i X_i \to Y$ be a Scott-continuous function. Let $C_Y$ be a Scott-closed subset of $Y$. Let $U_i \subseteq X_i$ be arbitrary subsets, such that $f(\prod_i U_i) \subseteq C_Y$.
Then $f(\prod_i \overline{U_i}) \subseteq C_Y$, where $\overline{U_i}$ is the Scott-closure of $U_i$ in $X_i$.
\end{lemma}

\begin{lemma}[{\cite[Lemma 77]{m-monad}}]
\label{lem:logical-composition}
Let $\tleq_{X_1, P}^{e_1}$ and $\tleq_{X_2,P}^{e_2}$ be two logical relations and $\cdot \vdash M : P$ a term.
Assume that $g: X_1 \kto X_2$ is a morphism, such that $v \tleq_{X_1, P}^{e_1} V$ implies $g \kcirc v \in \mathcal S(\tleq_{X_2, P}^{e_2} ; V),$ for any $V \in \Val(M).$
If $m \in \mathcal S(\tleq_{X_1, P}^{e_1} ; M)$, then $g \kcirc m \in \mathcal S(\tleq_{X_2, P}^{e_2} ; M).$
\end{lemma}

Next, we define important \emph{closure relations} which we use for terms.

\begin{definition}
\label{def:logical-closure}
If $\tleqd \in \ValRel(X, P, e)$, let ${\qtleqd} \subseteq \KL(1, X) \times \Prog(P)$ be the relation defined by
  \[ m \qtleqd M \text{ iff } m \in \mathcal S(\tleqd; M) . \]
\end{definition} 

\begin{lemma}
For any term  $\cdot \vdash M : P$ and $\tleqd \in \ValReld$, the set $(-) \qtleqd M$ is a Scott-closed subset of $\KL(1,X).$
\end{lemma}
\begin{proof}
This follows immediately by definition, because $\mathcal S(\tleqd; M)$ is Scott-closed.
\end{proof}

\begin{lemma}[{\cite[Lemma 80]{m-monad}}]
\label{lem:scott-closure-embedding}
Let $C$ be a Scott-closed subset of a dcpo $X$. Let $W \eqdef \{ \delta_x\ |\ x \in C \} \subseteq \MM X$ and let $\overline W$ be the Scott-closure of $W$ in $\MM X.$ Then, $\delta_y \in \overline W$ iff $y \in C.$
\end{lemma}

\begin{lemma}[{\cite[Lemma 81]{m-monad}}]
\label{lem:id}
Let $X$ be a dcpo, let $v \in \TD(1,X)$ and let $V$ be a value. Then $v \tleqd V $ iff $v \qtleqd V.$
\end{lemma}

\begin{lemma}[{\cite[Lemma 82]{m-monad}}]
\label{lem:logical-inequality}
For any value $\cdot \vdash V : P$ and $\tleqd \in \ValReld$, if $m \qtleqd V$ then $ e \kcirc m \leq \lrb V$.
\end{lemma}

\subsection{Logical Relations for types 1 and $Q(\AAA,\BBB)$}

The unit type $1$ and the type of quantum functions $Q(\AAA, \BBB)$ have simple
type structures, because they do not depend on any other classical types. As a
result, it is easy to define the required logical relations at those types and
we do so now.

\begin{definition}
\label{def:logical-unit}
  We define a logical relation $\tleq_1 \in \ValRel(\lrb 1, 1, \kid_{\lrb 1})$ by:
  \[ f \tleq_1 ()\ \text{iff}\ f = \kid_{\lrb 1} .  \]
\end{definition}

\begin{definition}
\label{def:logical-quantum-function}
  For every two closed quantum types $\AAA$ and $\BBB$, we define a logical relation $\tleq_{Q(\AAA,\BBB)} \in \ValRel(\lrb{Q(\AAA,\BBB)}, Q(\AAA,\BBB), \kid_{\lrb{Q(\AAA,\BBB)}})$ by:
  \[ f \tleq_{Q(\AAA, \BBB)} \ff \text{ iff }  f \leq \lrb{\ff} \text{ and } \forall k \in \mathbb N. \forall (\cdot \vdash [\ket{\psi}, \ell, \vq] : \AAA; \qbit^k) . \] \[ (\beta(f(*)) \otimes \id) \circ \lrb{[\ket{\psi}, \ell, \vq]} \btleq_{\BBB}^k [\ket \psi, \ell, \ff \vq] , \]
  where the second quantifier ranges over well-formed \emph{value confiugrations} of the indicated type.
\end{definition}

It is easy to see that both logical relations are well-defined.
Furthermore, notice that the second family of logical relations is defined via the quantum logical relations on configurations.

\subsection{Categories of Logical Relations}

\begin{definition}
\label{def:logical-categories}
For any type $P$, we define a category $\RR(P)$ where:
\begin{itemize}
\item Each object is a triple $(X, e_X, \tleq_X)$, where $X$ is a dcpo, $e_X \colon X \kto \lrb P$ is a morphism in $\PD_e$ and $\tleq_X \in \ValRel(X,P,e_X)$.
\item A morphism $f: (X, e_X, \tleq_{X}) \to (Y, e_Y, \tleq_{Y})$ is a morphism $f : X \kto Y$ in $\PD_e$, which satisfies the three additional conditions:
  \begin{itemize}
    \item If $v \tleq_{X} V,$ then $f \kcirc v\ \overline{\tleq_{Y}}\ V.$
    \item If $v \tleq_{Y} V,$ then $f^p \kcirc v\ \overline{\tleq_{X}}\ V.$
    \item $e_X = e_Y \kcirc f.$
  \end{itemize}
\item Composition and identities coincide with those in $\PD_e.$
\end{itemize}
\end{definition}

\begin{lemma}[{\cite[Lemma 84]{m-monad}}]
For every type $P$, the category $\RR(P)$ is indeed well-defined.
\end{lemma}

\begin{lemma}
\label{lem:logical-composition-fancy}
Let $\cdot \vdash M : P$ be a term and let $g \colon (X, e_X, \tleq_X) \to (Y, e_Y, \tleq_Y)$ be a morphism in $\RR(P)$. If $m \ol{\tleq_X} M$ then $g \kcirc m \ol{\tleq_Y} M$.
Moreover, if $n \ol{\tleq_Y} N,$ then $g^p \kcirc n \ol{\tleq_X}$ N.
\end{lemma}
\begin{proof}
This follows immediately by Lemma \ref{lem:logical-composition}.
\end{proof}

\begin{definition}
\label{def:logical-forgetful}
For every type $P$, we define the obvious forgetful functor $U^P \colon \RR(P) \to \PDe$ by
\begin{align*}
U^P(X,e, \tleq) &= X \\
U^P(f) &= f.
\end{align*}
\end{definition}

\begin{proposition}[{\cite[Proposition 87]{m-monad}}]
\label{prop:logical-colimits}
For each type $P$, the category $\RR(P)$ has an initial object and all $\omega$-colimits. Furthermore, the forgetful functor $U^P \colon \RR(P) \to \PDe$ preserves and reflects $\omega$-colimits (and also the initial objects).
\end{proposition}

Next, we introduce important relation constructors and some new notation.

\begin{notation}
Given morphisms $m_i : 1 \kto X_i$, for $i \in \{1, \ldots, n\},$ we define
\[ \llangle m_1, \ldots, m_n \rrangle \eqdef (m_1 \ktimes \cdots \ktimes m_n) \kcirc \JJ \langle \id_1, \ldots, \id_1 \rangle : 1 \kto X_1 \times \cdots \times X_n . \]
\end{notation}
\begin{notation}
\label{not:application}
Given morphisms $x : 1 \kto X$ and $f: 1 \kto [X \kto Y]$ in $\KL$, let $f[x] : 1 \kto Y$ be the morphism defined by
\[ f[x] \eqdef \epsilon \kcirc (f \ktimes x) \kcirc \JJ\langle \id_1, \id_1 \rangle . \]
\end{notation}

\begin{definition}[Relation Constructions]
\label{def:logical-constructions}
We define relation constructors:
\begin{itemize}
\item If $\tleqone \in \ValRel(X_1, P_1, e_1)$ and $\tleqtwo \in \ValRel(X_2, P_2, e_2)$, define
  \begin{align*}
  (\tleqone &+ \tleqtwo) \in \ValRel(X_1+ X_2, P_1 + P_2, e_1 \kplus e_2) \text{ by: } \\
  \JJ \emph{in}_i \kcirc v\ (\tleqone &+ \tleqtwo)\ \mathtt{in}_i V \text{ iff }  v \tleq_{X_i, P_i}^{e_i} V   \qquad\qquad (\text{for $i \in \{1,2\}$}). 
  \end{align*}
\item If $\tleqone \in \ValRel(X_1, P_1, e_1)$ and $\tleqtwo \in \ValRel(X_2, P_2, e_2)$, define
  \begin{align*}
  (\tleqone &\times \tleqtwo) \in \ValRel(X_1 \times X_2, P_1 \times P_2, e_1 \ktimes e_2) \text{ by: } \\
  \llangle v_1, v_2 \rrangle \ (\tleqone &\times \tleqtwo)\ (V_1,V_2) \text{ iff }  v_1 \tleqone V_1 \text{ and } v_2 \tleqtwo V_2.
  \end{align*}
\item If $\tleqone \in \ValRel(X_1, P_1, e_1)$ and $\tleqtwo \in \ValRel(X_2, P_2, e_2)$, define
  \begin{align*}
  (\tleqone &\to\ \tleqtwo) \in \ValRel([X_1 \kto X_2], P_1 \to P_2, \JJ [e_1^p \kto e_2]) \text{ by: } \\
  f\ (\tleqone &\to\ \tleqtwo)\ \lambda x. M \text{ iff }  \JJ [e_1^p \kto e_2] \kcirc f \leq \lrb{\lambda x.M} \text{ and } \forall (v \tleqone V).\  f[v] \qtleqtwo (\lambda x.M)V. 
  \end{align*}
\end{itemize}
\end{definition}

\begin{lemma}[{\cite[Lemma 91]{m-monad}}]
The assignments in Definition \ref{def:logical-constructions} are indeed well-defined.
\end{lemma}

\begin{notation}
Throughout the rest of the paper we shall write $(- \ktwo_e -) \eqdef [- \kto -]_e^\JJ : \PDe \times \PDe \to \PDe$. That is, we just introduce a more concise notation for the functor $[- \kto - ]_e^\JJ$ from Proposition \ref{prop:omega-functors}.
\end{notation}

The next definition is crucial. Given two logical relations, it is used to define the product, coproduct and function space logical relations. Moreover, this is done in a functorial sense on the categories $\RR(P)$.

\begin{proposition}[{\cite[Proposition 95]{m-monad}}]
\label{prop:logical-functors1}
Let $P$ and $R$ be types. The following assignments (recall Definition \ref{def:logical-constructions}):
\begin{enumerate}
\item $\times^{P,R} \colon \RR(P) \times \RR(R) \to \RR(P \times R)$ by
  \begin{align*}
  (X, e_X, \tleq_X) \times^{P,R} (Y, e_Y, \tleq_Y) &\eqdef (X \times Y, e_X \ktimes_e e_Y, \tleq_X \times \tleq_Y) \\
  f \times^{P,R} g &\eqdef f \ktimes_e g
  \end{align*}
\item $+^{P,R} \colon \RR(P) \times \RR(R) \to \RR(P + R)$ by
  \begin{align*}
  (X, e_X, \tleq_X) +^{P,R} (Y, e_Y, \tleq_Y) &\eqdef (X + Y, e_X \kplus_e e_Y, \tleq_X + \tleq_Y) \\
  f +^{P,R} g &\eqdef f \kplus_e g
  \end{align*}
\item $\to^{P,R} \colon \RR(P) \times \RR(R) \to \RR(P \to R)$ by
  \begin{align*}
  (X, e_X, \tleq_X) \to^{P,R} (Y, e_Y, \tleq_Y) &\eqdef ([X \kto Y], e_X \ktwo_e e_Y, \tleq_X \to \tleq_Y) \\
  f \to^{P,R} g &\eqdef f \ktwo_e g
  \end{align*}
\end{enumerate}
define \emph{covariant} functors with the indicated types.
\end{proposition}

Observe that Proposition \ref{prop:logical-functors1} lifts the functors that we use to interpret our types in the category $\KL$ to the categories $\RR(P)$. Next, we show that the functors we just defined are also suitable for forming (parameterised) initial algebras.

\begin{proposition}[{\cite[Proposition 96]{m-monad}}]
\label{prop:logical-functors2}
For $\star \in \{\times, +, \to\},$ for all types $P$ and $R$, the functor $\star^{P,R} : \RR(P) \times \RR(R) \to \RR(P \star R)$ is $\omega$-cocontinuous and the following diagram:
\[ \stikz{logical-functors.tikz} \]
commutes.
\end{proposition}

Next, we establish an isomorphism between the categories $\RR(\mu X. P)$ and $\RR(P[\mu X.P / X]).$

\begin{definition}
\label{def:logical-fold-unfold}
We define constructors for folding and unfolding logical relations as follows:
\begin{itemize}
\item If $\tleq_{X, P[\mu Y. P / Y]}^e \in \ValRel(X, P[\mu Y.P / Y], e)$, define
  \begin{align*}
  (\mathbb I^{\mu Y. P}    &\tleq_{X, P[\mu Y. P / Y]}^e ) \in \ValRel(X, \mu Y. P, \sfold \kcirc e) \text{ by:} \\
  v\ (\mathbb I^{\mu Y. P} &\tleq_{X, P[\mu Y. P / Y]}^e )\ \mathtt{fold}\ V \text{ iff } v \tleq_{X, P[\mu Y. P / Y]}^e V .
  \end{align*}
\item If $\tleq_{X, \mu Y.P}^e \in \ValRel(X, \mu Y.P, e )$, define
  \begin{align*}
  (\mathbb E^{\mu Y. P}    &\tleq_{X, \mu Y.P}^e) \in \ValRel(X, P[\mu Y. P / Y], \sunfold \kcirc e) \text{ by:}\\
  v\ (\mathbb E^{\mu Y. P} &\tleq_{X, \mu Y.P}^e) )\ V \text{ iff } v \tleq_{X, \mu Y.P}^e \mathtt{fold}\ V .
  \end{align*}
\end{itemize}
\end{definition}

\begin{proposition}[{\cite[Proposition 99]{m-monad}}]
\label{prop:logical-folding-isomorphism}
For every type $\cdot \vdash \mu X.P,$ we have an isomorphism of categories
\[\FOLD{\mu X.P} : \RR(P[\mu X.P / X]) \cong \RR(\mu X. P) : \UNFOLD{\mu X. P} , \]
where the functors are defined by 
\begin{align*}
&\FOLD{\mu X.P} :  \RR(P[\mu X.P / X]) \to \RR(\mu X. P)   & & \UNFOLD{\mu X.P} : \RR(\mu X. P) \to \RR(P[\mu X.P / X]) \\
&\FOLD{\mu X.P}(Y, e, \tleq) = (Y, \sfold \kcirc e , \FOLD{\mu X.P} \tleq) & & \UNFOLD{\mu X.P}(Y, e, \tleq) = (Y, \sunfold \kcirc e, \UNFOLD{\mu X.P} \tleq)\\
&\FOLD{\mu X.P}(f) = f & & \UNFOLD{\mu X.P}(f) = f .
\end{align*}
\end{proposition}

This finishes the categorical development of the categories $\RR(P)$.

\subsection{Augmented Interpretation of Types}

We have now established sufficient categorical structure in order to construct parameterised initial algebras in the categories $\RR(P).$ Furthermore, we have sufficient structure to also define an \emph{augmented} interpretation of types in these categories.
The main idea behind providing the augmented interpretation is to show how to pick out the logical relations we need from all those that exist in the categories $\RR(P)$.

\begin{notation}
Given any type context $\Theta = X_1, \ldots , X_n$ and closed types $\cdot \vdash C_i$ with $i \in \{1, \ldots, n \}$,
we shall write $\vec C$ for $C_1, \ldots , C_n$ and we also write $[\vec C / \Theta]$ for $[C_1 / X_1, \ldots , C_n / X_n]$.
\end{notation}

\begin{definition}
For any type $\Theta \vdash P$ and closed types $\vec C$, we define their \emph{augmented interpretation} to be the functor
\[ \elrbc{\Theta \vdash P} : \RR(C_1) \times \cdots \times \RR(C_n) \to \RR(P[ \vec C / \Theta ]) \]
defined by induction on the derivation of $\Theta \vdash P$:
\begin{align*}
\elrbc{\Theta \vdash \Theta_i} &:= \Pi_i &&\\
  \elrbc{\Theta \vdash 1} &:= K_{(\lrb 1, \kid_{\lrb 1}, \tleq_1)} &&\\
  \elrbc{\Theta \vdash Q(\AAA, \BBB)} &:= K_{(\lrb{Q(\AAA, \BBB)}, \kid_{\lrb{Q(\AAA,\BBB)}}, \tleq_{Q(\AAA,\BBB)})} &&\\
\elrbc{\Theta \vdash P \star R } &:= \star^{P[\vec C / \Theta], R[\vec C / \Theta]} \circ \langle \elrbc{\Theta \vdash P}, \elrbc{\Theta \vdash R} \rangle &(\text{for }\star \in \{ +, \times, \to \})&\\
\elrbc{\Theta \vdash \mu X.P} &:= \left(\FOLD{\mu X. P[\vec C / \Theta]} \circ \elrb{\Theta, X \vdash P}{\vec C, \mu X. P[\vec C / \Theta]} \right)^\sharp , &&
\end{align*}
where $K_Y$ is the constant functor on $Y$ and the $(-)^\sharp$ operation is from Definition \ref{def:initial-algebra}.
\end{definition} 
 
\begin{proposition}
\label{prop:augmented-interpretation}
Each functor $\elrbc{\Theta \vdash P}$ is well-defined and $\omega$-cocontinuous. Moreover, the following diagram:
\cstikz{augmented-diagram.tikz}
commutes.
\end{proposition}
\begin{proof}
The proof is essentially the same as \cite[Proposition 7.26]{lnl-fpc-lmcs}.
\end{proof}

Next, a corollary which shows that parameterised initial algebras for our type expressions are constructed in the same way in both categories.

\begin{corollary}
\label{cor:augmented-algebras}
The 2-categorical diagram:
\cstikz{parameterised-initial-algebra-augmented.tikz}
commutes, where $\iota$ is the parameterised initial algebra isomorphism (see Definition \ref{def:initial-algebra}).
\end{corollary}
\begin{proof}
The proof is the same as \cite[Corollary 7.27]{lnl-fpc-lmcs}.
\end{proof}

Proposition \ref{prop:augmented-interpretation} shows that the first component of the augmented interpretation coincides with the standard interpretation. This is true for all types, including open ones.
In the special case for closed types, let $\elrbs{P} \eqdef \elrb{\cdot \vdash P}{\cdot}(*)$, where $*$ is the unique object of the terminal category $\mathbf 1 = \RR(P)^0$.
Proposition \ref{prop:augmented-interpretation} therefore shows that $U \elrbs P = \lrb P$, which means that $\elrbs P$ has the form $\elrbs P = (\lrb P, e, \tleq),$ where $e : \lrb P \kto \lrb P$ is some
embedding. Next, we show that $e = \kid.$ In order to do this, we prove a stronger proposition first. We show that the action of the functor $\elrbc{\Theta \vdash P}$ on the embedding component is also completely determined by the action of $\lrb{\Theta \vdash P}$ on embeddings.

\begin{proposition}
\label{prop:augmented-embedding}
For every functor $\elrbc{\Theta \vdash P}$ and objects $(X_i, e_i, \tleq_i)$ with $i \in \{1, \ldots, n\}$, we have:
\[ \pi_e \left( \elrbc{\Theta \vdash P}\left( (X_1, e_1, \tleq_1), \ldots, (X_n, e_n, \tleq_n) \right) \right) =  \lrb{\Theta \vdash P}(e_1, \ldots, e_n) , \]
where for an object $(Z, e_Z, \tleq_Z)$ in any category $\RR(R)$, we define $\pi_e(Z, e_Z, \tleq_Z) = e_Z.$ 
\end{proposition}
\begin{proof}
By induction on the derivation of $\Theta \vdash P.$

\paragraph*{\textbf{Case} $P=1$ } The functors on both sides are constant ones, so this is a trivial verification.

\paragraph*{\textbf{Case} $P=Q(\AAA,\BBB)$ } The functors on both sides are constant ones, so this is a trivial verification.

  The remaining cases follow using exactly the same arguments as in \cite[Proposition 104]{m-monad}.
\end{proof}

\begin{corollary}
\label{cor:cool-form}
For every closed classical type $P$, we have $\elrbs P = (\lrb P, \kid_{\lrb P}, \tleq_P)$ for some logical relation $\tleq_P.$ 
\end{corollary}
\begin{proof}
We already know that the first component is $\lrb P$. For the second component, the previous proposition shows that $\pi_e \elrbs P = \pi_e \elrb{\cdot \vdash P}{\cdot}(*) = \lrb{\cdot \vdash P}(\id_*) = \kid_{\lrb P},$ where $*$ denotes the empty tuple of objects and $\id_*$ the empty tuple of embeddings.
\end{proof}

Finally, we want to show that the third component of $\elrbs P$ is the logical relation that we need to carry out the adequacy proof. For this, we have to prove a substitution lemma first.

\begin{lemma}[Substitution]
\label{lem:substitution-augmented}
For any classical types $\Theta, X \vdash P$ and $\Theta \vdash R$ and closed types $C_1, \ldots, C_{n}$, we have:
\[ \elrbc{\Theta \vdash P[R/X]} = \elrb{\Theta, X \vdash P}{\vec C, R[\vec C / \Theta]} \circ \langle \Id, \elrbc{\Theta \vdash R} \rangle . \]
\end{lemma}
\begin{proof}
  By induction on the derivation of $\Theta, X \vdash P$. The cases for the types $1$ and $Q(\AAA,\BBB)$ are trivial, because they involve constant functors. The remaining cases can be proven in the same way as \cite[Lemma 7.30]{lnl-fpc-lmcs}.
\end{proof}

For each type $P$, we have now provided an augmented interpretation $\elrbs P$ of $P$ in the category $\RR(P).$
The interpretation $\elrbs{-}$ satisfies all the fundamental properties of $\lrb{-},$ as we have now shown. It should now be clear that this augmented interpretation is true to its name, because it carries strictly more information compared to the standard interpretation of types.
The additional information that $\elrbs{P}$ carries is precisely the logical relation that we need at type $P$, as we show in the next subsection.

\subsection{Existence of the Logical Relations}

We can now show that the logical relations we need for the adequacy proof exist.

\begin{theorem}
\label{thm:formal-relations}
For each closed classical type $P$, there exist \emph{formal approximation relations:}
\begin{align*}
\tleq_P &\subseteq \TD(1, \lrb P) \times \mathrm{Val}(P) \\
\ol{\tleq_P} &\subseteq \KL(1, \lrb P) \times \mathrm{Prog}(P)
\end{align*}
which satisfy the following properties:
\begin{enumerate}
\setlength\itemsep{0.4em}
\item [(A0)] $ v \tleq_{1} () \text{ iff } v = \kid_{1}$.
\item [(A1)] $ \JJ \emph{in}_i \kcirc v \tleq_{P_1 + P_2} \mathtt{in}_i V \text{ iff } v \tleq_{P_i} V$, where $i \in \{1,2\}$.
\item [(A2)] $\llangle v_1, v_2 \rrangle \tleq_{P_1 \times P_2} (V_1,V_2) \text{ iff }  v_1 \tleq_{P_1} V_1 \text{ and } v_2 \tleq_{P_2} V_2 .$
\item [(A3)] $ f \tleq_{P \to R} \lambda x. M \text{ iff }  f \leq \lrb{\lambda x.M} \text{ and } \forall (v \tleq_P V).\  f[v] \ol{\tleq_R} (\lambda x.M)V. $
\item [(A4)] $ v \tleq_{\mu X.P} \mathtt{fold}\ V \text{ iff } \sunfold \kcirc v \tleq_{P[\mu X. P / X]} V$.
\item [(A5)] $ f \tleq_{Q(\AAA, \BBB)} \ff \text{ iff }  f \leq \lrb{\ff} \text{ and } \forall k \in \mathbb N. \forall (\cdot \vdash [\ket{\psi}, \ell, \vq] : \AAA; \qbit^k) . $ \[ (\beta(f(*)) \otimes \id) \circ \lrb{[\ket{\psi}, \ell, \vq]} \btleq_{\BBB}^k [\ket \psi, \ell, \ff \vq]. \]
\item [(R)] $ m \ol{\tleq_P} M \text{ iff } m \in \mathcal S(\tleq_P; M), $ where $\mathcal S(\tleq_P; M)$ is the Scott-closure in $\KL(1, \lrb P)$ of the set
    \[ \mathcal S_{0}(\tleq_P; M) \defeq \left\{ \sum_{\pi \in F} P(\pi) v_\pi \ |\ F \subseteq \TPaths(M)  \text{ is finite and }  \text{$v_\pi \tleq_P V_\pi$ for each $\pi \in F$} \right\} \ (\text{see Definition \ref{def:paths}}) . \]
\item [(C1)]\label{item:below} If $v \tleq_P V$, then $v \leq \lrb V$.
\item [(C2)] $(- \tleq_P V)$ is a Scott-closed subset of $\TD(1,\lrb P).$
\item [(C3)] If $m \ol{\tleq_P} M$, then $m \leq \lrb M$.
\item [(C4)] $(- \ol{\tleq_P} M)$ is a Scott-closed subset of $\KL(1,\lrb P).$
\item [(C5)] If $v \in \TD(1, \lrb P)$ and $V$ is a value, then $v \tleq_P V$ iff $v \ol{\tleq_P} V.$
\end{enumerate}
\end{theorem}
\begin{proof}
Consider the object $\elrbs P \in \RR(P).$ We have already shown that $\elrbs P = (\lrb P, \kid_{\lrb P}, \tleq_P)$ for some logical relation $\tleq_P \in \ValRel(\lrb P, P, \kid_{\lrb P}).$ We now show that $\tleq_P$ satisfies the required properties.
Notice that the embedding components are just identities.

Properties (A0) and (A5) are satisfied by construction. To show the remaining properties are satisfied we simply use the same arguments as in \cite[Theorem 107]{m-monad}.

\end{proof}

\subsection{Closure Properties of the Logical Relations}
\label{sub:closure-properties}

Here we establish some important closure properties of our logical relations.

\begin{lemma}[{\cite[Lemma 108]{m-monad}}]
\label{lem:logical-convex-sum}
Let $\cdot \vdash M : P$ be a term and let $F$ be some finite index set. Assume that we are given morphisms $m_i$ and terms $M_i$ such that $m_i \ol{\tleq_P} M_i$ for $i \in F$.
Assume further that for each $i \in F$, we are given a reduction path $\pi_i \in \Paths(M, M_i)$, such that all paths $\pi_i$ are distinct.
Then
\[ \sum_{i \in F} P(\pi_i) m_i \ol{\tleq_P} M. \]
\end{lemma}

\begin{lemma}[{\cite[Lemma 109]{m-monad}}]
\label{lem:logical-probabilistic-choice}
If $m \ol{\tleq_P} M$ and $n \ol{\tleq_P} N,$ then $p \cdot m + (1-p) \cdot n \ol{\tleq_P} M\ \mathtt{or}_p\ N.$
\end{lemma}

\begin{lemma}[{\cite[Lemma 110]{m-monad}}]
\label{lem:logical-injections}
For $i \in \{1,2\}:$ if $m \ol{\tleq_{P_i}} M$, then $\JJ \emph{\emph{in}}_i \kcirc m \ol{\tleq_{P_1 + P_2}} \mathtt{in}_i M .$
\end{lemma}

\begin{lemma}[{\cite[Lemma 111]{m-monad}}]
\label{lem:logica-case}
Let $m \ol{\tleq_{P_1 + P_2}} M$. Next, assume that for $k \in \{1,2\}$ we have terms $x_k : P_k \vdash N_k : R$ and morphisms $n_k \colon \lrb{P_k} \kto \lrb{R}$, such that for every $v_k \tleq_{P_k} V_k,$ it is the case that $n_k \kcirc v_k \ol{\tleq_R} N_k[V_k / x_k].$
Then 
\[ [n_1, n_2] \kcirc m \ol{\tleq_R} \mathtt{case}\ M\ \mathtt{of}\ \mathtt{in}_1 x_1 \Rightarrow N_1\ |\ \mathtt{in}_2 x_2 \Rightarrow N_2 . \]
\end{lemma}

\begin{lemma}[{\cite[Lemma 112]{m-monad}}]
\label{lem:logical-pairs}
If $m_1 \ol{\tleq_{P_1}} M_1$ and $m_2 \ol{\tleq_{P_2}} M_2$ then $\llangle m_1, m_2 \rrangle \ol{\tleq_{P_1 \times P_2}} (M_1, M_2).$
\end{lemma}

\begin{lemma}[{\cite[Lemma 113]{m-monad}}]
\label{lem:logical-projections}
If $m \ol{\tleq_{P_1 \times P_2}} M$ then $\JJ \pi_i \kcirc m \ol{\tleq_{P_i}} \pi_i M$, for $i \in \{1,2\}.$
\end{lemma}

\begin{lemma}[{\cite[Lemma 114]{m-monad}}]
\label{lem:logical-unfold}
If $m \ol{\tleq_{\mu X. P}} M$ then $\sunfold{} \kcirc m \ol{\tleq_{ P[\mu X. P/X] }} \mathtt{unfold}\ M.$
\end{lemma}

\begin{lemma}[{\cite[Lemma 115]{m-monad}}]
\label{lem:logical-fold}
If $m \ol{\tleq_{P[\mu X. P / X]}} M$ then $\sfold{} \kcirc m \ol{\tleq_{\mu X. P}} \mathtt{fold}\ M.$
\end{lemma}

\begin{lemma}[{\cite[Lemma 116]{m-monad}}]
\label{lem:logical-aplication}
If $m \ol{\tleq_{P \to R}} M$ and $n \ol{\tleq_P} N,$ then $m[n] \ol{\tleq_R} MN.$
\end{lemma}

It is also helpful to state some closure lemmas for the quantum logical relations. However, instead of stating this for the logical relations on configurations, it is more convenient to extend those relations to \emph{quantum terms} and establish the closure properties for them.

\begin{notation}
Given a quantum term $ \Phi; \xx_1 : \AAA_1, \ldots, \xx_n : \AAA_n \vdash \qq : \BBB$ and a value configuration $\Phi \vdash \mathcal V : \AAA_1 \otimes \cdots \otimes \AAA_n ; \qbit^k$
where $\mathcal V = [\ket \psi, \ell, \vq_1 \otimes \cdots \otimes \vq_n]$, we shall write $[ \qq \qsubst \mathcal V ] \defeq [\ket \psi, \ell, \qq[  \vq_1, \ldots, \vq_n / \xx_1, \ldots, \xx_n  ]]$ for the configuration obtained by performing the indicated substitution.
  Then $\Phi \vdash [\qq \qsubst \mathcal V] : \BBB; \qbit^k.$
\end{notation}

\begin{definition}
  For each closed quantum type $\BBB$ and quantum context $\BGamma = \xx_1 : \AAA_1, \ldots, \xx_n : \AAA_n$, we define a \emph{logical relation}
  \begin{align*}
    \btleq_{\BGamma \vdash \BBB} &\subseteq \QQ(\lrb{\BGamma}, \lrb{\BBB}) \times \{ \qq \ |\ \cdot; \BGamma \vdash \qq : \BBB \} \qquad \text{by}\\
    q \btleq_{\BGamma \vdash \BBB} & \qq \text{ iff } \forall k \in \mathbb N. \forall \left( \cdot \vdash \mathcal V : \AAA_1 \otimes \cdots \otimes \AAA_n ; \qbit^k \right). (q \otimes \id_{\qbit^k}) \circ \lrb{\mathcal V} \btleq^k_{\BBB} [\qq \qsubst \mathcal V],
  \end{align*}
  where the second quantifier ranges over well-formed \emph{value confiugrations} of the indicated type. 
\end{definition}

\begin{lemma}
  For any quantum value $\cdot; \BGamma \vdash \vq : \BBB$, we have that $\lrb{\vq} \btleq_{\BGamma \vdash \BBB} \vq.$
\end{lemma}
\begin{proof}
  Let $\BGamma = \xx_1 : \AAA_1, \ldots, \xx_n : \AAA_n$ and let $k \in \mathbb N$ and $\cdot \vdash \mathcal V : \AAA_1 \otimes \cdots \otimes \AAA_n; \qbit^k$ be arbitrary.
  By the Substitution Lemma \ref{lem:term-substitution}, it follows that $\lrb{ [\vq \qsubst \mathcal V]} = (\lrb \vq \otimes \id_{\qbit^k}) \circ \lrb{\mathcal V}.$
  Then the proof follows by Lemma \ref{lem:value-configuration}.
\end{proof}

Before we may prove the necessary lemmas for terms dealing with observable primitives, the following lemma is useful.

\begin{lemma}
  Let $\cdot \vdash v : O$ be an \emph{observable} value. Then $\lrb v \tleq_O v.$ Furthermore, if $f \tleq_O v$, then $f = \lrb v.$
\end{lemma}
\begin{proof}
  Observable values have a very simple structure, because they do not involve any use of function space. The lemma follows by straightforward induction on the derivation of $v$ using Theorem \ref{thm:formal-relations}.
\end{proof}

\begin{lemma}
\label{lem:quantum-logical-convex-sum}
  Let $\cdot \vdash \mathcal C : \AAA; \qbit^k$ be a configuration and let $F$ be some finite index set. Assume that we are given morphisms $c_i \in \QQ(\mathbb C, \lrb \AAA \otimes \lrb{\qbit^k})$ and configurations $\cdot \vdash \mathcal C_i : \AAA;\qbit^k$ such that $c_i {\btleq_\AAA^k} \mathcal C_i$ for $i \in F$.
Assume further that for each $i \in F$, we are given a reduction path $\pi_i \in \Paths(\mathcal C, \mathcal C_i)$, such that all paths $\pi_i$ are distinct.
Then
\[ \sum_{i \in F} P(\pi_i) c_i {\btleq_\AAA^k} \mathcal C. \]
\end{lemma}
\begin{proof}
  Fully analogous to Lemma \ref{lem:logical-convex-sum}. 
\end{proof}

\begin{lemma}
  If $q_1 \btleq_{\BGamma_1 \vdash \BBB_1} \qq_1$ and $q_2 \btleq_{\BGamma_2 \vdash \BBB_2} \qq_2$, then $q_1 \otimes q_2 \btleq_{\BGamma_1, \BGamma_2 \vdash \BBB_1 \otimes \BBB_2} \qq_1 \otimes \qq_2$.
\end{lemma}

\begin{lemma}
  If $q \btleq_{\BGamma_1 \vdash \AAA_1 \otimes \AAA_2} \qq$ and $r \btleq_{\BGamma_2, \xx : \AAA_1, \yy : \AAA_2 \vdash \BBB} \rr$, then
  \[ r \circ (\id \otimes q) \circ \mathrm{swap} \btleq_{ \BGamma_1, \BGamma_2 \vdash \BBB }  \mathbf{let\ x \otimes y = q\ in\ r} . \]
\end{lemma}

\begin{lemma}
  For $i \in \{1,2\},$ if $q \btleq_{\BGamma \vdash \AAA} \qq$ then $\mathbf{in}_i \circ q \btleq_{\BGamma \vdash \AAA \oplus \BBB} \mathbf{in}_i\ \qq$.
\end{lemma}

\begin{lemma}
  If $q \btleq_{\BGamma_1 \vdash \AAA_1 \oplus \AAA_2} \qq$ and $r_1 \btleq_{\BGamma_2, \xx : \AAA_1 \vdash \BBB} \rr_1$ and
  $r_2 \btleq_{\BGamma_2, \yy : \AAA_2 \vdash \BBB} \rr_2$, then
  \[ [r_1, r_2] \circ  d \circ (\id \otimes q) \circ \mathrm{swap} \btleq_{ \BGamma_1, \BGamma_2 \vdash \BBB }  \mathbf{(case\ q\ of\ in_1 x \Rightarrow r_1\ |\ in_2 y \Rightarrow r_2)} . \]
\end{lemma}

\begin{lemma}
  If $m \ol{\tleq_{Q(\AAA,\BBB)}} M$ and $q \btleq_{\BGamma \vdash \AAA} \qq,$ then $\beta(m(*)) \circ q \btleq_{\BGamma \vdash \BBB} M\qq$.
\end{lemma}

\begin{lemma}
  If $c \btleq_{\OO}^k \mathcal C$, then $(* \mapsto r_{\OO}(\mathrm{drop}_k \circ c)) \ol{\tleq_{|\OO|}} \run\ \mathcal C.$
\end{lemma}

\begin{lemma}
  If $m \ol{\tleq_{|\OO|}} M$, then $r^{-1}_{\OO}(m_*) \btleq_{\cdot \vdash \OO} \textbf{init}\ M.$
\end{lemma}

\begin{lemma}
  Let $q \btleq_{\BGamma_1 \vdash \OO} \qq$. Let $x : |\OO|; \BGamma_2 \vdash \rr : \AAA$ be a term and $r : \lrb{|\OO|} \to \QQ(\lrb{\BGamma_2}, \lrb{\AAA})$ a Scott-continuous function such that for any observable value $\cdot \vdash V : |\OO|$ we have that
  $ r_{\lrb{V}} \btleq_{\BGamma_2 \vdash \AAA} \rr[V/x] . $ Then $\widehat r_* \circ (q \otimes \id) \btleq_{\BGamma_1, \BGamma_2 \vdash \AAA} \mathbf{let}\  x = \text{lift}\ \qq\ \mathbf{ in\ r} . $
\end{lemma}

\begin{lemma}
  If $q \btleq_{\BGamma \vdash \mathbf{P[\mu X. P / X]}} \qq$ then $\mathrm{fold} \circ q \btleq_{\BGamma \vdash \mathbf{\mu X. P}} \mathbf{fold}\ \qq$.
\end{lemma}

\begin{lemma}
  If $q \btleq_{\BGamma \vdash \mathbf{\mu X. P} } \qq$ then $\mathrm{unfold} \circ q \btleq_{\BGamma \vdash \mathbf{P[\mu X. P / X]} } \mathbf{unfold}\ \qq$.
\end{lemma}

\begin{lemma}
  If $q \btleq_{\BGamma_1 \vdash \mathbf{I} } \qq$ and $r \btleq_{\BGamma_2 \vdash \AAA} \rr$ then $\cong \circ (q \otimes r) \btleq_{\BGamma_1, \BGamma_2 \vdash \AAA} \qq; \rr$.
\end{lemma}

\subsection{Fundamental Lemma and Strong Adequacy}

We extend the definition of the logical relations to cover \emph{all terms}, including those whose non-linear context may be non-empty.

\begin{definition}
  For any classical term $x_1 : P_1, \ldots, x_n : P_n \vdash M : R,$ and morphism $m \in \KL(\lrb{ P_1} \times \cdots \times \lrb{ P_n }, \lrb R)$ we shall write $m \ol{\tleq_{x_1 : P_1, \ldots , x_n : P_n \vdash R}} M$
  iff $\forall (v_1 \tleq_{P_1} V_1), \ldots, (v_n \tleq_{P_n} V_n)$ we have that $m \kcirc \llangle\ \vec v\ \rrangle \ol{\tleq_R} M[\vec V / \vec x].$
\end{definition}

\begin{definition}
  For any quantum term $x_1 : P_1, \ldots, x_n : P_n ; \yy_1 : \AAA_1, \ldots, \yy_m : \AAA_m \vdash \qq : \BBB$ and morphism $q \in \DCPO(\lrb{ P_1} \times \cdots \times \lrb{ P_n }, \QQ(\lrb{\AAA_1} \otimes \cdots \otimes \lrb{\AAA_m}, \lrb{\BBB}) )$
  we shall write $q \btleq_{x_1 : P_1, \ldots , x_n : P_n; \yy_1 : \AAA_1, \ldots, \yy_m : \AAA_m \vdash \BBB} \qq$
  iff $\forall (v_1 \tleq_{P_1} V_1), \ldots, (v_n \tleq_{P_n} V_n)$ we have that $q_{\vec v} {\btleq}_{  \yy_1 : \AAA_1, \ldots, \yy_m : \AAA_m  \vdash \BBB} \qq[\vec V / \vec x].$
\end{definition}

We may now prove the Fundamental Lemma which then easily implies our adequacy result.

\begin{lemma}[Fundamental]
\label{lem:fundamental}
For any classical term $\Phi \vdash M : R$ and any quantum term $\Phi; \BGamma \vdash \qq : \BBB$ we have that
  \[ \lrb M \ol{\tleq_{\Phi \vdash R}} M \qquad \text{ and } \qquad \lrb{\qq} {\btleq}_{\Phi; \BGamma \vdash \BBB} \qq \]
\end{lemma}
\begin{proof}
By induction on the derivation of the term.


\noindent\textbf{Classical Lambda Abstractions.}
  The case for classical lambda abstractions follows using exactly the same arguments as \cite[Lemma 117]{m-monad}.

\noindent\textbf{Quantum Lambda Abstractions.}
The case for quantum lambda abstractions follows using similar arguments which we now present.

  Let us assume that the term of the induction hypothesis is
\[ x_1 : P_1, \ldots, x_n : P_n; \yy_1 : \AAA_1, \ldots, \yy_m : \AAA_m  \vdash \qq : \BBB. \]
Let $(v_1 \tleq_{P_1} V_1), \ldots, (v_n \tleq_{P_n} V_n)$ be arbitrary.
  Let us write $l \eqdef \lrb{\lambdaq \vec{\yy}. \qq} \kcirc \llangle\ \vec v\ \rrangle$ and $R \eqdef \lambdaq \vec{\yy}. \qq[ \vec V / \vec x]$.
Observe that $l \in \TD$ and therefore by Theorem \ref{thm:formal-relations} (C5), we may equivalently show that
  \[ l  \tleq_{Q(\AAA_1 \otimes \cdots \otimes \AAA_m, \BBB)}  R . \]
By Theorem \ref{thm:formal-relations} (A5), this is in turn equivalent to showing that
  \[ l \leq \lrb{R} \text{ and } \forall k \in \mathbb N. \forall (\cdot \vdash [\ket{\psi}, \ell, \vq_1 \otimes \cdots \otimes \vq_m] : \AAA_1 \otimes \cdots \otimes \AAA_m; \qbit^k) .   \] 
  \[  (\beta(l(*)) \otimes \id) \circ \lrb{[\ket{\psi}, \ell, \vq_1 \otimes \cdots \otimes \vq_m]} \btleq_{\BBB}^k [\ket \psi, \ell, R( \vq_1 \otimes \cdots \otimes \vq_m )] \]
The inequality is satisfied, because
\begin{align*}
  l &= \lrb{\lambdaq \vec{\yy}. \qq} \kcirc \llangle\ \vec v\ \rrangle & & \\
    & \leq \lrb{\lambdaq \vec{\yy}. \qq} \kcirc \llangle\ \vec{\lrb V}\ \rrangle & & (\text{Theorem \ref{thm:formal-relations} (C1)} ) \\
& = \lrb R . & & (\text{Lemma \ref{lem:term-substitution}} )
\end{align*}
For the other requirement, we reason as follows
\begin{align*}
  \beta( l(*) )   &= \beta( \lrb{\lambdaq \vec{\yy}. \qq}_{\vec v} ) & & (\text{Definition})\\
  &= \beta( \JJ\lrb{\qq}_{\vec v} )  & & \\
  &= \lrb{\qq}_{\vec v}  & & \\
  & {\btleq_{\yy_1 : \AAA_1, \ldots, \yy_m : \AAA_m \vdash \BBB}} \qq[\vec V/ \vec x]  . & & \text{(Induction Hypothesis)}
\end{align*}
It now follows (by definition) that for any $k \in \mathbb N$ and any $\cdot \vdash [\ket{\psi}, \ell, \vq_1 \otimes \cdots \otimes \vq_m] : \AAA_1 \otimes \cdots \otimes \AAA_m; \qbit^k$ we have that
  \[  (\beta(l(*)) \otimes \id) \circ \lrb{[\ket{\psi}, \ell, \vq_1 \otimes \cdots \otimes \vq_m]} \btleq_{\BBB}^k [\ket \psi, \ell, \qq[\vec\vq / \vec \yy, \vec V / \vec x ] ] . \]
  Finally, observe that $[\ket \psi, \ell, R( \vq_1 \otimes \cdots \otimes \vq_m )] \probto{1} [\ket \psi, \ell, \qq[\vec\vq / \vec \yy, \vec V / \vec x ] ]$ and then by Lemma \ref{lem:quantum-logical-convex-sum} it follows that
  \[  (\beta(l(*)) \otimes \id) \circ \lrb{[\ket{\psi}, \ell, \vq_1 \otimes \cdots \otimes \vq_m]} \btleq_{\BBB}^k [\ket \psi, \ell, R( \vq_1 \otimes \cdots \otimes \vq_m )]  , \]
  as required.


The cases for terms whose formation rules do not have a premise follow by straightforward verification.
All other cases follow by induction using the relevant closure lemma from Section \ref{sub:closure-properties}.
\end{proof}

\begin{corollary}
\label{cor:logical-configuration}
  For any closed quantum configuration $\cdot \vdash \mathcal C : \AAA; \qbit^k$, we have that $\lrb{\mathcal C} \btleq_{\AAA}^k \mathcal C.$
\end{corollary}
\begin{proof}
  Let $\mathcal C = [\ket \psi, \ell, \qq]$. It follows that for $n = \dim(\ket \psi) - k$ we have $\xx_1 : \qbit, \ldots, \xx_n : \qbit \vdash \qq : \AAA$ and then by the Fundamental Lemma we know that
  \[ \lrb{\qq} \btleq_{\xx_1 : \qbit, \ldots, \xx_n : \qbit  \vdash \AAA} \qq. \]
  By definition it now follows that for $\cdot \vdash [\ket \psi, \ell, \xx_1 \otimes \cdots \otimes \xx_n ] : \qbit \otimes \cdots \otimes \qbit; \qbit^k$ we have
  \[ (\lrb{\qq} \otimes \id_{\qbit^k}) \circ \lrb{ [\ket \psi, \ell, \xx_1 \otimes \cdots \otimes \xx_n ] } \btleq_{\AAA}^k [\qq \qsubst [\ket \psi, \ell, \xx_1 \otimes \cdots \otimes \xx_n ] ] . \]
  The LHS is precisely $\lrb{\mathcal C}$ and the RHS is exactly $\mathcal C$ which completes the proof.
\end{proof}

Adequacy now follows as a consequence of the Fundamental Lemma.

\begin{theorem}[Strong Adequacy]
  Let $\cdot \vdash m : P$ be a closed classical term and $\cdot \vdash \mathcal C: \mathbf A ; \qbit^k$ a closed quantum configuration. Then:
  \begin{align*}
  \lrb m = \sum_{v \in \text{Val}(m)} P(m \probto{}_* v) \lrb{v} \qquad \qquad
  \lrb{\mathcal C} = \sum_{\mathcal V \in \mathrm{ValC}(\mathcal C)} P(\mathcal C \probto{}_* \mathcal V) \lrb{\mathcal V} .
  \end{align*}
\end{theorem}
\begin{proof}
  The equation for classical terms may be established using exactly the same arguments as \cite[Theorem 118]{m-monad}. The equation for quantum configurations also follows using the same arguments which we repeat now.
Let 
  \[ u \eqdef \sum_{\mathcal V \in \ValC(\mathcal C)} P(\mathcal C \probto{}_* \mathcal V) \lrb{\mathcal V} . \]
  From Corollary \ref{cor:soundness-inequality}, we know that $\lrb{\mathcal C} \geq u.$ It remains to show the converse inequality.
Towards this end, observe that $\mathcal S_0(\btleq_\AAA^k ; \mathcal C) \subseteq \down u$.
  To establish this, we reason as follows. Taking an arbitrary element of $\mathcal S_0(\btleq_\AAA^k ; \mathcal C)$ as in Definition \ref{def:logical-configuration}:
\begin{align*}
\sum_{\pi \in F} P(\pi) \lrb{\mathcal V_\pi} &= \sum_{\mathcal V \in \cup\{\mathcal V_\pi | \pi \in F\}} \left( \sum_{\substack{\pi \in F \\ \mathcal V_\pi = \mathcal V}} P(\pi) \right) \lrb{\mathcal V} & & \\
&\leq \sum_{\mathcal V \in \cup\{\mathcal V_\pi | \pi \in F\}} \left( \sum_{\pi \in \Paths(\mathcal C, \mathcal V) } P(\pi) \right) \lrb{\mathcal V} & & \\
&= \sum_{\mathcal V \in \cup\{\mathcal V_\pi | \pi \in F\}} P(\mathcal C \probto{}_* \mathcal V) \lrb{\mathcal V} & & \\
&\leq \sum_{\mathcal V \in \ValC(\mathcal C)} P(\mathcal C \probto{}_* \mathcal V) \lrb{\mathcal V} . & &
\end{align*}

The set $\down u $ is Scott-closed and therefore $\mathcal S(\btleq_\AAA^k ; \mathcal C) \subseteq \down u $. By Corollary \ref{cor:logical-configuration}, we know
that $\lrb{\mathcal C} {\btleq_\AAA^k} \mathcal C.$ By definition of ${\btleq_\AAA^k}$ it follows $\lrb{\mathcal C} \in \mathcal S(\btleq_\AAA^k ; \mathcal C)$ and therefore $\lrb{\mathcal C} \leq u,$ thus finishing the proof.
\end{proof}

\end{document}